\documentclass[a4paper]{article}%
\usepackage{amsfonts}
\usepackage{amsmath}
\usepackage{amssymb}
\usepackage{amstext}
\usepackage{graphicx}
\usepackage[a4paper]{geometry}
\usepackage[toc,page,title,titletoc,header]{appendix}%
\providecommand{\U}[1]{\protect\rule{.1in}{.1in}}
\newtheorem{theorem}{Theorem}[section]

\newtheorem{condition}[theorem]{Condition}

\newtheorem{definition}[theorem]{Definition}

\newtheorem{lemma}[theorem]{Lemma}

\newtheorem{proposition}[theorem]{Proposition}
\newtheorem{remark}[theorem]{Remark}

\newenvironment{proof}[1][Proof]{\noindent\textbf{#1.} }{\ \rule{0.5em}{0.5em}}
\numberwithin{equation}{section}
\ifx\pdfoutput\relax\let\pdfoutput=\undefined\fi
\newcount\msipdfoutput
\ifx\pdfoutput\undefined\else
\ifcase\pdfoutput\else
\ifx\paperwidth\undefined\else
\ifdim\paperheight=0pt\relax\else\pdfpageheight\paperheight\fi
\ifdim\paperwidth=0pt\relax\else\pdfpagewidth\paperwidth\fi
\fi\fi\fi
\begin{document}

\title{A quantum evolution problem in the regime of quantum wells in a semiclassical
island with artificial interface conditions.}
\author{Andrea Mantile\thanks{Laboratoire de Math\'{e}matiques, Universit\'{e} de
Reims - FR3399 CNRS, Moulin de la Housse BP 1039, 51687 Reims, France.}}
\date{}
\maketitle

\begin{abstract}
We introduce a modified Schr\"{o}dinger operator where the semiclassical
Laplacian is perturbed by artificial interface conditions occurring at the
boundaries of the potential's support. The corresponding dynamics is analysed
in the regime of quantum wells in a semiclassical island. Under a suitable
energy constraint for the initial states, we show that the time propagator is
stable w.r.t. the non-selfadjont perturbation, provided that this is
parametrized through infinitesimal functions of the semiclassical parameter
'$h$'.

It has been recently shown that $h$-dependent artificial interface conditions
allow a new approach to the adiabatic evolution problem for the shape
resonances in models of resonant heterostructures. Our aim is to provide with
a rigorous justification of this method.

\end{abstract}

\begin{description}
\item[AMS Subject Classification: ] 81Q12, 47A40, 58J50, 81Q20.
\end{description}

\section{\label{Section_Intro}Introduction}

Artificial interface conditions have been used in \cite{FMN2} to develop an
alternative approach to the adiabatic evolution problem for shape resonances.
This consists in replacing the usual kinetic part of a 1D Schr\"{o}dinger
operator with a modified Laplacian whose domain is a restrictions of
$H^{2}(\mathbb{R}\backslash\left\{  a,b\right\}  )$ to functions fulfilling
the $\theta$-dependent conditions%
\begin{equation}
\left\{
\begin{array}
[c]{l}%
\smallskip e^{-\frac{\theta}{2}}u(b^{+})=u(b^{-})\,,\quad e^{-\frac{3}%
{2}\theta}u^{\prime}(b^{+})=u^{\prime}(b^{-})\,,\\
\\
e^{-\frac{\theta}{2}}u(a^{-})=u(a^{+}),\quad e^{-\frac{3}{2}\theta}u^{\prime
}(a^{-})=u^{\prime}(a^{+})\,,
\end{array}
\right.  \label{BC_theta}%
\end{equation}
occurring at the boundary of the interval $\left(  a,b\right)  $ (here
$u(x^{\pm})$ denote the right and left limits of $u$ in $x$). For $\theta
\in\mathbb{C}\backslash\left\{  0\right\}  $, the corresponding operator
$\Delta_{\theta}$ describes a singularly perturbed Laplacian, with
non-selfadjoint point interactions acting in $\left\{  a,b\right\}  $, while,
for $\theta=0$, the usual selfadjoint realization of $\Delta$ on
$H^{2}(\mathbb{R})$ is recovered.

The interest in quantum models arising from $\Delta_{\theta}$ stands upon the
fact that, for $\theta=i\tau$ and $\tau>0$, the sharp exterior complex
dilation mapping $x\rightarrow e^{\theta1_{\mathbb{R}\backslash\left(
a,b\right)  }(x)}x$, transforms $-i\Delta_{\theta}$ into the accretive
operator: $\left.  -ie^{-2\theta\,1_{\mathbb{R}\backslash\left(  a,b\right)
}(x)}\Delta_{2\theta}\right.  $ (being $1_{\mathbb{R}\backslash\left(
a,b\right)  }$ the characteristic function of the exterior domain). For
potentials $\mathcal{V}$ having compact support over $\left[  a,b\right]  $,
the complex deformed Schr\"{o}dinger operators: $\mathcal{H}_{\theta}\left(
\mathcal{V},\theta\right)  =-ie^{-2\theta\,1_{\mathbb{R}\backslash\left(
a,b\right)  }(x)}\Delta_{2\theta}+\mathcal{V}$, although nonselfadjoint, are
the generators of semigroup of contractions (we refer to Lemma 3.1 in
\cite{FMN2}). Since, according to the complex dilation technique (see
\cite{AgCo}, \cite{BaCo}), the quantum resonances of%
\begin{equation}
\mathcal{H}_{\theta}\left(  \mathcal{V}\right)  =-\Delta_{\theta}%
+\mathcal{V}\,, \label{H_mod}%
\end{equation}
identify with the spectral points of $\mathcal{H}_{\theta}\left(
\mathcal{V},\theta\right)  $ in the sector of the second Riemann sheet defined
by \newline$\left\{  z\in\mathbb{C}\,,\ -2\operatorname{Im}\theta<\arg
z\leq0\right\}  $, the adiabatic evolution problem for the resonances of
$\mathcal{H}_{i\tau}\left(  \mathcal{V}\right)  $, $\tau>0$,\ rephrases as an
adiabatic problem for the corresponding eigenstates of $\mathcal{H}_{i\tau
}\left(  \mathcal{V},i\tau\right)  $ and, accounting the contractivity
property of $e^{-it\mathcal{H}_{i\tau}\left(  \mathcal{V},i\tau\right)  }$, a
'standard' adiabatic theory can be developed (e.g. in \cite{Nenciu}). This
approach has been introduced in \cite{FMN2} where an adiabatic theorem is
obtained for shape resonances in the regime of quantum wells in a
semiclassical island.

Shape resonances play a central r\^{o}le in the mathematical analysis of
semiconductor heterostructures (like tunneling diodes or possibly more complex
structures; see \cite{BNP1}, \cite{BNP2}, \cite{BNP3}, and references
therein). In particular, the adiabatic approximation of the quantum evolution
appears to be a key point in the derivation of reduced models for the dynamics
of transverse quantum transport with concentrated non-linearities (e.g. in
\cite{JoPrSj}, \cite{PrSj}, \cite{PrSj1}; see also \cite{FMN3} for a
linearized case). The purpose of this work is to justify, under suitable
conditions, the use of Hamiltonians of the type $\mathcal{H}_{\theta}\left(
\mathcal{V}\right)  $ in the modelling of quantum transport through resonant heterostructures.

The relevance of the artificial interface conditions (\ref{BC_theta}) in the
description of physical systems, stands upon the fact that they are expected
to introduce small errors controlled by $\left\vert \theta\right\vert $. The
quantum dynamics generated by $\mathcal{H}_{\theta}\left(  \mathcal{V}\right)
$ has been considered in \cite{Man1}, where an accurate resolvent analysis,
and explicit formulas for the generalized eigenfunctions of the modified
operator, allow to obtain a small-$\theta$ expansion of the stationary waves
operators for couple $\left\{  \mathcal{H}_{\theta}\left(  \mathcal{V}\right)
,\mathcal{H}_{0}\left(  \mathcal{V}\right)  \right\}  $, provided that
$\mathcal{V}\in L^{2}\left(  \mathbb{R}\right)  $ is compactly supported on
$\left[  a,b\right]  $ and $1_{\left[  a,b\right]  }\mathcal{V}>0$. Then, the
quantum evolution group generated by $i\mathcal{H}_{\theta}\left(
\mathcal{V}\right)  $ is determined by conjugation from $e^{-it\mathcal{H}%
_{0}\left(  \mathcal{V}\right)  }$ and an uniform-in-time estimate for the
'distance' between the two dynamics follows from the expansion (see Theorem
1.2 in \cite{Man1})%
\begin{equation}
e^{-it\mathcal{H}_{\theta}\left(  \mathcal{V}\right)  }=e^{-it\mathcal{H}%
_{0}\left(  \mathcal{V}\right)  }+\mathcal{R}\left(  t,\theta\right)  \,.
\label{Exp_Man1}%
\end{equation}
For $\left\vert \theta\right\vert $ suitably small, the reminder
$\mathcal{R}\left(  t,\theta\right)  $ is strongly continuous w.r.t. $t$ and
$\theta$, exhibits the group property w.r.t. the time variable and is such
that
\begin{equation}
\sup_{t\in\mathbb{R\,}}\left\Vert \mathcal{R}\left(  t,\theta\right)
\right\Vert _{\mathcal{L}\left(  L^{2}(\mathbb{R})\right)  }=\mathcal{O}%
\left(  \left\vert \theta\right\vert \right)  \,. \label{Exp_Man1_rem}%
\end{equation}

In the modelling of 1D resonant heterostructures, the interaction
$\mathcal{V}$ usually describes a potential island whose profile is formed by
multiple barriers (e.g. in \cite{BNP1}). The barriers depth fixes the time
scale for the dispersion of the resonant states related to the shape
resonances; in those physical situations where this parameter can be rather
large compared to the size of the wave pockets (as in quantum tunnelling
diods), resonances have small imaginary parts and resonant states behave as
quasi-stationary. An unitarily equivalent description of the model consists in
replacing the kinetic part of the Hamiltonian with the 'semiclassical'
Laplacian, $-h^{2}\Delta$, while the potential is the superposition of a
barrier, supported on the bounded interval $\left[  a,b\right]  $, and
potential wells with support of size $h$ inside $\left(  a,b\right)  $. The
parameter $h$ now fixes the quantum scale of the system and, coherently with
the features of the physical model, is assumed to be small. The artificial
interface conditions (\ref{BC_theta}) are integrated to this framework through
a multiparameter family of operators%
\begin{equation}
\mathcal{H}_{\theta}^{h}\left(  \mathcal{V}^{h}\right)  =-h^{2}\Delta_{\theta
}+\mathcal{V}^{h}\,, \label{H_h_theta}%
\end{equation}
where $\mathcal{V}^{h}$ describes quantum wells in a semiclassical island. The
use of modified Hamiltonians of the type $\mathcal{H}_{\theta}^{h}\left(
\mathcal{V}\right)  $ in modelling realistic physical situations has to be
justified. In particular, we need to control the difference between the
modified dynamics and the unitary evolution generated by the selfadjoint
operator $\mathcal{H}_{0}^{h}\left(  \mathcal{V}^{h}\right)  $ both w.r.t.
$\left\vert \theta\right\vert $ and $h$ in a neighbourhood of the origin. Our
aim is to generalize to this framework the analysis leading to the expansion
(\ref{Exp_Man1}) in the $h$-independent case.

As it has been shown in \cite{Man1}, the key to obtain small-$\theta$
expansions of the generalized eigenfunctions, and then of the wave operators,
consists in controlling the boundary values of Green's functions as $z$
approaches the continuous spectrum. Introducing quantum wells in the model,
produces resonances with exponentially small imaginary parts as $h\rightarrow
0$. This means that the Green's functions will be exponentially large w.r.t.
$h$ somewhere in the potential structure when $z$ is close to the shape
resonances. Nevertheless, using energy estimates with exponential weights, it
is possible to show that their values on the boundary of the potential's
support exhibit only a polynomial growth in $1/h$, for $h\rightarrow0$, (an
explicit example of this mechanism can be found in \cite{FMN3}). Studying the
Green function around a resonant energy requires the introduction of a
Dirichlet problem in order to resolve the spectral singularity and to match
the complete problem with some combination of this spectral problem with the
filled wells spectral problem. Following \cite{HeSj1}, \cite{Hel}, the Grushin
technique can be used for handling this matching and obtain resolvent approximations.

In the next Section, non-mixed interface conditions occurring at the
boundaries of the potential's support are introduced. The corresponding
modified operators, forming a larger class which includes both the cases of
$\mathcal{H}_{\theta}^{h}\left(  \mathcal{V}\right)  $ and $\left(
\mathcal{H}_{\theta}^{h}\left(  \mathcal{V}\right)  \right)  ^{\ast}$, will be
considered along our work. In the Section \ref{Section_Resonances}, the
potential profile's is specified in order to describe quantum wells in a
semiclassical island. Here, we fix the spectral assumptions yielding the
existence of shape resonances in a suitable energy range (see Condition
\ref{condition_1} below), and implement the Grushin technique, with
semiclassical resolvent estimates, in order to obtain trace estimates for the
Green's functions as $h\rightarrow0$. These estimates are used in the Section
\ref{Section_evolution} where a similarity between the modified and the
selfadjoint Hamiltonians is proved in a spectral subspace. This result
requires a standard assumption of lower bounds for the imaginary parts of the
resonances (see eq. (\ref{Fermi_golden_0})), as well as an upper bound for the
interface-condition parameters, which are assumed to be polynomially small
w.r.t. $h$ as it goes to zero. Finally, the modified quantum dynamical system
is defined by conjugation and, in the Theorem \ref{Theorem_1}, we provide with
an expansion of the propagator in the limit $h\rightarrow0$. At the end of the
Section \ref{Section_evolution}, some comments about this result and the
related perspectives are proposed. Technical details concerned with the proofs
in the Subsections \ref{Section_resest}, \ref{Section_traceest} and
\ref{Section_simil} are provided with in the Appendix \ref{App_expest} and
\ref{App_jostest}.

\subsection{Notation}

We use a generalization of the Landau notation $\mathcal{O}\left(
\cdot\right)  $ and an $h$-dependent $H^{n}$ norm defined according to:

\begin{definition}
\label{Landau_Notation}Let be $X$ a metric space and $f,g:X\rightarrow
\mathbb{C}$. Then $f=\mathcal{O}\left(  g\right)  $
$\overset{def}{\Longleftrightarrow}$ $\forall\,x\in X$ it holds: $\left.
f(x)=p(x)g(x)\right.  $, being $p$ a bounded map $X\rightarrow\mathbb{C}$.
\end{definition}

\begin{definition}
\label{def_h_norms}$H^{n,h}$ is the space of $n$-times weakly differentiable
$L^{2}$-functions equipped with the norm%
\begin{equation}
\left\Vert u\right\Vert _{H^{n,h}}^{2}=\sum_{j\leq n}\left\Vert \left(
h\partial_{x}\right)  ^{j}u\right\Vert _{L^{2}}^{2} \label{h_norm}%
\end{equation}

\end{definition}

In what follows: $\mathcal{B}_{\delta}(p)$ is the open disk of radius $\delta$
centered in a point $p\in\mathbb{C}$. $\mathbb{C}^{\pm}$ are the upper and
lower complex half-plane. $1_{\Omega}(\cdot)$ is the characteristic function
of a domain $\Omega$. $d\left(  X,Y\right)  $ is the distance between the sets
$X,Y\subset\mathbb{R}$ or $\mathbb{C}$. $\partial_{j}f$, denotes the
derivative of $f$ w.r.t. the $j$-th variable. Moreover, the notation
'$\lesssim$', appearing in some of the proofs, denotes the inequality: '$\leq
C$' being $C$ a suitable positive constant.

\section{\label{Section_interface_conditions}A semiclassical island with
non-mixed interface conditions}

We start considering the family of modified Schr\"{o}dinger operators
$Q_{\theta_{1},\theta_{2}}^{h}(\mathcal{V})$, depending on the parameters
$h>0$, $\left(  \theta_{1},\theta_{2}\right)  \in\mathbb{C}^{2}$ and on a
potential $\mathcal{V}$ which is assumed to be selfadjoint and compactly
supported on the bounded interval $\left[  a,b\right]  $. In particular, we
set%
\begin{equation}
\mathcal{V}\in L^{2}(\mathbb{R},\mathbb{R})\,,\qquad\text{supp }%
\mathcal{V}=\left[  a,b\right]  \,. \label{V}%
\end{equation}
The parameters $\theta_{1}$ and $\theta_{2}$ fix the interface conditions,%
\begin{equation}
\left\{
\begin{array}
[c]{ccc}%
e^{-\frac{\theta_{1}}{2}}u(b^{+})=u(b^{-})\,, &  & e^{-\frac{\theta_{2}}{2}%
}u^{\prime}(b^{+})=u^{\prime}(b^{-})\,,\\
&  & \\
e^{-\frac{\theta_{1}}{2}}u(a^{-})=u(a^{+})\,, &  & e^{-\frac{\theta_{2}}{2}%
}u^{\prime}(a^{-})=u(a^{+})\,,
\end{array}
\right.  \label{B_C_1}%
\end{equation}
occurring at the boundary of the potential's support and $Q_{\theta_{1}%
,\theta_{2}}^{h}(\mathcal{V})$ is defined as follows%
\begin{equation}
Q_{\theta_{1},\theta_{2}}^{h}(\mathcal{V}):\left\{
\begin{array}
[c]{l}%
D\left(  Q_{\theta_{1},\theta_{2}}^{h}(\mathcal{V})\right)  =\left\{  u\in
H^{2}\left(  \mathbb{R}\backslash\left\{  a,b\right\}  \right)  \,\left\vert
\ \text{(\emph{\ref{B_C_1}}) holds}\right.  \right\}  \,,\\
\\
\left(  Q_{\theta_{1},\theta_{2}}^{h}(\mathcal{V})\,u\right)  (x)=-h^{2}%
u^{\prime\prime}(x)+\mathcal{V}(x)\,u(x)\,,\qquad x\in\mathbb{R}%
\backslash\left\{  a,b\right\}  \,.
\end{array}
\right.  \label{Q_teta}%
\end{equation}
The set $\left\{  Q_{\theta_{1},\theta_{2}}^{h}(\mathcal{V})\,,\ \left(
\theta_{1},\theta_{2}\right)  \in\mathbb{C}^{2}\right\}  $ is closed w.r.t.
the adjoint operation: a direct computation shows that%
\begin{equation}
\left(  Q_{\theta_{1},\theta_{2}}^{h}(\mathcal{V})\right)  ^{\ast}%
=Q_{-\theta_{2}^{\ast},-\theta_{1}^{\ast}}(\mathcal{V})\,. \label{Q_teta_adj}%
\end{equation}
The subset of selfadjoint operators in this class is identified by the
conditions: for $\theta_{j}=r_{j}e^{i\varphi_{j}}$, $j=1,2$,%
\begin{equation}
\left\{
\begin{array}
[c]{l}%
\varphi_{1}+\varphi_{2}=\pi+2\pi k\,,\quad k\in\mathbb{Z}\,,\\
r_{1}=r_{2}\,.
\end{array}
\right.  \label{Selafadj_cond}%
\end{equation}
When (\ref{Selafadj_cond}) are not satisfied, the corresponding operator
$Q_{\theta_{1},\theta_{2}}^{h}(\mathcal{V})$ is neither selfadjoint nor
symmetric, since in this case: $Q_{\theta_{1},\theta_{2}}^{h}(\mathcal{V}%
)\not \subset \left(  Q_{\theta_{1},\theta_{2}}^{h}(\mathcal{V})\right)
^{\ast}$. For each couple $\left\{  \theta_{1},\theta_{2}\right\}  $,
$Q_{\theta_{1},\theta_{2}}^{h}\left(  \mathcal{V}\right)  $ identifies with a
(possibly non-selfadjoint) extension of the \emph{simple} symmetric operator
$Q^{0,h}(\mathcal{V})$%
\begin{equation}
D\left(  Q^{0,h}(\mathcal{V})\right)  =\left\{  u\in H^{2}\left(
\mathbb{R}\right)  \,\left\vert \ u(\alpha)=u^{\prime}(\alpha)=0\,,\ \alpha
=a,b\right.  \right\}  \,,
\end{equation}
whose action is defined as in (\ref{Q_teta}) (recall that a symmetric operator
is simple if there is no non-trivial subspace \ which reduces it to a
selfadjoint operator). In particular, $Q_{\theta_{1},\theta_{2}}^{h}\left(
\mathcal{V}\right)  $ defines an explicitly solvable model w.r.t. the
selfadjoint Hamiltonian $Q_{0,0}^{h}\left(  \mathcal{V}\right)  $.

\subsection{\label{Section_Resolvent_1}Boundary triples and Krein-like
resolvent formulas.}

Following \cite{Man1}, we next use the \emph{boundary triples} approach to
obtain explicit resolvent and generalized eigenfunction formulas for
$Q_{\theta_{1},\theta_{2}}^{h}\left(  \mathcal{V}\right)  $.

Point perturbation models, as $Q_{\theta_{1},\theta_{2}}^{h}(\mathcal{V})$,
can be described as restrictions of a \emph{larger} operator through linear
relations on an Hilbert space. Let introduce $Q^{h}(\mathcal{V})$%
\begin{equation}
\left\{
\begin{array}
[c]{l}%
D(Q^{h}(\mathcal{V}))=H^{2}\left(  \mathbb{R}\backslash\left\{  a,b\right\}
\right)  \,,\\
\\
\left(  Q^{h}(\mathcal{V})\,u\right)  (x)=-h^{2}u^{\prime\prime}%
(x)+\mathcal{V}(x)\,u(x)\qquad\text{for }x\in\mathbb{R}\backslash\left\{
a,b\right\}  \,,
\end{array}
\right.  \label{Q}%
\end{equation}
with $\mathcal{V}$ defined according to (\ref{V}), and let $Q^{0,h}%
(\mathcal{V})$ be such that: $\left(  Q^{0,h}(\mathcal{V})\right)  ^{\ast
}=Q^{h}(\mathcal{V})$. Explicitly, $Q^{0,h}(\mathcal{V})$ identifies with the
symmetric restriction of $Q^{h}(\mathcal{V})$ to the domain $D\left(
Q^{0,h}(\mathcal{V})\right)  $. The related defect spaces, $\mathcal{N}%
_{z,h}=\ker(Q^{h}(\mathcal{V})-z)$, are $4$-dimensional subspaces of
$D(Q^{h}(\mathcal{V}))$ generated, for $z\in\mathbb{C}\backslash\mathbb{R}$,
by the independent solutions of the problem%
\begin{equation}
\left\{
\begin{array}
[c]{l}%
(-h^{2}\partial_{x}^{2}+\mathcal{V}-z)u(x)=0\,,\quad x\in\mathbb{R}%
\backslash\left\{  a,b\right\}  \,,\\
\\
u\in D(Q^{h}(\mathcal{V}))\,.
\end{array}
\right.  \label{defect_fun}%
\end{equation}
A \emph{boundary triple} $\left\{  \mathbb{C}^{4},\Gamma_{0},\Gamma
_{1}\right\}  $ for $Q^{h}(\mathcal{V})$ is defined with two linear boundary
maps $\Gamma_{i=1,2}:D(Q^{h}(\mathcal{V}))\rightarrow\mathbb{C}^{4}$, possibly
depending on $h$, fulfilling for any $\psi,\varphi\in D(Q^{h}(\mathcal{V}))$
the equation%
\begin{equation}
\left\langle \psi,Q^{h}(\mathcal{V})\varphi\right\rangle _{L^{2}(\mathbb{R}%
)}-\left\langle Q^{h}(\mathcal{V})\psi,\varphi\right\rangle _{L^{2}%
(\mathbb{R})}=\left\langle \Gamma_{0}\psi,\Gamma_{1}\varphi\right\rangle
_{\mathbb{C}^{4}}-\left\langle \Gamma_{1}\psi,\Gamma_{0}\varphi\right\rangle
_{\mathbb{C}^{4}}\,, \label{BVT_1}%
\end{equation}
and such that the transformation $\left(  \Gamma_{0},\Gamma_{1}\right)
:D(Q^{h}(\mathcal{V}))\rightarrow\mathbb{C}^{4}\times\mathbb{C}^{4}$ is
surjective. With the help of this structure, it is possible to parametrize the
\emph{proper extensions} of $Q^{0,h}(\mathcal{V})$, i.e. the closed operators
$Q^{h}$ such that $Q^{0,h}(\mathcal{V})\subseteq Q^{h}\subseteq Q^{h}%
(\mathcal{V})$, in terms of \emph{linear relations} on $\mathbb{C}^{4}$ (that
is the linear subspaces of $\mathbb{C}^{4}\times\mathbb{C}^{4}$). The
extension $Q_{\theta}^{h}\left(  \mathcal{V}\right)  $, corresponding to
$\theta\in\mathbb{C}^{4}\times\mathbb{C}^{4}$, is defined by the restriction
(we refer to \cite{BeMaNe}, Proposition 2.2)%
\begin{equation}
Q_{\theta}^{h}\left(  \mathcal{V}\right)  =Q^{h}(\mathcal{V})\upharpoonright
\left\{  u\in D\left(  Q^{h}(\mathcal{V})\right)  \,\left\vert \ \left(
\Gamma_{0}u,\Gamma_{1}u\right)  \in\theta\right.  \right\}  \,.
\label{parametrization_teta}%
\end{equation}
Let us denote with $\tilde{Q}^{h}(\mathcal{V})$ the particular extension
associated to the subspace $\left\{  0,\mathbb{C}^{4}\right\}  $: it plays the
r\^{o}le of a 'reference extension' of $Q^{0,h}(\mathcal{V})$. This operator
depends on the choice of the maps $\Gamma_{i=0,1}$ (which is not unique) and
coincides with the restriction of $Q^{h}(\mathcal{V})$ to the functions
fulfilling the condition: $\Gamma_{0}u=0$. According to the relation
(\ref{BVT_1}), $\tilde{Q}^{h}(\mathcal{V})$ is selfadjoint and $\mathbb{C}%
\backslash\mathbb{R}\subset\mathcal{\rho}\left(  \tilde{Q}^{h}(\mathcal{V}%
)\right)  $.

If a proper extension $Q_{\theta}^{h}\left(  \mathcal{V}\right)  $ is
\emph{disjoint} with $\tilde{Q}^{h}(\mathcal{V})$ (i.e.: $D\left(  Q_{\theta
}^{h}\left(  \mathcal{V}\right)  \right)  \cap D\left(  \tilde{Q}%
^{h}(\mathcal{V})\right)  =D\left(  Q^{0,h}(\mathcal{V})\right)  $), then
$\theta$ identifies with the graph of a linear operator on $\mathbb{C}^{4}$,
i.e. it exists $M\in\mathbb{C}^{4,4}$ s.t. $\theta=\left\{  \left(
u,Mu\right)  \,,\ u\in\mathbb{C}^{4}\right\}  $. In this case, $Q_{\theta}%
^{h}\left(  \mathcal{V}\right)  $ rephrases as%
\begin{equation}
Q_{M}^{h}(\mathcal{V})=Q^{h}(\mathcal{V})\upharpoonright\ker\left(
M\Gamma_{0}-\Gamma_{1}\right)  \,, \label{parametrization_M}%
\end{equation}
and the characterization: $\left(  Q_{M}^{h}(\mathcal{V})\right)  ^{\ast
}=Q_{M^{\ast}}^{h}(\mathcal{V})\,$, holds (see \cite{BeMaNe}, Proposition
2.2). A resolvent formula expresses the difference $\left.  \left(  Q_{M}%
^{h}(\mathcal{V})-z\right)  ^{-1}-\left(  \tilde{Q}^{h}(\mathcal{V})-z\right)
^{-1}\right.  $ in terms of a finite rank operator with range $\mathcal{N}%
_{z}$ (e.g. in \cite{Ryz1}, Theorem 1.2). Let introduce the \emph{Gamma field
}and the \emph{Weyl function} associated with the triple $\left\{
\mathbb{C}^{4},\Gamma_{0},\Gamma_{1}\right\}  $; these are the holomorphic
families of bounded operators in $\mathcal{L}\left(  \mathbb{C}^{4}%
,L^{2}\left(  \mathbb{R}\right)  \right)  $ and $\mathcal{L}\left(
\mathbb{C}^{4},\mathbb{C}^{4}\right)  $ defined by%
\begin{equation}
\gamma_{z,h}(\mathcal{V})=\left(  \left.  \Gamma_{0}\right\vert _{\mathcal{N}%
_{z,h}}\right)  ^{-1}\,,\qquad q(z,\mathcal{V},h)=\Gamma_{1}\circ\gamma
_{z,h}(\mathcal{V})\,,\qquad z\in\mathcal{\rho}\left(  \tilde{Q}%
^{h}(\mathcal{V})\right)  \,, \label{gamma_q_def}%
\end{equation}
where $\left.  \Gamma_{0}\right\vert _{\mathcal{N}_{z,h}}$ is the restriction
of $\Gamma_{0}$ to $\mathcal{N}_{z,h}$. Then the disjoint couple $\left(
Q_{M}^{h}\left(  \mathcal{V}\right)  ,\tilde{Q}^{h}(\mathcal{V})\right)  $
fulfills the relation
\begin{equation}
\left(  Q_{M}^{h}(\mathcal{V})-z\right)  ^{-1}-\left(  \tilde{Q}%
^{h}(\mathcal{V})-z\right)  ^{-1}=\gamma_{z,h}(\mathcal{V})\left(
M-q(z,\mathcal{V},h)\right)  ^{-1}\gamma_{\bar{z},h}^{\ast}(\mathcal{V)}%
\,,\quad z\in\rho\left(  Q_{M}^{h}(\mathcal{V})\right)  \cap\mathcal{\rho
}\left(  \tilde{Q}^{h}(\mathcal{V})\right)  \,. \label{krein}%
\end{equation}

In the perspective of a comparison between $Q_{\theta_{1},\theta_{2}}%
^{h}(\mathcal{V})$ and $Q_{0,0}(\mathcal{V})$ a natural choice is%
\begin{equation}%
\begin{array}
[c]{ccc}%
\Gamma_{0}^{h}u=h^{2}%
\begin{pmatrix}
u^{\prime}(b^{-})-u^{\prime}(b^{+})\smallskip\medskip\\
u(b^{+})-u(b^{-})\smallskip\medskip\\
u^{\prime}(a^{-})-u^{\prime}(a^{+})\smallskip\medskip\\
u(a^{+})-u(a^{-})
\end{pmatrix}
\,, &  & \Gamma_{1}u=\frac{1}{2}%
\begin{pmatrix}
u(b^{+})+u(b^{-})\smallskip\medskip\\
u^{\prime}(b^{+})+u^{\prime}(b^{-})\smallskip\medskip\\
u(a^{+})+u(a^{-})\smallskip\medskip\\
u^{\prime}(a^{+})+u^{\prime}(a^{-})
\end{pmatrix}
\,,
\end{array}
\, \label{BVT}%
\end{equation}
which leads to: $\tilde{Q}^{h}(\mathcal{V})=Q_{0,0}^{h}(\mathcal{V})$.
According to the definitions (\ref{Q_teta}) and (\ref{Q}), the operator
$Q_{\theta_{1},\theta_{2}}^{h}(\mathcal{V})$ identifies with the restriction
of $Q^{h}(\mathcal{V})$ parametrized by the $\mathbb{C}^{4,4}$-block-diagonal
matrices%
\begin{equation}%
\begin{array}
[c]{ccc}%
A_{\theta_{1},\theta_{2}}^{h}=\frac{1}{h^{2}}%
\begin{pmatrix}
a(\theta_{1},\theta_{2}) &  & \\
&  & \\
&  & a(-\theta_{1},-\theta_{2})
\end{pmatrix}
\,, &  & B_{\theta_{1},\theta_{2}}=%
\begin{pmatrix}
b(\theta_{1},\theta_{2}) &  & \\
&  & \\
&  & b(-\theta_{1},-\theta_{2})
\end{pmatrix}
\,,
\end{array}
\label{AB_teta1,2}%
\end{equation}
where%
\begin{equation}%
\begin{array}
[c]{ccc}%
a(\theta_{1},\theta_{2})=%
\begin{pmatrix}
1+e^{\frac{\theta_{2}}{2}} & 0\\
0 & 1+e^{\frac{\theta_{1}}{2}}%
\end{pmatrix}
\,, &  & b(\theta_{1},\theta_{2})=2%
\begin{pmatrix}
0 & 1-e^{\frac{\theta_{2}}{2}}\\
e^{\frac{\theta_{1}}{2}}-1 & 0
\end{pmatrix}
\,.
\end{array}
\label{ab_teta1,2}%
\end{equation}
Using (\ref{BVT}) and (\ref{ab_teta1,2})-(\ref{ab_teta1,2}), the linear
relations (\emph{\ref{B_C_1}}) rephrase as%
\begin{equation}
A_{\theta_{1},\theta_{2}}^{h}\Gamma_{0}^{h}u=B_{\theta_{1},\theta_{2}}%
\Gamma_{1}u\,, \label{bc_teta1_teta_2}%
\end{equation}
which leads to the equivalent definition
\begin{equation}
Q_{\theta_{1},\theta_{2}}^{h}(\mathcal{V}):\left\{
\begin{array}
[c]{l}%
\mathcal{D}\left(  Q_{\theta_{1},\theta_{2}}^{h}(\mathcal{V})\right)
=\left\{  u\in D(Q^{h}(\mathcal{V}))\,\left\vert \ A_{\theta_{1},\theta_{2}%
}^{h}\Gamma_{0}^{h}u=B_{\theta_{1},\theta_{2}}\Gamma_{1}u\right.  \right\}
\,,\\
\\
Q_{\theta_{1},\theta_{2}}^{h}(\mathcal{V})\,u=Q^{h}(\mathcal{V})\,u\,.
\end{array}
\right.  \label{Q_teta1_teta2}%
\end{equation}
Adopting this parametrization, the resolvent's formula rephrases as%
\begin{equation}
\left(  Q_{\theta_{1},\theta_{2}}^{h}(\mathcal{V})-z\right)  ^{-1}=\left(
Q_{0,0}^{h}(\mathcal{V})-z\right)  ^{-1}-\sum_{i,j=1}^{4}\left[  \left(
B_{\theta_{1},\theta_{2}}\,q(z,\mathcal{V},h)-A_{\theta_{1},\theta_{2}}%
^{h}\right)  ^{-1}B_{\theta_{1},\theta_{2}}\right]  _{ij}\left\langle
\gamma_{\bar{z},h}(e_{j},\mathcal{V)},\cdot\right\rangle _{L^{2}(\mathbb{R}%
)}\gamma_{z,h}(e_{i},\mathcal{V})\,, \label{krein_1}%
\end{equation}
where $\left\{  e_{i}\right\}  _{i=1}^{4}$ is the standard basis in
$\mathbb{C}^{4}$, while $\gamma_{z,h}(v,\mathcal{V})$ denotes the action of
$\gamma_{z,h}(\mathcal{V})$ on the vector $v$. The corresponding integral
kernel, $\mathcal{G}_{\theta_{1},\theta_{2}}^{z,h}(x,y)$, is%
\begin{equation}
\mathcal{G}_{\theta_{1},\theta_{2}}^{z,h}(x,y)=\mathcal{G}_{0,0}%
^{z,h}(x,y)-\sum_{i,j=1}^{4}\left[  \left(  B_{\theta_{1},\theta_{2}%
}\,q(z,\mathcal{V},h)-A_{\theta_{1},\theta_{2}}^{h}\right)  ^{-1}B_{\theta
_{1},\theta_{2}}\right]  _{ij}\left[  \gamma_{z,h}(e_{j},\mathcal{V})\right]
(y)\,\left[  \gamma_{z,h}(e_{i},\mathcal{V})\right]  (x)\,. \label{G_z_teta}%
\end{equation}

The Krein's-like formula (\ref{krein_1}) allows a detailed resolvent analysis
for the operators $Q_{\theta_{1},\theta_{2}}^{h}(\mathcal{V})$. The case $h=1$
has been explicitly considered in \cite{Man1}. Since $Q_{\theta_{1},\theta
_{2}}^{h}(\mathcal{V})$ and $Q_{\theta_{1},\theta_{2}}^{1}(\mathcal{V})$ are
unitarily equivalent by an $h$-dependent dilation, they have the same spectral
profile and the most of the results holding for $Q_{\theta_{1},\theta_{2}}%
^{1}(\mathcal{V})$ extend to the present framework for any fixed $h>0$. In
what follows we recall the spectral properties of $Q_{\theta_{1},\theta_{2}%
}^{h}$, and give a formula for its generalized eigenfunctions; details of the
proofs may be found in \cite{Man1}.

\begin{proposition}
\label{Proposition_spectrum}For a fixed $h>0$, the operator $Q_{\theta
_{1},\theta_{2}}^{h}(\mathcal{V})$ defined according to (\ref{V}),
(\ref{Q_teta}) is Kato-analytic w.r.t. the parameters $\theta_{j=1,2}$. For
any couple $\left(  \theta_{1},\theta_{2}\right)  \in\mathbb{C}^{2}$, the
essential part of the spectrum is $\sigma_{ess}\left(  Q_{\theta_{1}%
,\theta_{2}}^{h}(\mathcal{V})\right)  =\mathbb{R}_{+}$. If, in addition,
$\mathcal{V}$ is assumed to be defined positive, i.e.%
\begin{equation}
\left\langle u,\mathcal{V\,}u\right\rangle _{L^{2}(a,b)}>0\qquad\forall\,u\in
L^{2}(\mathbb{R})\,, \label{V_pos}%
\end{equation}
then it exists $\delta>0$, possibly depending on $h$, such that:
$\sigma\left(  Q_{\theta_{1},\theta_{2}}^{h}(\mathcal{V})\right)
=\mathbb{R}_{+}$ for all $\left(  \theta_{1},\theta_{2}\right)  $ s.t.
$\theta_{j=1,2}\in\mathcal{B}_{\delta}\left(  0\right)  $.
\end{proposition}

\begin{description}
\item[\emph{Notice}] The essential spectrum of $A$ is here defined according
to \cite{Wolf} as $\sigma_{ess}\left(  A\right)  =\mathbb{C}\backslash
\mathcal{F}\left(  A\right)  $, being $\mathcal{F}\left(  A\right)  $ the set
of complex $\lambda\in\mathbb{C}$ s.t. $\left(  A-\lambda\right)  $ is Fredholm.
\end{description}

In order to obtain explicit representations of the operators $\gamma
_{z,h}(\mathcal{V})$ and $q(z,\mathcal{V},h)$ appearing at the r.h.s. of
(\ref{krein_1}), it is necessary to define a particular basis of the defect
spaces $\mathcal{N}_{z,h}$. A possible choice is given in terms of the Green's
function of the operator $\left(  Q_{0,0}^{h}(\mathcal{V})-z\right)  $ and of
their derivatives. Let introduce the functions $\mathcal{G}^{z,h}%
(x,y,\mathcal{V})$ and $\mathcal{H}^{z,h}(x,y,\mathcal{V})$; for
$z\in\mathbb{C}\backslash\sigma\left(  Q_{0,0}^{h}(\mathcal{V})\right)  $, the
maps $x\rightarrow\mathcal{G}^{z,h}(\cdot,y,\mathcal{V})$ and $x\rightarrow
\mathcal{H}^{z,h}(\cdot,y,\mathcal{V})$ fulfill the boundary value problems%
\begin{equation}
\left\{
\begin{array}
[c]{lll}%
\left(  -h^{2}\partial_{x}^{2}+\mathcal{V-}z\right)  \mathcal{G}^{z,h}%
(\cdot,y,\mathcal{V})=0\,, &  & \text{in }\mathbb{R}/\left\{  y\right\}  \,,\\
&  & \\
\mathcal{G}^{z,h}(y^{+},y,\mathcal{V})=\mathcal{G}^{z,h}(y^{-},y,\mathcal{V}%
)\,, &  & h^{2}\left(  \partial_{1}\mathcal{G}^{z}(y^{+},y,\mathcal{V}%
)-\partial_{1}\mathcal{G}^{z}(y^{-},y,\mathcal{V})\right)  =-1\,,
\end{array}
\right.  \label{Green_eq}%
\end{equation}
and%
\begin{equation}
\left\{
\begin{array}
[c]{lll}%
\left(  -h^{2}\partial_{x}^{2}+\mathcal{V-}z\right)  \mathcal{H}^{z,h}%
(\cdot,y,\mathcal{V})=0\,, &  & \text{in }\mathbb{R}/\left\{  y\right\}  \,,\\
&  & \\
h^{2}\left(  \mathcal{H}^{z,h}(y^{+},y,\mathcal{V})-\mathcal{H}^{z,h}%
(y^{-},y,\mathcal{V})\right)  =1\,, &  & \partial_{1}\mathcal{H}^{z}%
(y^{+},y,\mathcal{V})=\partial_{1}\mathcal{H}^{z}(y^{-},y,\mathcal{V})\,,
\end{array}
\right.  \label{D_Green_eq}%
\end{equation}
while the defect space $\mathcal{N}_{z,h}$ writes as%
\begin{equation}
\mathcal{N}_{z,h}=l.c.\left\{  \mathcal{G}^{z,h}(x,b,\mathcal{V}%
)\,,\ \mathcal{H}^{z,h}(x,b,\mathcal{V})\,,\ \mathcal{G}^{z,h}(x,a,\mathcal{V}%
)\,,\ \mathcal{H}^{z,h}(x,a,\mathcal{V})\right\}  \,. \label{defect1}%
\end{equation}
For each $h>0$, the maps $z\rightarrow\mathcal{G}^{z,h}(x,y,\mathcal{V})$,
$z\rightarrow\mathcal{H}^{z,h}(x,y,\mathcal{V})$ are meromorphic in
$\mathbb{C}\backslash\mathbb{R}_{+}$ with a branch cut along the positive real
axis and poles, corresponding to the points in $\sigma_{p}\left(  Q_{0,0}%
^{h}(\mathcal{V})\right)  $, located on the negative real axis. In particular,
these functions continuously extend up to the branch cut, both in the limits:
$z\rightarrow k^{2}\pm i0$, with the only possible exception of the point
$z=0$. In the case of positive defined potentials $z\rightarrow\mathcal{G}%
^{z,h}(x,y,\mathcal{V})$\ and $z\rightarrow\mathcal{H}^{z,h}(x,y,\mathcal{V})$
extend to the whole branch cut (we refer to \cite{Man1} for this point).

The relations (\ref{krein_1}) and (\ref{G_z_teta}) can be made explicit by
using the matrix representation w.r.t. the basis $\left\{  e_{j}\right\}
_{j=1}^{4}$ and (\ref{defect1}). From the definition (\ref{gamma_q_def}), a
direct computation yields%
\begin{equation}
\gamma_{z,h}(\mathcal{V})=%
\begin{pmatrix}
1 &  &  & \\
& -1 &  & \\
&  & 1 & \\
&  &  & -1
\end{pmatrix}
\,,\quad\text{with:}\ \left\{
\begin{array}
[c]{l}%
\gamma_{z,h}(e_{1},\mathcal{V})=\mathcal{G}^{z,h}(\cdot,b,\mathcal{V}%
)\,,\quad\gamma_{z,h}(e_{2},\mathcal{V})=-\mathcal{H}^{z,h}(\cdot
,b,\mathcal{V})\,,\\
\\
\gamma_{z,h}(e_{3},\mathcal{V})=\mathcal{G}^{z,h}(\cdot,a,\mathcal{V}%
)\,,\quad\gamma_{z,h}(e_{4},\mathcal{V})=-\mathcal{H}^{z,h}(\cdot
,a,\mathcal{V})\,.
\end{array}
\right.  \label{gamma_z}%
\end{equation}
The matrix coefficients of $q(z,\mathcal{V},h)$ depend on the boundary values
of the functions $\gamma_{z,h}(e_{i},\mathcal{V})$, $i=1...4$, as
$x\rightarrow b^{\pm}$ or $x\rightarrow a^{\pm}$. According to the
representations (\ref{G_z_h})-(\ref{H_z_h}), we have: $\partial_{1}%
\mathcal{G}^{z,h}(y^{\pm},y)=\mathcal{H}^{z,h}(y^{\mp},y)$. Then, taking into
account the boundary conditions in (\ref{Green_eq})-(\ref{D_Green_eq}), a
direct computation yields%
\begin{align}
&  q(z,\mathcal{V},h)\label{q_z}\\
&  =%
\begin{pmatrix}
\mathcal{G}^{z,h}(b,b,\mathcal{V})\medskip & -\left(  \mathcal{H}^{z,h}%
(b^{-},b,\mathcal{V})+\frac{1}{2h^{2}}\right)  & \mathcal{G}^{z,h}%
(b,a,\mathcal{V}) & -\mathcal{H}^{z,h}(b,a,\mathcal{V})\\
\mathcal{H}^{z,h}(b^{-},b,\mathcal{V})+\frac{1}{2h^{2}} & -\partial
_{1}\mathcal{H}^{z,h}(b,b,\mathcal{V}) & \mathcal{H}^{z,h}(a,b,\mathcal{V}) &
-\partial_{1}\mathcal{H}^{z,h}(b,a,\mathcal{V})\\
\mathcal{G}^{z,h}(a,b,\mathcal{V})\medskip & -\mathcal{H}^{z,h}%
(a,b,\mathcal{V}) & \mathcal{G}^{z,h}(a,a,\mathcal{V}) & -\left(
\mathcal{H}^{z,h}(a^{+},a,\mathcal{V})-\frac{1}{2h^{2}}\right) \\
\mathcal{H}^{z,h}(b,a,\mathcal{V}) & -\partial_{1}\mathcal{H}^{z,h}%
(a,b,\mathcal{V}) & \mathcal{H}^{z,h}(a^{+},a,\mathcal{V})-\frac{1}{2h^{2}} &
-\partial_{1}\mathcal{H}^{z,h}(a,a,\mathcal{V})
\end{pmatrix}
\,.\nonumber
\end{align}
As it follows by rephrasing the results of \cite{Man1} in our framework, there
exists $\delta>0$ small enough, possibly depending on $h$, such that the
matrices $\left(  B_{\theta_{1},\theta_{2}}\,q(z,\mathcal{V},h)-A_{\theta
_{1},\theta_{2}}^{h}\right)  $ are invertible for all $z\in\mathbb{C}%
/\mathbb{R}_{+}$ and $\theta_{j=1,2}\in\mathcal{B}_{\delta}\left(  0\right)
$, provided that (\ref{V}) and (\ref{V_pos}) hold. Under this assumption, the
coefficients of the inverse matrix are holomorphic w.r.t. $\left(
z,\theta_{1},\theta_{2}\right)  \in\mathbb{C}/\mathbb{R}_{+}\times
\mathcal{B}_{\delta}^{2}\left(  0\right)  $ and have continuous extensions to
the branch cut both in the limits $z=k^{2}+i\varepsilon$, $\varepsilon
\rightarrow0^{\pm}$.

The generalized eigenfunctions of our model, next denoted with $\psi
_{\theta_{1},\theta_{2}}^{h}(\cdot,k,\mathcal{V})$, solve of the boundary
value problem%
\begin{equation}
\left\{
\begin{array}
[c]{l}%
\left(  -h^{2}\partial_{x}^{2}+\mathcal{V}\right)  u=k^{2}u\,,\qquad
x\in\mathbb{R}\backslash\left\{  a,b\right\}  \,,\ k\in\mathbb{R}\,,\\
\\
A_{\theta_{1},\theta_{2}}^{h}\Gamma_{0}^{h}u=B_{\theta_{1},\theta_{2}}%
\Gamma_{1}u\,,
\end{array}
\right.  \label{gen_eigenfun_eq}%
\end{equation}
and fulfill the exterior conditions%
\begin{equation}
\psi_{\theta_{1},\theta_{2}}^{h}(x,k,\mathcal{V})\left\vert
_{\substack{x<a\\k>0}}\right.  =e^{i\frac{k}{h}x}+R^{h}(k,\theta_{1}%
,\theta_{2})e^{-i\frac{k}{h}x}\,,\quad\psi_{\theta_{1},\theta_{2}}%
^{h}(x,k,\mathcal{V})\left\vert _{\substack{x>b\\k>0}}\right.  =T^{h}%
(k,\theta_{1},\theta_{2})e^{i\frac{k}{h}x}\,, \label{gen_eigenfun_ext1}%
\end{equation}%
\begin{equation}
\psi_{\theta_{1},\theta_{2}}^{h}(x,k,\mathcal{V})\left\vert
_{\substack{x<a\\k<0}}\right.  =T^{h}(k,\theta_{1},\theta_{2})e^{i\frac{k}%
{h}x}\,,\quad\psi_{\theta_{1},\theta_{2}}^{h}(x,k,\mathcal{V})\left\vert
_{\substack{x>b\\k<0}}\right.  =e^{i\frac{k}{h}x}+R^{h}(k,\theta_{1}%
,\theta_{2})e^{-i\frac{k}{h}x}\,, \label{gen_eigenfun_ext2}%
\end{equation}
describing an incoming wave function of momentum $k$ with reflection and
transmission coefficients.$R^{h}$ and $T^{h}$. In the case $\left(  \theta
_{1},\theta_{2}\right)  =\left(  0,0\right)  $, $\psi_{0,0}^{h}(\cdot
,k,\mathcal{V})$ are the corresponding generalized eigenfunction of the
selfadjoint model $Q_{0,0}^{h}\left(  \mathcal{V}\right)  $. In the case of
defined positive potentials, an approach similar to the one leading to the
Krein-like resolvent formula (\ref{krein_1}) allows to obtain an expansion for
the difference: $\left.  \psi_{\theta_{1},\theta_{2}}^{h}(x,k,\mathcal{V}%
)-\psi_{0,0}^{h}(x,k,\mathcal{V})\right.  $ for $\left.  \left(  \theta
_{1},\theta_{2}\right)  \rightarrow\left(  0,0\right)  \right.  $. To this
aim, we need an explicit expression of the finite rank terms, appearing at the
r.h.s. of (\ref{krein_1}), in the limits where $z$ approaches the branch cut.
Let $G^{k,h}$ and $H^{k,h}$ be defined by
\begin{equation}%
\begin{array}
[c]{ccc}%
G^{\pm\left\vert k\right\vert ,h}\left(  \cdot,y,\mathcal{V}^{h}\right)
=\lim_{z\rightarrow k^{2}\pm i0}\mathcal{G}^{z,h}\left(  \cdot,y,\mathcal{V}%
^{h}\right)  \,, &  & H^{\pm\left\vert k\right\vert ,h}\left(  \cdot
,y,\mathcal{V}^{h}\right)  =\lim_{z\rightarrow k^{2}\pm i0}\mathcal{H}%
^{z,h}\left(  \cdot,y,\mathcal{V}^{h}\right)  \,,
\end{array}
\label{GH_k}%
\end{equation}
and denote with $g_{k,h}$, $\mathcal{M}^{h}$ the corresponding limits of
$\gamma_{z,h}$ and $\left(  B_{\theta_{1},\theta_{2}}\,q(z,\mathcal{V}%
,h)-A_{\theta_{1},\theta_{2}}^{h}\right)  $, i.e. we set: $g_{\pm\left\vert
k\right\vert ,h}(e_{j},\mathcal{V})=\lim_{z\rightarrow k^{2}\pm i0}%
\gamma_{z,h}(e_{i},\mathcal{V})$, and%
\begin{equation}
\mathcal{M}^{h}\left(  \pm\left\vert k\right\vert ,\theta_{1},\theta
_{2},\mathcal{V}\right)  =\lim_{z\rightarrow k^{2}\pm i0}\left(  B_{\theta
_{1},\theta_{2}}\,q(z,\mathcal{V},h)-A_{\theta_{1},\theta_{2}}^{h}\right)  \,.
\label{krein_coeff_1}%
\end{equation}
According to the previous remarks, $g_{k,h}(e_{j},\mathcal{V})$ and
$\mathcal{M}^{h}\left(  k,\theta_{1},\theta_{2},\mathcal{V}\right)  $ are well
defined and continuous w.r.t. $k\in\mathbb{R}$, with the only possible
exception of the origin, for potentials defined as in (\ref{V}). Next, we
denote with $S^{h}\left(  \theta_{1},\theta_{2},\mathcal{V}\right)  $ the set
of the singular points%
\begin{equation}
S^{h}\left(  \theta_{1},\theta_{2},\mathcal{V}\right)  =\left\{
k\in\mathbb{R\,}\left\vert \ \det\mathcal{M}^{h}\left(  k,\theta_{1}%
,\theta_{2},\mathcal{V}\right)  \neq0\right.  \right\}  \,.
\label{Singular_Set_h_theta}%
\end{equation}

\begin{proposition}
\label{Proposition_gen_eigenfun}Let $\mathcal{V}$ fulfill (\ref{V}) and
$k\in\mathbb{R}^{\ast}\backslash S^{h}\left(  \theta_{1},\theta_{2}%
,\mathcal{V}\right)  $. The relation%
\begin{align}
&  \psi_{\theta_{1},\theta_{2}}^{h}(\cdot,k,\mathcal{V}%
)\label{gen_eigenfun_Krein_h}\\
&  =\left\{
\begin{array}
[c]{lll}%
\psi_{0,0}^{h}(\cdot,k,\mathcal{V})-\sum_{i,j=1}^{4}\left[  \left(
\mathcal{M}^{h}\left(  k,\theta_{1},\theta_{2},\mathcal{V}\right)  \right)
^{-1}B_{\theta_{1},\theta_{2}}\right]  _{ij}\left[  \Gamma_{1}\psi_{0,0}%
^{h}(\cdot,k,\mathcal{V})\right]  _{j}\,g_{k,h}(e_{i},\mathcal{V})\,, &  &
\text{for }k>0\,,\\
&  & \\
\psi_{0,0}^{h}(\cdot,k,\mathcal{V})-\sum_{i,j=1}^{4}\left[  \left(
\mathcal{M}^{h}\left(  -k,\theta_{1},\theta_{2},\mathcal{V}\right)  \right)
^{-1}B_{\theta_{1},\theta_{2}}\right]  _{ij}\left[  \Gamma_{1}\psi_{0,0}%
^{h}(\cdot,k,\mathcal{V})\right]  _{j}\,g_{-k,h}(e_{i},\mathcal{V})\,, &  &
\text{for }k<0\,,
\end{array}
\right. \nonumber
\end{align}
holds for any fixed $h>0$ and $\left(  \theta_{1},\theta_{2}\right)
\in\mathbb{C}^{2}$.
\end{proposition}

\begin{proof}
Since the limits (\ref{krein_coeff_1}) and the traces $\Gamma_{1}\psi
_{0,0}^{h}(\cdot,k,\mathcal{V})$ exists in $\mathbb{R}^{\ast}$ for potentials
fulfilling (\ref{V}), the r.h.s. of (\ref{gen_eigenfun_Krein_h}) is well
defined for $k\in\mathbb{R}^{\ast}\backslash S^{h}\left(  \theta_{1}%
,\theta_{2},\mathcal{V}\right)  $. With this assumption, the proof of the
equality coincides with the one given in \cite{Man1} (Proposition 2.8).
\end{proof}

If, in addition, (\ref{V_pos}) holds, it is possible to extend the above
representation to the whole real line provided that $\theta_{j=1,2}%
\in\mathcal{B}_{\delta}\left(  0\right)  $ with $\delta>0$ small enough
depending on $h$ (see \cite{Man1}). Under this condition, the above formula
can be used to construct the stationary wave operators for the pair $\left\{
Q_{0,0}^{h},Q_{\theta_{1},\theta_{2}}^{h}\right\}  $ and allows to obtain an
expansion of the modified quantum propagator for small values of the
parameters $\theta_{j=1,2}$ at any fixed value of $h>0$. A generalization of
this approach as $h\rightarrow0$ would require a control of the coefficients
at the r.h.s. of (\ref{gen_eigenfun_Krein_h}) when both $\theta_{j=1,2}$ and
$h$ are small. Then accurate estimates for the boundary values of
$g_{k,h}(e_{j},\mathcal{V})$ (occurring in the definition of the matrix
$\mathcal{M}^{h}\left(  k,\theta_{1},\theta_{2},\mathcal{V}\right)  $) and
$\psi_{0,0}^{h}(\cdot,k,\mathcal{V})$ are needed. In the next section we
consider this problem only for a finite energy range when the potential
describes quantum wells in a semiclassical island.

\section{\label{Section_Resonances}Modelling resonant heterostructures}

We start the analysis of the parameter dependent quantities as $h\rightarrow
0$. Our aim is to obtain small-$h$ estimates for the boundary values of the
functions $G^{k,h}$, $H^{k,h}$, $\partial_{1}H^{k,h}$ and $\psi_{0,0}^{h}$
which are involved in the definition of the coefficients at the r.h.s. of
(\ref{gen_eigenfun_Krein_h}). In what follows, after introducing the
double-scale potentials related to regime of quantum wells in a semiclassical
island, we fix our spectral assumptions and recall the definition and the
asymptotic characterization of the \emph{shape resonances} (localization and
Fermi's Golden rule). Choosing a potential with quantum wells naturally leads
to a model having a finite number of resonances accumulating in a
neighbourhood of an asymptotic resonant energy $\lambda^{0}$ as $h\rightarrow
0$. In this situation, the Green's functions and the generalized
eigenfunctions corresponding to energies close to $\lambda^{0}$ are expected
to grow exponentially w.r.t. $h$ somewhere in the potential structure (see
e.g. in (\cite{FMN3})), and their traces may be large depending on the
imaginary part of the resonances. A resolvent formula, arising from applying
the Grushin method to our spectral problem, allows to obtain accurate
resolvent estimates close to the asymptotic energy. Finally, trace estimates
for the relevant functions appearing in the generalized eigenfunction formula
are provided.

\subsection{Shape resonances in the regime of quantum wells}

We consider the family of 1D Schr\"{o}dinger operators defined by%
\begin{equation}
Q_{0,0}^{h}(\mathcal{V})=-h^{2}\Delta+\mathcal{V}\,,\qquad D\left(
Q_{0,0}^{h}(\mathcal{V})\right)  =H^{2}\left(  \mathbb{R}\right)  \,.
\label{Q_0_h}%
\end{equation}
These operators depend on the small positive parameter $h$ and on the
potential $\mathcal{V}$, a real-valued $L^{\infty}$-function compactly
supported on the bounded interval $\left[  a,b\right]  $. In particular, we
focus on the case where $\mathcal{V=V}^{h}$ is assigned as the superposition
of a barrier $V$ and an $h$-dependent part $W^{h}$ such that%
\begin{equation}
\mathcal{V}^{h}=V+W^{h}\,,\quad\text{supp }V=\left[  a,b\right]
\,,\quad\text{supp }W^{h}=\left\{  x\in\left(  a,b\right)  \,,\ d\left(
x,U\right)  \leq h\right\}  \,, \label{V_h}%
\end{equation}
with $d\left(  \cdot,\cdot\right)  $ denoting the euclidean distance in
$\mathbb{R}$. Here $U$ is a strict subset of $\left(  a,b\right)  $, while the
constraints%
\begin{equation}
1_{\left[  a,b\right]  }V\geq c\,,\quad\sup\left\{  \left\Vert V\right\Vert
_{L^{\infty}\left(  \mathbb{R}\right)  },\left\Vert W^{h}\right\Vert
_{L^{\infty}\left(  \mathbb{R}\right)  }\right\}  \leq\frac{1}{c}%
\,,\quad\text{supp }W^{h}\subset\subset\left(  a,b\right)  \,, \label{V_h1}%
\end{equation}
are assumed to hold for some $c>0$, uniformly w.r.t. $h\in\left(
0,h_{0}\right]  $. Hamiltonians of the type $Q_{0,0}^{h}(\mathcal{V}^{h})$
have been adopted as models for the quantum transport through resonant
heterostructures (e.g. in \cite{BNP1}, \cite{BNP2}). In these models, the
parameter $h$ fixes the 'quantum scale' of the problem and the scaling
introduced in (\ref{V_h}) is used to describe the regime of \emph{quantum
wells} in a semiclassical island, which consists in taking $W^{h}$ as a sum of
attractive potentials supported on regions of size $h$. It is known that this
particular potential's shape prevents the accumulation of the possible
eigenvalues of the corresponding 'Dirichlet operator', $Q_{D}^{h}%
(\mathcal{V}^{h})$,%
\begin{equation}
Q_{D}^{h}(\mathcal{V}^{h})=-h^{2}\partial_{x}^{2}+\mathcal{V}^{h}\,,\qquad
D\left(  Q_{0,0}^{h}(\mathcal{V}^{h})\right)  =H^{2}\left(  \left[
a,b\right]  \right)  \cap H_{0}^{1}\left(  \left[  a,b\right]  \right)  \,,
\label{Q_D_h}%
\end{equation}
in the energy region $\left(  c,\inf_{\left[  a,b\right]  }V\right)  $ when
the limit $h\rightarrow0$ is considered (e.g. in \cite{BNP1}). Thus, using
quantum wells allows to realize models fulfilling the next spectral conditions
uniformly w.r.t. $h$.

\begin{condition}
\label{condition_1}The functions $\mathcal{V}^{h}=V+W^{h}$ are defined
according to (\ref{V_h})-(\ref{V_h1}), with $U\subset\left(  a,b\right)  $,
and the positive constants $c$ and $h_{0}$ suitably small. We assume that
there exists a real $\lambda^{0}$ and a cluster of eigenvalues $\left\{
\lambda_{j}^{h}\right\}  _{j=1}^{\ell}\subset\sigma\left(  Q_{D}%
^{h}(\mathcal{V}^{h})\right)  $ such that the conditions%
\begin{equation}%
\begin{array}
[c]{lll}%
i) &  & c\leq\lambda^{0}\leq\inf_{\left[  a,b\right]  }V-c\leq\left\Vert
V\right\Vert _{L^{\infty}\left(  \mathbb{R}\right)  }\leq\frac{1}{c}\,,\\
&  & \\
ii) &  & d\left(  \lambda^{0},\sigma\left(  Q_{D}^{h}(\mathcal{V}^{h})\right)
\backslash\left\{  \lambda_{j}^{h}\right\}  _{j=1}^{\ell}\right)  \geq c\,,\\
&  & \\
iii) &  & \max\limits_{1\leq j\leq\ell}\left\vert \lambda_{j}^{h}-\lambda
^{0}\right\vert \leq\frac{c}{h}\,.
\end{array}
\label{H1}%
\end{equation}
hold for all $h\in\left(  0,h_{0}\right]  $.
\end{condition}

Following the approach developed in (\cite{AgCo}, \cite{BaCo}), the resonances
of $Q_{0,0}^{h}(\mathcal{V}^{h})$ identify with the poles of the meromorphic
extension of the map $z\rightarrow\left(  Q_{0,0}^{h}(\mathcal{V})-z\right)
^{-1}$ into the second Riemann sheet ($\operatorname{Im}\sqrt{z}<0$), and are
detected as eigenvalues of a suitable complex deformed operator on
$L^{2}\left(  \mathbb{R}\right)  $. In our particular setting, an exterior
complex dilation, parametrized by $\tau>0$ and mapping $x\rightarrow e^{i\tau
}x$ outside the interval $\left(  a,b\right)  $, leads to an eigenvalue
problem whose solutions $\varphi_{res}$, corresponding to the resonances
$z_{res}$, exhibit the exponential modes%
\begin{equation}
\left\{
\begin{array}
[c]{lll}%
\varphi_{res}(x)=c_{+}e^{i\frac{\left(  z_{res}\right)  ^{\frac{1}{2}}%
e^{i\tau}}{h}\left(  x-b\right)  }\,, &  & x>b\,,\\
&  & \\
\varphi_{res}(x)=c_{-}e^{i\frac{\left(  z_{res}\right)  ^{\frac{1}{2}}%
e^{i\tau}}{h}\left(  a-x\right)  }\,, &  & x<a\,,
\end{array}
\right.
\end{equation}
in the exterior domain and fulfill the interior problem
\begin{equation}
1_{\left(  a,b\right)  }\left(  Q_{0,0}^{h}(\mathcal{V}^{h})-z_{res}\right)
^{-1}1_{\left(  a,b\right)  }\varphi_{res}=0 \label{resonance_eq}%
\end{equation}
Let introduce the notation: $\mathcal{P}_{z}^{h}\left(  \mathcal{V}\right)
=\left(  -h^{2}\Delta+\mathcal{V}\right)  $%
\begin{equation}
D(\mathcal{P}_{z}^{h}\left(  \mathcal{V}\right)  )=\left\{  u\in H^{2}\left(
\left(  a,b\right)  \right)  \,,\ \left[  h\partial_{x}+iz^{\frac{1}{2}%
}\right]  u(a)=0\,,\ \left[  h\partial_{x}-iz^{\frac{1}{2}}\right]
u(b)=0\right\}  \,, \label{P_z_def}%
\end{equation}
where $\left(  z\right)  ^{\frac{1}{2}}$ is determined according to $\arg
z\in\left[  -\frac{\pi}{2},\frac{3}{2}\pi\right)  $. The equation
(\ref{resonance_eq}) rephrases as a non-linear eigenvalue problem for the
operator $\mathcal{P}_{z}^{h}\left(  \mathcal{V}^{h}\right)  $
\begin{equation}
\left(  \mathcal{P}_{z_{res}}^{h}\left(  \mathcal{V}^{h}\right)
-z_{res}\right)  \varphi_{res}=0\,. \label{resonance_eq1}%
\end{equation}
Under the conditions (\ref{H1}), some of the solutions of this problem can be
localized in small neighbourhoods of the corresponding Dirichlet's eigenvalues
$\lambda_{j}^{h}$ as $h\rightarrow0$. These are usually referred to as
\emph{shape resonances} and their semiclassical beahviour for Schr\"{o}dinger
operators of the type (\ref{Q_0_h}) have been investigated in (\cite{HeSj1},
\cite{Hel}) (see also \cite{FMN2} where the case of modified Hamiltonians with
interface conditions in $\left\{  a,b\right\}  $\ is considered). The next
proposition resumes known results (e.g. in \cite{HeSj1}); the proof follows by
using semiclassical exponential estimates for $\left(  \mathcal{P}_{z}%
^{h}\left(  V\right)  -z\right)  ^{-1}$ (i.e.: the resolvent referring to the
model with 'filled wells') and the Grushin method to compare the nonlinear
eigenvalues problem for $\mathcal{P}_{z}^{h}\left(  \mathcal{V}^{h}\right)  $
with the corresponding problem in the Dirichlet case. We refer to \cite{FMN2}
where detailed computations are given for a wider class of operators (of the
type $Q_{\theta_{1},\theta_{2}}^{h}(\mathcal{V}^{h})$) including $Q_{0,0}%
^{h}(\mathcal{V}^{h})$. We denote with $\omega_{\delta}$ the neighbourhood
\begin{equation}
\omega_{\delta}=\left\{  z\in\mathbb{C}\,,\ d\left(  z,\left\{  \lambda
_{j}^{h}\right\}  _{j=1}^{\ell}\right)  \leq\delta\right\}  \,,
\label{omega_delta}%
\end{equation}
and with $d_{Ag}\left(  x,y,\mathcal{V},\lambda\right)  $%
\begin{equation}
d_{Ag}\left(  x,y,\mathcal{V},\lambda\right)  =\int_{y}^{x}\sqrt{\left(
\mathcal{V}(s)-\lambda\right)  _{+}}ds\,,\qquad x\geq y\,, \label{Agmon_d}%
\end{equation}
the \emph{Agmon distance} between $x$ and $y$ related to a potential
$\mathcal{V}$ and an energy $\lambda\in\mathbb{R}_{+}$.

\begin{proposition}
\label{proposition_resonance}Let $\mathcal{V}^{h}=V+W^{h}$ be defined
according to (\ref{V_h})-(\ref{V_h1}), and assume the conditions (\ref{H1}) to
hold. Then, for all $h\in\left(  0,h_{0}\right]  $, the operator $Q_{0,0}%
^{h}(\mathcal{V}^{h})$ has exactly $\ell$ resonances in $\omega_{ch}$,
$\left\{  z_{j}^{h}\right\}  _{j=1}^{\ell}$, possibly counted with
multiciplicities. Considered as functions of $h$, after the proper labelling
w.r.t. $j$, these fulfill the relations%
\begin{equation}
z_{j}^{h}-\lambda_{j}^{h}=\mathcal{O}\left(  \frac{e^{-\frac{2S_{0}}{h}}%
}{h^{3}}\right)  \,. \label{resonance_est}%
\end{equation}
where $S_{0}=d_{Ag}\left(  \left\{  a,b\right\}  ,U,V,\lambda^{0}\right)  $ is
the Agmon distance between the asymptotic support of $W^{h}$ and the barrier's
boundary.\newline
\end{proposition}

\begin{remark}
\label{remark_resonances}According to our assumptions, $Q_{0,0}^{h}%
(\mathcal{V}^{h})$ is selfadjoint and has an absolutely continuous spectrum
coinciding with the positive real axis. This implies that the map
$z\rightarrow\left(  \mathcal{P}_{z}^{h}\left(  \mathcal{V}^{h}\right)
-z\right)  ^{-1}$ is holomorphic when $z\in\omega_{\frac{ch}{2}}$ and
$\operatorname{Im}z\geq0$ (with the square root's determination as in
(\ref{P_z_def})), and the poles $z_{j}^{h}$ have $\operatorname{Im}z_{j}%
^{h}<0$, while, due to the relation (\ref{resonance_est}), $\left\vert
\operatorname{Im}z_{j}^{h}\right\vert \lesssim h^{-3}e^{-\frac{2S_{0}}{h}}$.
\end{remark}

In order to control the operator-norm of $\left(  \mathcal{P}_{z}^{h}\left(
\mathcal{V}^{h}\right)  -z\right)  ^{-1}$ in the limit $z\rightarrow k^{2}+i0$
with $k^{2}$ close to the asymptotic resonant energy $\lambda^{0}$, it is
important to have a lower bound for the imaginary part of the resonances as
$h\rightarrow0$. Providing such a lower bound is a standard result in
semiclassical analysis (e.g. in \cite{HeSj1}); in \cite{BNP2} this problem is
analyzed by considering a 1D (possibly non-linear) Schr\"{o}dinger operator,
depending on the scaling parameter $h$ according to the rules prescribed in
(\ref{Q_0_h})-(\ref{V_h1}) with%
\begin{equation}
W^{h}=-\sum_{n=1}^{N}w_{n}\left(  \frac{x-x_{n}}{h}\right)  \,,\quad w_{n}%
\in\mathcal{C}^{0}\left(  \mathbb{R},\mathbb{R}\right)  \,,\quad\text{supp
}w_{n}=\left[  -d,d\right]  \,. \label{W_h_well}%
\end{equation}
This relation defines multiple \emph{quantum wells} supported around a
collection of points $x_{n}\in\left(  a,b\right)  $, $n=1,...,N$, for some
$d>0$. Assuming the potential to fulfill the spectral conditions (\ref{H1}),
one obtains%
\begin{equation}
\left\vert \operatorname{Im}z_{j}^{h}\right\vert \gtrsim e^{-2\frac{\left(
S_{0}+S_{U}\right)  +\eta}{h}}\,, \label{Fermi_golden_0}%
\end{equation}
where $\eta>0$ is arbitrarily small, while
\begin{equation}
S_{U}=\max_{n\neq n^{\prime}}d_{Ag}\left(  x_{n},x_{n^{\prime}},V,\lambda
^{0}\right)  \label{S_U}%
\end{equation}
measures the diameter of the union of the resonant wells (see Proposition 4.1
in \cite{BNP2}). As this relation shows, the lower bound of $\left\vert
\operatorname{Im}z_{j}^{h}\right\vert $ can be much smaller than the upper
bound given in (\ref{resonance_est}). Actually, the estimate can be improved,
under some additional spectral conditions, by using a \emph{Fermi golden rule}
for the imaginary part of the resonances (see e.g. in \cite{BNP2} and
\cite{FMN2}). The case of two \emph{isolated wells} is explicitly considered
in \cite{BNP2}. For $\eta>0$, we introduce $\tilde{S}_{U}$%
\begin{equation}
\tilde{S}_{U,\eta}=\max_{n\leq N}\sqrt{V(x_{n})+\eta-\lambda^{0}}\,,
\label{S_U_tilda}%
\end{equation}
which measures the diameter of the area containing all the wells in
(\ref{W_h_well}).

\begin{proposition}
\label{proposition_goldenrule}Let $\mathcal{V}^{h}=V+W^{h}$ verify the
Condition \ref{condition_1} and further assume $W^{h}$ to be defined as in
(\ref{W_h_well}) with the following restrictions: $\ell=N\leq2$, $w_{n}$ are
even functions and $S_{0}>8\tilde{S}_{U,\eta}$ for some $\eta>0$. Then%
\begin{equation}
\left\vert \operatorname{Im}z_{j}^{h}\right\vert \gtrsim e^{-\frac{2S_{0}}{h}%
}\,. \label{Fermi_golden}%
\end{equation}

\end{proposition}

\begin{proof}
The relation (\ref{Fermi_golden}) follows from the result of Proposition 8.3
and the Remark 8.4 in \cite{BNP2}.
\end{proof}

\subsection{\label{Section_resest}Weighted resolvent estimates around the
asymptotic resonant energy}

Next the operator $\left(  \mathcal{P}_{z}^{h}\left(  \mathcal{V}^{h}\right)
-z\right)  ^{-1}$ is considered in the quantum wells case; our aim is to
provide with an estimate for the boundary values of $\left(  \mathcal{P}%
_{z}^{h}\left(  \mathcal{V}^{h}\right)  -z\right)  ^{-1}u$, for $u\in
L^{2}\left(  \left(  a,b\right)  \right)  $ when $z$ is 'close' to
$\lambda^{0}$, in a sense specified later, and $\operatorname{Im}z^{\frac
{1}{2}}\geq0$ (recall that in the definition (\ref{P_z_def}), $\left(
z\right)  ^{\frac{1}{2}}$ is determined by $\arg z\in\left[  -\frac{\pi}%
{2},\frac{3}{2}\pi\right)  $). To study this problem we use the relation%
\begin{equation}
\left(  \mathcal{P}_{z}^{h}\left(  \mathcal{V}^{h}\right)  -z\right)
^{-1}=E\left(  z\right)  +E^{+}\left(  z\right)  \left(  E^{-+}\left(
z\right)  \right)  ^{-1}E^{-}\left(  z\right)  \,, \label{resolvent_grushin}%
\end{equation}
given in \cite{FMN2} after introducing a Grushin problem. The operators at the
r.h.s. of (\ref{resolvent_grushin}) act according to: $E\left(  z\right)
\in\mathcal{L}\left(  L^{2}\left(  \left(  a,b\right)  \right)  ,L^{2}\left(
\left(  a,b\right)  \right)  \right)  $, $E^{-+}\left(  z\right)
\in\mathbb{C}^{\ell,\ell}$, $E^{-}\left(  z\right)  \in\mathcal{L}\left(
L^{2}\left(  \left(  a,b\right)  \right)  ,\mathbb{C}^{\ell}\right)  $ and
$E^{+}\left(  z\right)  \in\mathcal{L}\left(  \mathbb{C}^{\ell},L^{2}\left(
\left(  a,b\right)  \right)  \right)  $. For $z\in\omega_{ch}$, the first
contribution, $E\left(  z\right)  $, is an holomorphic operator valued family,
while the possible poles of (\ref{resolvent_grushin}) are encoded by the
singularities of the matrix $E^{-+}\left(  z\right)  $. In particular, for
$\eta>0$ small and $K_{a,b,c}>0$ depending on the data, the next estimates
hold uniformly w.r.t. $z\in\omega_{ch}$ (for this point we refer to the
results in the sections 4 and 5 of \cite{FMN2})%
\begin{equation}
\left\Vert E\left(  z\right)  \right\Vert _{\mathcal{L}\left(  L^{2}\left(
\left(  a,b\right)  \right)  ,H^{1,h}\left(  \left[  a,b\right]  \right)
\right)  }=\mathcal{O}\left(  e^{-\frac{S_{0}-K_{a,b,c}\eta}{h}}\right)
\,,\qquad\left\Vert E^{-+}\left(  z\right)  -diag\left(  z-z_{j}^{h}\right)
\right\Vert _{\mathbb{C}^{\ell,\ell}}=\mathcal{O}\left(  \frac{e^{-\frac
{2S_{0}}{h}}}{h^{3}}\right)  \,, \label{est_grushin1}%
\end{equation}%
\begin{equation}
\left\Vert E^{+}\left(  z\right)  -\chi_{h}E_{0}^{+}\left(  z\right)
\right\Vert _{\mathcal{L}\left(  \mathbb{C}^{\ell},H^{1,h}\left(  \left[
a,b\right]  \right)  \right)  }=\mathcal{O}\left(  \frac{e^{-\frac{S_{0}}{h}}%
}{h}\right)  \,,\quad\left\Vert E^{-}\left(  z\right)  -E_{0}^{-}\left(
z\right)  \right\Vert _{\mathcal{L}\left(  L^{2}\left(  \left(  a,b\right)
\right)  ,\mathbb{C}^{\ell}\right)  }=\mathcal{O}\left(  e^{-\frac
{S_{0}-K_{a,b,c}\eta}{2h}}\right)  \,. \label{est_grushin2}%
\end{equation}
The maps $E_{0}^{\pm}\left(  z\right)  $ are explicitly defined by
\begin{equation}
E_{0}^{+}\left(  z\right)  :\mathbb{C}^{\ell}\rightarrow L^{2}\left(  \left(
a,b\right)  \right)  \,,\quad E_{0}^{+}\left(  z\right)  \left(  p\right)
=\sum_{j=1}^{\ell}p_{j}\Phi_{j}^{h}\,, \label{E_plus}%
\end{equation}%
\begin{equation}
E_{0}^{-}\left(  z\right)  :L^{2}\left(  \left(  a,b\right)  \right)
\rightarrow\mathbb{C}^{\ell}\,,\quad E_{0}^{-}\left(  z\right)  \left(
\varphi\right)  =%
\begin{pmatrix}
\left\langle \Phi_{1}^{h},\varphi\right\rangle \\
\vdots\\
\left\langle \Phi_{\ell}^{h},\varphi\right\rangle
\end{pmatrix}
\,, \label{E_minus}%
\end{equation}
where $\Phi_{j}^{h}$ is the normalized eigenfunction of $Q_{D}^{h}%
(\mathcal{V}^{h})$ related to the point $\lambda_{j}^{h}$, while $\chi_{h}$ is
a cutoff function such that%
\begin{equation}
\chi_{h}\in\mathcal{C}_{0}^{\infty}\left(  \left(  a,b\right)  \right)
\,,\quad\left\Vert \left(  h\partial_{x}\right)  ^{m}\chi_{h}\right\Vert \leq
C_{m}\,,\ m\in\mathbb{N}\,,\quad\chi_{h}(x)=1\text{ if }d(x,\left\{
a,b\right\}  )\geq h\,. \label{chi_h}%
\end{equation}
In the next Proposition we provide with weighted energy estimates for $\left(
\mathcal{P}_{z}^{h}\left(  \mathcal{V}^{h}\right)  -z\right)  ^{-1}$ when
$z\in\omega_{ch}$.

\begin{lemma}
\label{Lemma_Grushin}Let $\mathcal{P}_{z}^{h}\left(  \mathcal{V}^{h}\right)  $
be defined by (\ref{P_z_def}), with $\mathcal{V}^{h}$ fulfilling the Condition
\ref{condition_1} and denote with $z_{j}^{h}$, $j=1,...\ell$, the
corresponding cluster of shape resonances. For $z\in\omega_{ch}$, the estimate%
\begin{equation}
\left\Vert e^{\frac{\varphi}{h}}\left(  \mathcal{P}_{z}^{h}\left(
\mathcal{V}^{h}\right)  -z\right)  ^{-1}u\right\Vert _{H^{1,h}\left(  \left[
a,b\right]  \right)  }\leq C_{a,b,c}\left(  \frac{1}{h}\left(  \sup_{j\leq
\ell}\left\vert \operatorname{Im}z_{j}^{h}\right\vert ^{-1}\right)
+e^{\frac{K_{a,b,c}\eta}{h}}\right)  \left\Vert u\right\Vert _{L^{2}\left(
\left[  a,b\right]  \right)  }\,, \label{resolvent_grushin_est}%
\end{equation}
holds with $\varphi\left(  \cdot\right)  =d_{Ag}\left(  \cdot,U,\mathcal{V}%
,\lambda^{0}\right)  $, $C_{a,b,c}$ and $K_{a,b,c}$ being positive constants
possibly depending on the data, and $\eta>0$ arbitrarily small.
\end{lemma}

\begin{proof}
Using the representation (\ref{resolvent_grushin}), we have%
\begin{equation}
\left(  \mathcal{P}_{z}^{h}\left(  \mathcal{V}^{h}\right)  -z\right)
^{-1}u=E\left(  z\right)  u+E^{+}\left(  z\right)  \left(  E^{-+}\left(
z\right)  \right)  ^{-1}E^{-}\left(  z\right)  u\,. \label{resolvent_grushin1}%
\end{equation}
The first estimate in (\ref{est_grushin1}) yields
\begin{equation}
\left\Vert E\left(  z\right)  u\right\Vert _{H^{1,h}\left(  \left[
a,b\right]  \right)  }\lesssim e^{-\frac{S_{0}-K_{a,b,c}\eta}{h}}\left\Vert
u\right\Vert _{L^{2}\left(  \left(  a,b\right)  \right)  }\,.
\label{est_grushin0}%
\end{equation}
Due to the regularity of $\varphi$ (see the definition (\ref{Agmon_d})), it
results: $e^{\frac{\varphi}{h}}E\left(  z\right)  u\in H^{1,h}\left(  \left[
a,b\right]  \right)  $, for $u\in L^{2}\left(  \left(  a,b\right)  \right)  $.
Then, using the relation: $S_{0}=\sup_{\left[  a,b\right]  }\varphi$, we get%
\begin{equation}
\left\Vert e^{\frac{\varphi}{h}}E\left(  z\right)  u\right\Vert _{H^{1,h}%
\left(  \left[  a,b\right]  \right)  }\lesssim\left(  \sup_{\left[
a,b\right]  }e^{\frac{\varphi}{h}}\right)  e^{-\frac{S_{0}-K_{a,b,c}\eta}{h}%
}\left\Vert u\right\Vert _{L^{2}\left(  \left(  a,b\right)  \right)
}=e^{\frac{K_{a,b,c}\eta}{h}}\left\Vert u\right\Vert _{L^{2}\left(  \left(
a,b\right)  \right)  }\,, \label{resolvent_grushin_est1}%
\end{equation}
For the second contribution at the r.h.s. of (\ref{resolvent_grushin1}), we
exploit the small-$h$ expansion%
\[
E^{+}\left(  z\right)  =\chi_{h}E_{0}^{+}\left(  z\right)  +\mathcal{O}\left(
\frac{e^{-\frac{S_{0}}{h}}}{h}\right)  \,,\qquad E_{0}^{+}\left(  z\right)
\left(  p\right)  =\sum_{j=1}^{\ell}p_{j}\Phi_{j}^{h}\,,
\]
where $\mathcal{O}\left(  \cdot\right)  $ is intended in the $\mathcal{L}%
\left(  \mathbb{C}^{\ell},H^{1,h}\left(  \left[  a,b\right]  \right)  \right)
$-norm sense (see (\ref{est_grushin2})). The exponential decay estimates for
the Dirichlet eigenstates $\Phi_{j}^{h}$ read as (e.g. in in \cite{FMN2},
Proposition 4.1)%
\[
h^{\frac{1}{2}}\sup_{\left[  a,b\right]  }\left\vert e^{\frac{\varphi_{j}}{h}%
}\Phi_{j}^{h}\right\vert +\left\Vert he^{\frac{\varphi_{j}}{h}}\partial
_{x}\Phi_{j}^{h}\right\Vert _{L^{2}\left(  \left[  a,b\right]  \right)
}+\left\Vert e^{\frac{\varphi_{j}}{h}}\Phi_{j}^{h}\right\Vert _{L^{2}\left(
\left[  a,b\right]  \right)  }\leq\frac{C_{a,b,c}}{h}\,,
\]
with $\varphi_{j}(\cdot)=d_{Ag}\left(  \cdot,U,V,\lambda_{j}^{h}\right)  $,
being $U$ the asymptotic support of $W^{h}$. According to the assumption
(\ref{H1}) ($\left\vert \lambda_{j}^{h}-\lambda^{0}\right\vert \leq ch$), it
results: $e^{\pm\frac{\varphi_{j}}{h}}=\mathcal{O}\left(  e^{\pm\frac{\varphi
}{h}}\right)  $ and $\varphi_{j}$ can be replaced with $\varphi(\cdot
)=d_{Ag}\left(  \cdot,U,V,\lambda^{0}\right)  $ in the expressions above.
Taking into account the regularity of this function, we get%
\begin{equation}
\left\Vert e^{\frac{\varphi}{h}}\Phi_{j}^{h}\right\Vert _{H^{1,h}\left(
\left[  a,b\right]  \right)  }\leq\frac{C_{a,b,c}}{h}\,,\qquad\varphi
(\cdot)=d_{Ag}\left(  \cdot,U,V,\lambda^{0}\right)  \,,
\end{equation}
which implies%
\begin{equation}
\left\Vert e^{\frac{\varphi}{h}}E^{+}\left(  z\right)  \right\Vert
_{\mathcal{L}\left(  \mathbb{C}^{\ell},H^{1,h}\left(  \left[  a,b\right]
\right)  \right)  }\leq\left\Vert e^{\frac{\varphi}{h}}\chi_{h}E_{0}%
^{+}\left(  z\right)  u\right\Vert _{\mathcal{L}\left(  \mathbb{C}^{\ell
},H^{1,h}\left(  \left[  a,b\right]  \right)  \right)  }+\mathcal{O}\left(
\frac{e^{-\frac{S_{0}}{h}}}{h}\right)  \leq C_{a,b,c}\left(  \frac{1}%
{h}+\mathcal{O}\left(  \frac{e^{-\frac{S_{0}}{h}}}{h}\right)  \right)  \,.
\label{resolvent_grushin_est0}%
\end{equation}
From (\ref{est_grushin1})-(\ref{est_grushin2}) and
(\ref{resolvent_grushin_est0}) follows%
\begin{align}
&  \left\Vert e^{\frac{\varphi}{h}}E^{+}\left(  z\right)  \left(
E^{-+}\left(  z\right)  \right)  ^{-1}E^{-}\left(  z\right)  u\right\Vert
_{H^{1,h}\left(  \left[  a,b\right]  \right)  } \label{resolvent_grushin_est2}%
\\
& \nonumber\\
&  \leq C_{a,b,c}\left(  \frac{1}{h}+\mathcal{O}\left(  \frac{e^{-\frac{S_{0}%
}{h}}}{h}\right)  \right)  \left(  \sup_{\omega_{ch}\cap\overline
{\mathbb{C}^{+}}}\left\Vert \left(  E^{-+}\left(  z\right)  \right)
^{-1}\right\Vert +\mathcal{O}\left(  \frac{e^{-\frac{2S_{0}}{h}}}{h^{3}%
}\right)  \right)  _{\mathbb{C}^{\ell,\ell}}\left(  1+\mathcal{O}\left(
e^{-\frac{S_{0}-K_{a,b,c}\eta}{2h}}\right)  \right)  \left\Vert u\right\Vert
_{L^{2}\left(  \left[  a,b\right]  \right)  }\,,\nonumber
\end{align}
where $\left\Vert \cdot\right\Vert _{\mathbb{C}^{\ell,\ell}}$ is a matrix norm
and the inequality: $\left\Vert E_{0}^{-}\left(  z\right)  u\right\Vert
_{\mathbb{C}^{\ell}}\lesssim\left\Vert u\right\Vert _{L^{2}\left(  \left(
a,b\right)  \right)  }$ have been used. As a consequence of the Proposition
\ref{proposition_resonance} and the Remark \ref{remark_resonances}, the
singular points of $E^{-+}\left(  z\right)  $ are embedded in $\omega_{ch}%
\cap\mathbb{C}^{-}$. Hence we have%
\begin{equation}
\sup_{\omega_{ch}\cap\overline{\mathbb{C}^{+}}}\left\Vert \left(
E^{-+}\left(  z\right)  \right)  ^{-1}\right\Vert _{\mathbb{C}^{\ell,\ell}%
}\lesssim\sup_{j\leq\ell}\left\vert \operatorname{Im}z_{j}^{h}\right\vert
^{-1}\,. \label{resolvent_grushin_est3}%
\end{equation}
The relation (\ref{resolvent_grushin_est}), is a consequence of
(\ref{resolvent_grushin_est1}), (\ref{resolvent_grushin_est2}) and
(\ref{resolvent_grushin_est3}).
\end{proof}

We are interested in trace estimates of $\left(  \mathcal{P}_{z}^{h}\left(
\mathcal{V}^{h}\right)  -z\right)  ^{-1}u$ on $\left\{  a,b\right\}  $ when
$z\rightarrow k^{2}+i0$ and $k^{2}$ is 'close' to $\lambda^{0}$ in the
following sense: we assume $\operatorname{Re}z\in\left[  \Lambda_{1}%
,\Lambda_{2}\right]  $ such that%
\begin{equation}
c\leq\Lambda_{1}<\Lambda_{2}\leq\inf_{\left[  a,b\right]  }V-c\,,\quad
\sigma\left(  Q_{D}^{h}(\mathcal{V}^{h})\right)  \cap\left[  \Lambda
_{1},\Lambda_{2}\right]  =\left\{  \lambda_{j}^{h}\right\}  _{j=1}^{\ell}\,,
\label{Lambda_0_around}%
\end{equation}
uniformly w.r.t. $h\in\left(  0,h_{0}\right]  $. In the case $z\in\omega_{ch}$
and $\operatorname{Re}z\in\left[  \Lambda_{1},\Lambda_{2}\right]  $, trace
estimates of $\left(  \mathcal{P}_{z}^{h}\left(  \mathcal{V}^{h}\right)
-z\right)  ^{-1}u$ easily follows from the result of the Lemma
\ref{Lemma_Grushin} by using the Gagliardo-Nirenberg inequality:
$\sup_{\left[  a,b\right]  }\left\vert \varphi\right\vert \leq C_{b-a}%
\left\Vert \varphi^{\prime}\right\Vert _{L^{2}\left(  \left(  a,b\right)
\right)  }^{\frac{1}{2}}\left\Vert \varphi\right\Vert _{L^{2}\left(  \left(
a,b\right)  \right)  }^{\frac{1}{2}}$ and the equivalence of $\left\Vert
\varphi\right\Vert _{H^{1,h}\left(  \left[  a,b\right]  \right)  }$ with
$\left\Vert h\varphi^{\prime}\right\Vert _{L^{2}\left(  \left(  a,b\right)
\right)  }+\left\Vert \varphi\right\Vert _{L^{2}\left(  \left(  a,b\right)
\right)  }^{\frac{1}{2}}$ leading to%
\begin{equation}
h^{\frac{1}{2}}\sup_{\left[  a,b\right]  }\left\vert u\right\vert \leq
\tilde{C}_{b-a}\left\Vert u\right\Vert _{H^{1,h}\left(  \left[  a,b\right]
\right)  }\,. \label{h_GN}%
\end{equation}
When $z\notin$ $\omega_{ch}$, this problem can be analyzed by using trace
estimates of the corresponding resolvent in the 'filled wells' case (i.e.
$W^{h}=0$). At this concern, let recall the usual energy estimate for $\left(
\mathcal{P}_{z}^{h}\left(  \mathcal{V}\right)  -z\right)  ^{-1}$.

\begin{lemma}
Assume $\mathcal{V}\in L^{\infty}\left(  \mathbb{R},\mathbb{R}\right)  $ and
$\zeta\in\overline{\mathbb{C}^{+}}\cap\left\{  \operatorname{Im}\zeta^{2}%
\geq0\right\}  $ such that: $\mathcal{V}-\operatorname{Re}\zeta^{2}\geq c>0$.
The estimate%
\begin{equation}
\left\Vert \left(  \mathcal{P}_{\zeta^{2}}^{h}\left(  \mathcal{V}\right)
-\zeta^{2}\right)  ^{-1}f\right\Vert _{H^{1,h}\left(  \left[  a,b\right]
\right)  }\leq C_{a,b,c}\left\Vert f\right\Vert _{L^{2}\left(  \left(
a,b\right)  \right)  } \label{P_z_est}%
\end{equation}
holds with $C_{a,b,c}>0$ possibly depending on the data.
\end{lemma}

\begin{proof}
We notice at first that, under the condition $\operatorname{Im}\zeta^{2}\geq
0$, the identity: $\left(  \zeta^{2}\right)  ^{\frac{1}{2}}=\zeta$ has to be
used in the definition of $\mathcal{P}_{\zeta^{2}}^{h}$ (see (\ref{P_z_def})).
For $f\in L^{2}\left(  \left(  a,b\right)  \right)  $, we consider the
equation: $\left(  \mathcal{P}_{\zeta^{2}}^{h}\left(  \mathcal{V}\right)
-\zeta^{2}\right)  u=f$. Due to the boundary conditions in (\ref{P_z_def}),
the solution $u$ fulfills the integral relation%
\[
\left\Vert hu^{\prime}\right\Vert _{L^{2}\left(  \left[  a,b\right]  \right)
}^{2}+\int_{a}^{b}\left(  \mathcal{V}-\zeta^{2}\right)  \left\vert
u\right\vert ^{2}\,dx-ih\zeta\left(  \left\vert u(b)\right\vert ^{2}%
+\left\vert u(a)\right\vert ^{2}\right)  =\left\langle u,f\right\rangle
_{L^{2}\left(  \left[  a,b\right]  \right)  }\,.
\]
and, taking the real part of it, we obtain the identity%
\[
\left\Vert hu^{\prime}\right\Vert _{L^{2}\left(  \left[  a,b\right]  \right)
}^{2}+\int_{a}^{b}\left(  \mathcal{V}-\operatorname{Re}\zeta^{2}\right)
\left\vert u\right\vert ^{2}\,dx+h\operatorname{Im}\zeta\left(  \left\vert
u(b)\right\vert ^{2}+\left\vert u(a)\right\vert ^{2}\right)  =\left\langle
u,f\right\rangle _{L^{2}\left(  \left[  a,b\right]  \right)  }\,.
\]
Hence, according to our assumptions and to the definition (\ref{def_h_norms}),
the relation (\ref{P_z_est}) follows.
\end{proof}

Next the boundary values of the functions $\left(  \mathcal{P}_{z}^{h}\left(
\mathcal{V}^{h}\right)  -z\right)  ^{-1}u$ are considered when $z$ belongs to
the half disk $\mathcal{B}_{\frac{\Lambda_{2}-\Lambda_{1}}{2}}(\left(
\Lambda_{1}+\Lambda_{2}\right)  /2)\cap\overline{\mathbb{C}^{+}}$ and $\left[
\Lambda_{1},\Lambda_{2}\right]  $ has the property (\ref{Lambda_0_around}). In
these conditions, the traces of $\left(  \mathcal{P}_{z}^{h}\left(
\mathcal{V}^{h}\right)  -z\right)  ^{-1}\varphi$, $\varphi\in L^{2}\left(
\left(  a,b\right)  \right)  $ are controlled on a scale fixed by $\left\vert
\operatorname{Im}z_{j}^{h}\right\vert ^{-1}$.

\begin{proposition}
\label{proposition_trace-est}Let $\mathcal{P}_{z}^{h}\left(  \mathcal{V}%
^{h}\right)  $ be defined by (\ref{P_z_def}), with $\mathcal{V}^{h}$
fulfilling the Condition \ref{condition_1} and $z_{j}^{h}$, $j=1,...\ell$,
denoting the corresponding cluster of shape resonances. The estimate%
\begin{equation}
\sup_{y=a,b}\left\vert \left(  \mathcal{P}_{z}^{h}\left(  \mathcal{V}%
^{h}\right)  -z\right)  ^{-1}u(y)\right\vert \lesssim\left(  \frac
{e^{-\frac{S_{0}}{h}}}{h^{\frac{3}{2}}}\left(  \sup_{j\leq\ell}\left\vert
\operatorname{Im}z_{j}^{h}\right\vert ^{-1}\right)  +\frac{1}{h^{\frac{1}{2}}%
}\right)  \left\Vert u\right\Vert _{L^{2}\left(  \left(  a,b\right)  \right)
}\,, \label{trace_est_grushin}%
\end{equation}
hold for any $z$ in the half disk $\mathcal{B}_{\frac{\Lambda_{2}-\Lambda_{1}%
}{2}}(\left(  \Lambda_{1}+\Lambda_{2}\right)  /2)\cap\overline{\mathbb{C}^{+}%
}$, with $\left[  \Lambda_{1},\Lambda_{2}\right]  $ having the property
(\ref{Lambda_0_around}).
\end{proposition}

\begin{proof}
Let consider at first the case $z\in\omega_{ch}\cap\left\{  \mathcal{B}%
_{\frac{\Lambda_{2}-\Lambda_{1}}{2}}(\left(  \Lambda_{1}+\Lambda_{2}\right)
/2)\cap\overline{\mathbb{C}^{+}}\right\}  $. From the result of Lemma
\ref{Lemma_Grushin} and the inequality (\ref{h_GN}) we have%
\[
\sup_{y=a,b}\left\vert e^{\frac{\varphi}{h}}\left(  \mathcal{P}_{z}^{h}\left(
\mathcal{V}^{h}\right)  -z\right)  ^{-1}u\right\vert \leq C_{a,b,c}\left(
\frac{1}{h^{\frac{3}{2}}}\left(  \sup_{j\leq\ell}\left\vert \operatorname{Im}%
z_{j}^{h}\right\vert ^{-1}\right)  +\frac{e^{\frac{K_{a,b,c}\eta}{h}}%
}{h^{\frac{1}{2}}}\right)  \left\Vert u\right\Vert _{L^{2}\left(  \left[
a,b\right]  \right)  }\,,
\]
where $\varphi\left(  \cdot\right)  =d_{Ag}\left(  \cdot,U,\mathcal{V}%
,\lambda^{0}\right)  $, and $\eta>0$ is arbitrarily small. Since
$\sup_{\left\{  a,b\right\}  }\varphi=S_{0}$, it follows%
\[
\sup_{y=a,b}\left\vert \left(  \mathcal{P}_{z}^{h}\left(  \mathcal{V}%
^{h}\right)  -z\right)  ^{-1}u\right\vert \leq C_{a,b,c}\left(  \frac
{e^{-\frac{S_{0}}{h}}}{h^{\frac{3}{2}}}\left(  \sup_{j\leq\ell}\left\vert
\operatorname{Im}z_{j}^{h}\right\vert ^{-1}\right)  +\frac{e^{-\frac
{S_{0}-K_{a,b,c}\eta}{h}}}{h^{\frac{1}{2}}}\right)  \left\Vert u\right\Vert
_{L^{2}\left(  \left[  a,b\right]  \right)  }\,.
\]
This yields the relation (\ref{trace_est_grushin}) in the case $z\in
\omega_{ch}\cap\left\{  \mathcal{B}_{\frac{\Lambda_{2}-\Lambda_{1}}{2}%
}(\left(  \Lambda_{1}+\Lambda_{2}\right)  /2)\cap\overline{\mathbb{C}^{+}%
}\right\}  $.

For $z\in\left\{  \mathcal{B}_{\frac{\Lambda_{2}-\Lambda_{1}}{2}}\left(
\left(  \Lambda_{1}+\Lambda_{2}\right)  /2\right)  \cap\overline
{\mathbb{C}^{+}}\right\}  \backslash\omega_{ch}$, we use the representation%
\begin{equation}
\left(  \mathcal{P}_{z}^{h}\left(  \mathcal{V}^{h}\right)  -z\right)
^{-1}u=\left(  Q_{D}^{h}(\mathcal{V}^{h})-z\right)  ^{-1}\chi_{h}u+\left(
\mathcal{P}_{z}^{h}\left(  V\right)  -z\right)  ^{-1}\left(  1-\chi
_{h}\right)  u\,,
\end{equation}
where $\chi_{h}$ is the cutoff function defined in (\ref{chi_h}). According to
the definition (\ref{omega_delta}) and the condition (\ref{Lambda_0_around}),
the Dirichlet operator $Q_{D}^{h}(\mathcal{V}^{h})$ has no spectrum in
$\mathcal{B}_{\frac{\Lambda_{2}-\Lambda_{1}}{2}}\left(  \left(  \Lambda
_{1}+\Lambda_{2}\right)  /2\right)  \backslash\omega_{ch}$ and $\left(
Q_{D}^{h}(\mathcal{V}^{h})-z\right)  ^{-1}$ is well defined as a bounded
operator on $L^{2}\left(  \left(  a,b\right)  \right)  $ with values in
$D\left(  Q_{D}^{h}(\mathcal{V}^{h})\right)  $. In particular, this implies%
\[
\left.  \left(  \mathcal{P}_{z}^{h}\left(  \mathcal{V}^{h}\right)  -z\right)
^{-1}u\left(  y\right)  \right\vert _{y=a,b}=\left.  \left(  \mathcal{P}%
_{z}^{h}\left(  V\right)  -z\right)  ^{-1}\left(  1-\chi_{h}\right)  u\left(
y\right)  \right\vert _{y=a,b}\,,
\]
being $\mathcal{P}_{z}^{h}\left(  V\right)  $ corresponding to the 'filled
well' situation. The points of the domain $\left\{  \mathcal{B}_{\frac
{\Lambda_{2}-\Lambda_{1}}{2}}\left(  \left(  \Lambda_{1}+\Lambda_{2}\right)
/2\right)  \cap\overline{\mathbb{C}^{+}}\right\}  \backslash\omega_{ch}$
verify: $\inf_{\left[  a,b\right]  }V-\operatorname{Re}z\geq c>0$, and
$\operatorname{Im}z>0$. Under these conditions, the resolvent $\left(
\mathcal{P}_{z}^{h}\left(  V\right)  -z\right)  ^{-1}$ allows the energy
estimates (\ref{P_z_est}) which entails (see (\ref{h_GN}))%
\begin{equation}
\sup_{\left[  a,b\right]  }\left\vert \left(  \mathcal{P}_{z}^{h}\left(
V\right)  -z\right)  ^{-1}f\right\vert \lesssim\frac{1}{h^{\frac{1}{2}}%
}\left\Vert f\right\Vert _{L^{2}\left(  \left(  a,b\right)  \right)  }\,.
\end{equation}
For $f=\left(  1-\chi_{h}\right)  u$, it follows%
\begin{equation}
\sup_{y=a,b}\left\vert \left(  \mathcal{P}_{z}^{h}\left(  \mathcal{V}%
^{h}\right)  -z\right)  ^{-1}u\left(  y\right)  \right\vert \lesssim\frac
{1}{h^{\frac{1}{2}}}\left\Vert u\right\Vert _{L^{2}\left(  \left(  a,b\right)
\right)  }\,.
\end{equation}
This leads to the relation (\ref{trace_est_grushin}) in this case.
\end{proof}

\begin{remark}
For multiples quantum wells models, $\left\vert \operatorname{Im}z_{j}%
^{h}\right\vert ^{-1}$ is expected to grow exponentially w.r.t. $h$ according
to (see (\ref{Fermi_golden_0})). In this case, the boundary values of $\left(
\mathcal{P}_{z}^{h}\left(  \mathcal{V}^{h}\right)  -z\right)  ^{-1}u$ are
dominated by $e^{\frac{S_{0}+2S_{U}+\eta}{h}}$.
\end{remark}

\subsection{\label{Section_traceest}Trace estimates in the quantum wells case}

We are now in the position of deducing trace estimates for the Green's
functions and the generalized eigenfunctions related to $Q_{0,0}%
^{h}(\mathcal{V}^{h})$ in the quantum well case.

\begin{proposition}
\label{Proposition_trace_est_h}Let $\mathcal{G}^{\zeta^{2},h}\left(
\cdot,y,\mathcal{V}^{h}\right)  $ and $\mathcal{H}^{\zeta^{2},h}\left(
\cdot,y,\mathcal{V}^{h}\right)  $ be defined by (\ref{Green_eq}%
)-(\ref{D_Green_eq}), with $\mathcal{V}^{h}=V+W^{h}$ fulfilling the Condition
\ref{condition_1}, and further assume the energy interval $\left[  \Lambda
_{1},\Lambda_{2}\right]  $ to have the property (\ref{Lambda_0_around}). Then,
the relations%
\begin{align}
\sup\limits_{y,y^{\prime}=a,b}\left\vert 1_{\left[  a,b\right]  }%
\mathcal{G}^{\zeta^{2},h}\left(  y,y^{\prime},\mathcal{V}^{h}\right)
\right\vert  &  \leq C_{a,b,c}\left(  h^{-2}+h^{-3}e^{-\frac{2S_{0}}{h}%
}\left(  \sup_{j\leq\ell}\left\vert \operatorname{Im}z_{j}^{h}\right\vert
^{-1}\right)  \right)  \,,\label{Trace_est1_1}\\
& \nonumber\\
\sup\limits_{y,y^{\prime}=a,b}\left\vert 1_{\left[  a,b\right]  }%
\mathcal{H}^{\zeta^{2},h}\left(  y,y^{\prime},\mathcal{V}^{h}\right)
\right\vert  &  \leq C_{a,b,c}\left(  h^{-3}+h^{-4}e^{-\frac{2S_{0}}{h}%
}\left(  \sup_{j\leq\ell}\left\vert \operatorname{Im}z_{j}^{h}\right\vert
^{-1}\right)  \right)  \,,\label{Trace_est1_2}\\
& \nonumber\\
\sup\limits_{y,y^{\prime}=a,b}\left\vert 1_{\left[  a,b\right]  }\partial
_{1}\mathcal{H}^{\zeta^{2},h}\left(  y,y^{\prime},\mathcal{V}^{h}\right)
\right\vert  &  \leq C_{a,b,c}\left(  h^{-4}+h^{-5}e^{-\frac{2S_{0}}{h}%
}\left(  \sup_{j\leq\ell}\left\vert \operatorname{Im}z_{j}^{h}\right\vert
^{-1}\right)  \right)  \,, \label{Trace_est1_3}%
\end{align}
hold for $\zeta^{2}\in\mathcal{B}_{\frac{\Lambda_{2}-\Lambda_{1}}{2}}(\left(
\Lambda_{1}+\Lambda_{2}\right)  /2)\cap\overline{\mathbb{C}^{+}}$, being
$C_{a,b,c}>0$ possibly depending on the data. Moreover, let $\psi_{0,0}%
^{h}(\cdot,k,\mathcal{V}^{h})$ be defined by (\ref{gen_eigenfun_eq}%
)-(\ref{gen_eigenfun_ext2}) and $k^{2}\in\left[  \Lambda_{1},\Lambda
_{2}\right]  $; then%
\begin{equation}
\sup\limits_{y,y^{\prime}=a,b}\left\vert 1_{\left[  a,b\right]  }\psi
_{0,0}^{h}(\cdot,k,\mathcal{V}^{h})\right\vert \leq C_{a,b,c}\left(
h^{-1}+h^{-2}e^{-\frac{2S_{0}}{h}}\left(  \sup_{j\leq\ell}\left\vert
\operatorname{Im}z_{j}^{h}\right\vert ^{-1}\right)  \right)  \,.
\label{Trace_est2}%
\end{equation}

\end{proposition}

\begin{proof}
Let consider at first the estimates (\ref{Trace_est1_1})-(\ref{Trace_est1_3}).
We focus on the case $y^{\prime}=a$; the result in the case $y^{\prime}=b$
follows by similar computations. Concerning the estimate (\ref{Trace_est1_1}),
we recall that the function $\mathcal{G}^{\zeta^{2},h}$ solves in $\left[
a,b\right]  $ a boundary values problem of the type (\ref{Ag_eq}), depending
on the potential and having, for $y^{\prime}=a$: $\gamma_{a}=-\frac{1}{h}$ and
$\gamma_{b}=0$. Introducing the auxiliary function: $v=\mathcal{G}^{\zeta
^{2},h}\left(  \cdot,a,\mathcal{V}^{h}\right)  -\mathcal{G}^{\zeta^{2}%
,h}\left(  \cdot,a,V\right)  $, we get%
\begin{equation}
v=-\left(  \mathcal{P}_{\zeta^{2}}^{h}\left(  \mathcal{V}^{h}\right)
-\zeta^{2}\right)  ^{-1}\left(  W^{h}\mathcal{G}^{\zeta^{2},h}\left(
\cdot,a,V\right)  \right)  \,, \label{v1}%
\end{equation}
and, according to the result of the Proposition \ref{proposition_trace-est},
the estimate%
\[
\sup_{y=a,b}\left\vert v(y)\right\vert \lesssim\left(  \frac{e^{-\frac{S_{0}%
}{h}}}{h^{\frac{3}{2}}}\left(  \sup_{j\leq\ell}\left\vert \operatorname{Im}%
z_{j}^{h}\right\vert ^{-1}\right)  +\frac{1}{h^{\frac{1}{2}}}\right)
\left\Vert W^{h}\mathcal{G}^{\zeta^{2},h}\left(  \cdot,a,V\right)  \right\Vert
_{L^{2}\left(  \left(  a,b\right)  \right)  }\,,
\]
follows. Since the points $\zeta^{2}\in\mathcal{B}_{\frac{\Lambda_{2}%
-\Lambda_{1}}{2}}(\left(  \Lambda_{1}+\Lambda_{2}\right)  /2)\cap
\overline{\mathbb{C}^{+}}$ fulfill the condition: $\inf_{\left[  a,b\right]
}V-\operatorname{Re}\zeta^{2}>c$, the Lemma \ref{Lemma_H1_est} applies to our
case. Using the estimate $(i)$ of this Lemma with $\varphi\left(
\cdot\right)  =d_{Ag}\left(  \cdot,a,V,\operatorname{Re}\zeta^{2}\right)  $,
we obtain%
\begin{equation}
\left\Vert W^{h}\mathcal{G}^{\zeta^{2},h}\left(  \cdot,a,V\right)  \right\Vert
_{L^{2}\left(  \left(  a,b\right)  \right)  }\leq C_{a,b,c}h^{-\frac{3}{2}%
}e^{-\frac{S_{0}}{h}}\,. \label{G_decay1}%
\end{equation}
Replacing this expression into the previous inequality, it results%
\[
\sup_{y=a,b}\left\vert v(y)\right\vert \leq C_{a,b,c}\left(  h^{-3}%
e^{-\frac{2S_{0}}{h}}\left(  \sup_{j\leq\ell}\left\vert \operatorname{Im}%
z_{j}^{h}\right\vert ^{-1}\right)  +\mathcal{O}\left(  e^{-\frac{S_{0}-\eta
}{h}}\right)  \right)  \,.
\]
Since $\mathcal{G}^{\zeta^{2},h}\left(  \cdot,a,\mathcal{V}^{h}\right)
=v+\mathcal{G}^{\zeta^{2},h}\left(  \cdot,a,V\right)  $, this yields%
\[
\sup_{y=a,b}\left\vert \mathcal{G}^{\zeta^{2},h}\left(  y,a,\mathcal{V}%
^{h}\right)  \right\vert \leq\sup_{y=a,b}\left\vert \mathcal{G}^{\zeta^{2}%
,h}\left(  y,a,V\right)  \right\vert +C_{a,b,c}\left(  h^{-3}e^{-\frac{2S_{0}%
}{h}}\left(  \sup_{j\leq\ell}\left\vert \operatorname{Im}z_{j}^{h}\right\vert
^{-1}\right)  +\mathcal{O}\left(  e^{-\frac{S_{0}-\eta}{h}}\right)  \right)
\,,
\]
and using once more the estimate $(i)$ in Lemma \ref{Lemma_H1_est}, the
estimate (\ref{Trace_est1_1}) follows.

Next we consider the function $\mathcal{H}^{\zeta^{2},h}\left(  \cdot
,a,\mathcal{V}^{h}\right)  $; this solves an interior problem of the type
(\ref{Ag_eq}) with: $\gamma_{a}=\frac{i\zeta}{h^{2}}$ and $\gamma_{b}=0$.
Introducing the function: $v=\mathcal{H}^{\zeta^{2},h}\left(  \cdot
,a,\mathcal{V}^{h}\right)  -\mathcal{H}^{\zeta^{2},h}\left(  \cdot,a,V\right)
$, we have%
\begin{equation}
v=-\left(  \mathcal{P}_{\zeta^{2}}^{h}\left(  \mathcal{V}^{h}\right)
-\zeta^{2}\right)  ^{-1}\left(  W^{h}\mathcal{H}^{\zeta^{2},h}\left(
\cdot,a,V\right)  \right)  \,. \label{v2}%
\end{equation}
Then, (\ref{Trace_est1_2}) follows, as before, from the result of the
Proposition \ref{Lemma_Grushin} and the estimate $(ii)$ in Lemma
\ref{Lemma_H1_est} by using the fact that, according to our assumptions,
$\left\vert \zeta\right\vert $ is bounded.

For the trace estimates of $\partial_{1}\mathcal{H}^{\zeta^{2},h}\left(
\cdot,a,\mathcal{V}^{h}\right)  $, we discuss separately the cases $\zeta
^{2}\in\omega_{ch}$ and $\zeta^{2}\notin\omega_{ch}$. Let $\zeta^{2}\in
\omega_{ch}\cap\left\{  \mathcal{B}_{\frac{\Lambda_{2}-\Lambda_{1}}{2}%
}(\left(  \Lambda_{1}+\Lambda_{2}\right)  /2)\cap\overline{\mathbb{C}^{+}%
}\right\}  $; the result of the Lemma \ref{Lemma_Grushin} applies to the
function (\ref{v2}) and, taking into account (\ref{h_GN}), we have%
\begin{equation}
\left\Vert e^{\frac{\varphi}{h}}v\right\Vert _{H^{1,h}\left(  \left[
a,b\right]  \right)  }\leq C_{a,b,c}\left(  \frac{1}{h}\left(  \sup_{j\leq
\ell}\left\vert \operatorname{Im}z_{j}^{h}\right\vert ^{-1}\right)
+e^{\frac{K_{a,b,c}\eta}{h}}\right)  \left\Vert W^{h}\mathcal{H}^{\zeta^{2}%
,h}\left(  \cdot,a,V\right)  \right\Vert _{L^{2}\left(  \left[  a,b\right]
\right)  }\,, \label{H_decay0.0}%
\end{equation}
where $\varphi\left(  \cdot\right)  =d_{Ag}\left(  \cdot,U,V,\lambda
^{0}\right)  $. Let us notice that, using the exponential weight fixed by the
distance $d_{Ag}\left(  \cdot,a,V,\operatorname{Re}\zeta^{2}\right)  $, the
second estimate in Lemma \ref{Lemma_H1_est} yields the inequalities%
\begin{equation}
\left\Vert W^{h}\mathcal{H}^{\zeta^{2},h}\left(  \cdot,a,V\right)  \right\Vert
_{L^{2}\left(  \left[  a,b\right]  \right)  }\leq C_{a,b,c}h^{-\frac{5}{2}%
}e^{-\frac{S_{0}}{h}}\,, \label{H_decay1}%
\end{equation}
and%
\begin{equation}
\left\Vert \mathcal{H}^{\zeta^{2},h}\left(  \cdot,a,V\right)  \right\Vert
_{H^{1,h}\left(  \left[  a,b\right]  \right)  }\leq C_{a,b,c}h^{-\frac{5}{2}%
}\,. \label{H_decay0}%
\end{equation}
Replacing (\ref{H_decay1}) into (\ref{H_decay0.0}), the previous estimate can
be rephrased as%
\[
\left\Vert e^{\frac{\varphi}{h}}v\right\Vert _{H^{1,h}\left(  \left[
a,b\right]  \right)  }\leq C_{a,b,c}\left(  \frac{e^{-\frac{S_{0}}{h}}%
}{h^{\frac{7}{2}}}\left(  \sup_{j\leq\ell}\left\vert \operatorname{Im}%
z_{j}^{h}\right\vert ^{-1}\right)  +\frac{e^{-\frac{S_{0}-K_{a,b,c}\eta}{h}}%
}{h^{\frac{5}{2}}}\right)  \,.
\]
Next, we recall that: $\mathcal{H}^{\zeta^{2},h}\left(  \cdot,a,\mathcal{V}%
^{h}\right)  =v+\mathcal{H}^{\zeta^{2},h}\left(  \cdot,a,V\right)  $, and use
(\ref{H_decay0}) to write
\begin{equation}
\left\Vert e^{\frac{\varphi}{h}}\mathcal{H}^{\zeta^{2},h}\left(
\cdot,a,\mathcal{V}^{h}\right)  \right\Vert _{H^{1,h}\left(  \left[
a,b\right]  \right)  }\leq C_{a,b,c}\left(  \frac{e^{\frac{S_{0}}{h}}%
}{h^{\frac{5}{2}}}+\frac{e^{-\frac{S_{0}}{h}}}{h^{\frac{7}{2}}}\left(
\sup_{j\leq\ell}\left\vert \operatorname{Im}z_{j}^{h}\right\vert ^{-1}\right)
+\frac{e^{-\frac{S_{0}-K_{a,b,c}\eta}{h}}}{h^{\frac{5}{2}}}\right)  \,.
\label{H_decay2}%
\end{equation}
In order to control the boundary values of $\partial_{1}\mathcal{H}^{\zeta
^{2},h}$, as $\zeta^{2}\in\omega_{ch}$, an $H^{1,h}$-estimate with exponential
weight for this function is needed. According to the equation: $\left(
-h^{2}\partial_{x}^{2}+\mathcal{V}^{h}-z\right)  \mathcal{H}^{z,h}\left(
\cdot,a,\mathcal{V}^{h}\right)  =0$, we have
\begin{equation}
\left\Vert he^{\frac{\varphi}{h}}\partial_{1}^{2}\mathcal{H}^{\zeta^{2}%
,h}\left(  \cdot,y,\mathcal{V}^{h}\right)  \right\Vert _{L^{2}\left(  \left[
a,b\right]  \right)  }=\frac{1}{h}\left\Vert \left(  \mathcal{V}^{h}-\zeta
^{2}\right)  e^{\frac{\varphi}{h}}\mathcal{H}^{\zeta^{2},h}\left(
\cdot,y,\mathcal{V}^{h}\right)  \right\Vert _{L^{2}\left(  \left[  a,b\right]
\right)  }\lesssim\frac{1}{h}\left\Vert e^{\frac{\varphi}{h}}\mathcal{H}%
^{\zeta^{2},h}\left(  \cdot,y,\mathcal{V}^{h}\right)  \right\Vert
_{L^{2}\left(  \left[  a,b\right]  \right)  }\,. \label{H_decay3}%
\end{equation}
Using (\ref{H_decay2}), (\ref{Agmon_est3}) and the equivalence of $\left\Vert
u\right\Vert _{H^{1,h}\left(  \left[  a,b\right]  \right)  }$ and $\left\Vert
hu^{\prime}\right\Vert _{L^{2}\left(  \left[  a,b\right]  \right)
}+\left\Vert u\right\Vert _{L^{2}\left(  \left[  a,b\right]  \right)  }$, we
obtain
\begin{equation}
\left\Vert e^{\frac{\varphi}{h}}\partial_{1}\mathcal{H}^{\zeta^{2},h}\left(
\cdot,a,\mathcal{V}^{h}\right)  \right\Vert _{H^{1,h}\left(  \left[
a,b\right]  \right)  }\leq C_{a,b,c}\left(  \frac{e^{\frac{S_{0}}{h}}%
}{h^{\frac{7}{2}}}+\frac{e^{-\frac{S_{0}}{h}}}{h^{\frac{9}{2}}}\left(
\sup_{j\leq\ell}\left\vert \operatorname{Im}z_{j}^{h}\right\vert ^{-1}\right)
+\frac{e^{-\frac{S_{0}-K_{a,b,c}\eta}{h}}}{h^{\frac{7}{2}}}\right)  \,,
\label{H_decay3_1}%
\end{equation}
Since $\sup_{\left\{  a,b\right\}  }\varphi=S_{0}$ and the trace is controlled
by the $H^{1,h}$-norm according to (\ref{h_GN}), this yields the inequality
(\ref{Trace_est1_3}) when $\zeta^{2}\in\omega_{ch}$. In the case $\zeta^{2}%
\in\left\{  \mathcal{B}_{\frac{\Lambda_{2}-\Lambda_{1}}{2}}(\left(
\Lambda_{1}+\Lambda_{2}\right)  /2)\cap\overline{\mathbb{C}^{+}}\right\}
\backslash\omega_{ch}$, we proceed as in Proposition
\ref{proposition_trace-est} by using the representation%
\begin{align*}
-v  &  =\left(  \mathcal{P}_{\zeta^{2}}^{h}\left(  \mathcal{V}^{h}\right)
-\zeta^{2}\right)  ^{-1}\left(  W^{h}\mathcal{H}^{\zeta^{2},h}\left(
\cdot,a,V\right)  \right) \\
& \\
&  =\left(  Q_{D}^{h}(\mathcal{V}^{h})-\zeta^{2}\right)  ^{-1}\chi_{h}\left(
W^{h}\mathcal{H}^{\zeta^{2},h}\left(  \cdot,a,V\right)  \right)  +\left(
\mathcal{P}_{\zeta^{2}}^{h}\left(  V\right)  -\zeta^{2}\right)  ^{-1}\left(
1-\chi_{h}\right)  \left(  W^{h}\mathcal{H}^{\zeta^{2},h}\left(
\cdot,a,V\right)  \right)  \,,
\end{align*}
where $\chi_{h}$ is the cutoff function defined in (\ref{chi_h}). Since
$W^{h}\mathcal{H}^{\zeta^{2},h}\left(  \cdot,a,V\right)  $ and $\left(
1-\chi_{h}\right)  $ have disjoint supports, this \ reduces to the identity%
\[
v=-\left(  Q_{D}^{h}(\mathcal{V}^{h})-\zeta^{2}\right)  ^{-1}\left(
W^{h}\mathcal{H}^{\zeta^{2},h}\left(  \cdot,a,V\right)  \right)  \,.
\]
Then the relation (\ref{H_decay1}) and standard resolvent estimates for
selfadjoint operators yield (e.g. in \cite{FMN2}, Proposition 4.1)%
\[
\left\Vert v\right\Vert _{H^{1,h}\left(  \left[  a,b\right]  \right)  }%
\leq\frac{C_{a,b,c}}{h^{\frac{7}{2}}}e^{-\frac{S_{0}}{h}}\left(  \frac
{1}{d\left(  \zeta^{2},\sigma\left(  Q_{D}^{h}(\mathcal{V}^{h})\right)
\right)  }+1\right)  \,.
\]
Recalling that: $d\left(  \zeta^{2},\sigma\left(  Q_{D}^{h}(\mathcal{V}%
^{h})\right)  \right)  >ch$ and $\mathcal{H}^{\zeta^{2},h}\left(
\cdot,a,\mathcal{V}^{h}\right)  =v+\mathcal{H}^{\zeta^{2},h}\left(
\cdot,a,V\right)  $, and using (\ref{H_decay0}), it follows%
\begin{equation}
\left\Vert \mathcal{H}^{\zeta^{2},h}\left(  \cdot,a,\mathcal{V}^{h}\right)
\right\Vert _{H^{1,h}\left(  \left[  a,b\right]  \right)  }\leq C_{a,b,c}%
h^{-\frac{5}{2}}\,. \label{H_decay4}%
\end{equation}
Proceeding as before, we use (\ref{H_decay3}) (with $\varphi=0$) and
(\ref{H_decay4}) to obtain the $H^{1,h}$-norm estimate%
\begin{equation}
\left\Vert \partial_{1}\mathcal{H}^{\zeta^{2},h}\left(  \cdot,a,\mathcal{V}%
^{h}\right)  \right\Vert _{H^{1,h}\left(  \left[  a,b\right]  \right)  }\leq
C_{a,b,c}h^{-\frac{7}{2}}\,, \label{H_decay4_1}%
\end{equation}
which implies (\ref{Trace_est1_3}), once (\ref{h_GN}) is taken into account.

For the estimate (\ref{Trace_est2}), let us notice that, according to the
equations (\ref{gen_eigenfun_eq})-(\ref{gen_eigenfun_ext2}), $\psi_{0,0}%
^{h}(\cdot,k,\mathcal{V}^{h})$ solves a problem of the type (\ref{Ag_eq})
with: $\gamma_{a}=2ike^{i\frac{k}{h}a}$ and $\gamma_{b}=0$ for $k>0$, or
$\gamma_{a}=0$ and $\gamma_{b}=2ike^{i\frac{k}{h}b}$ for $k<0$. The auxiliary
function: $v=\psi_{0,0}^{h}(\cdot,k,\mathcal{V}^{h})-\psi_{0,0}^{h}%
(\cdot,k,V)$, fulfills the equation%
\[
v=-\left(  \mathcal{P}_{k^{2}}^{h}\left(  \mathcal{V}^{h}\right)  -\zeta
^{2}\right)  ^{-1}\left(  W^{h}\psi_{0,0}^{h}(\cdot,k,V)\right)  \,,
\]
and (\ref{Trace_est2}) follows by proceeding as before and using the results
of the Lemma \ref{Lemma_H1_est1} and the Proposition
\ref{proposition_trace-est}.
\end{proof}

\section{\label{Section_evolution}The quantum evolution problem}

Next we study the time propagator generated by the operators $Q_{\theta
_{1},\theta_{2}}^{h}(\mathcal{V}^{h})$ where $\mathcal{V}^{h}$ describes the
regime of quantum wells in a semiclassical island. The aim is to compare the
modified dynamics and the unitary one, related to the selfadjoint model
$Q_{0,0}^{h}(\mathcal{V}^{h})$, when the parameters $\theta_{j=1,2}$ and $h$
are small. Our strategy consists in using the generalized eigenfunctions of
$Q_{\theta_{1},\theta_{2}}^{h}$ in order to define a similarity between
$Q_{\theta_{1},\theta_{2}}^{h}$ and $Q_{0,0}^{h}$ in some spectral subspace
corresponding to energies close to the resonances. Then, a generalized
eigenfunction expansion allows to prove that the corresponding dynamical
systems are 'close' each other, uniformly in time, provided that
$\theta_{j=1,2}$ are infinitesimal functions of $h$.

\subsection{Generalized eigenfunctions expansions}

We are now in the position to give generalized eigenfunctions expansions,
starting from the formula (\ref{gen_eigenfun_Krein_h}), in the case of
$h$-dependent potentials describing quantum wells.

\begin{proposition}
\label{Proposition_gen_eigenfun_exp}Let $\mathcal{V}^{h}$ be defined according
to the Condition \ref{condition_1} and $\left\vert \theta_{j=1,2}\right\vert
\lesssim h^{N_{0}}$, with $N_{0}\geq4$. If the interval $\left[  \Lambda
_{1},\Lambda_{2}\right]  $ verifies the condition (\ref{Lambda_0_around}) and
the lower bound (\ref{Fermi_golden}) holds, then we have%
\begin{align}
&  1_{\left[  \Lambda_{1},\Lambda_{2}\right]  }\left(  k^{2}\right)  \left[
\psi_{\theta_{1},\theta_{2}}^{h}(x,k,\mathcal{V}^{h})-\psi_{0,0}%
^{h}(x,k,\mathcal{V}^{h})\right] \label{gen_eigenfun_exp_h}\\
&  =1_{\left[  \Lambda_{1},\Lambda_{2}\right]  }\left(  k^{2}\right)  \left[
\mathcal{O}\left(  \theta_{2}\right)  G^{\left\vert k\right\vert ,h}\left(
\cdot,b,\mathcal{V}^{h}\right)  +\mathcal{O}\left(  \theta_{1}\right)
H^{\left\vert k\right\vert ,h}(\cdot,b,\mathcal{V}^{h})+\mathcal{O}\left(
\theta_{2}\right)  G^{\left\vert k\right\vert ,h}(\cdot,a,\mathcal{V}%
^{h})+\mathcal{O}\left(  \theta_{1}\right)  H^{\left\vert k\right\vert
,h}(\cdot,a,\mathcal{V}^{h})\right]  \,,\nonumber
\end{align}
with $\mathcal{O}\left(  \cdot\right)  $ denoting functions of the variables
$\left(  k,\theta_{1},\theta_{2},h\right)  $ holomorphic w.r.t. $\left(
\theta_{1},\theta_{2}\right)  $.
\end{proposition}

\begin{proof}
From the definitions (\ref{AB_teta1,2})-(\ref{ab_teta1,2}), (\ref{q_z}),
(\ref{GH_k}) and (\ref{krein_coeff_1}), the matrix $\mathcal{M}^{h}\left(
k,\theta_{1},\theta_{2},\mathcal{V}^{h}\right)  $ expresses in terms of the
functions $G^{k,h}$ and $H^{k,h}$ according to%
\begin{align}
&  \mathcal{M}^{h}\left(  k,\theta_{1},\theta_{2},\mathcal{V}^{h}\right)
\label{krein_coeff_1_1}\\
=  &
\begin{pmatrix}
\mathcal{O}\left(  \theta_{2}\right)  H^{k,h}\left(  b^{-},b,\mathcal{V}%
^{h}\right)  & \mathcal{O}\left(  \theta_{2}\right)  \partial_{1}%
H^{k,h}\left(  b,b,\mathcal{V}^{h}\right)  & \mathcal{O}\left(  \theta
_{2}\right)  H^{k,h}\left(  a,b,\mathcal{V}^{h}\right)  & \mathcal{O}\left(
\theta_{2}\right)  \partial_{1}H^{k,h}\left(  b,a,\mathcal{V}^{h}\right) \\
\mathcal{O}\left(  \theta_{1}\right)  G^{k,h}\left(  b,b,\mathcal{V}%
^{h}\right)  & \mathcal{O}\left(  \theta_{1}\right)  H^{k,h}\left(
b^{-},b,\mathcal{V}^{h}\right)  & \mathcal{O}\left(  \theta_{1}\right)
G^{k,h}\left(  b,a,\mathcal{V}^{h}\right)  & \mathcal{O}\left(  \theta
_{1}\right)  H^{k,h}\left(  b,a,\mathcal{V}^{h}\right) \\
\mathcal{O}\left(  \theta_{2}\right)  H^{k,h}\left(  b,a,\mathcal{V}%
^{h}\right)  & \mathcal{O}\left(  \theta_{2}\right)  \partial_{1}%
H^{k,h}\left(  a,b,\mathcal{V}^{h}\right)  & \mathcal{O}\left(  \theta
_{2}\right)  H^{k,h}\left(  a^{+},a,\mathcal{V}^{h}\right)  & \mathcal{O}%
\left(  \theta_{2}\right)  \partial_{1}H^{k,h}\left(  a,a,\mathcal{V}%
^{h}\right) \\
\mathcal{O}\left(  \theta_{1}\right)  G^{k,h}\left(  a,b,\mathcal{V}%
^{h}\right)  & \mathcal{O}\left(  \theta_{1}\right)  H^{k,h}\left(
a,b,\mathcal{V}^{h}\right)  & \mathcal{O}\left(  \theta_{1}\right)
G^{k,h}\left(  a,a,\mathcal{V}^{h}\right)  & \mathcal{O}\left(  \theta
_{1}\right)  H^{k,h}\left(  a^{+},a,\mathcal{V}^{h}\right)
\end{pmatrix}
+\nonumber\\
-\frac{1}{h^{2}}  &
\begin{pmatrix}
\alpha\left(  \theta_{2}\right)  +\mathcal{O}\left(  \theta_{2}\right)  &  &
& \\
& \alpha\left(  \theta_{1}\right)  +\mathcal{O}\left(  \theta_{1}\right)  &  &
\\
&  & \alpha\left(  -\theta_{2}\right)  +\mathcal{O}\left(  \theta_{2}\right)
& \\
&  &  & \alpha\left(  -\theta_{1}\right)  +\mathcal{O}\left(  \theta
_{1}\right)
\end{pmatrix}
\nonumber
\end{align}
where $\mathcal{\alpha}\left(  \theta\right)  =1+e^{\frac{\theta}{2}}$. For
$k^{2}\in\left[  \Lambda_{1},\Lambda_{2}\right]  $ and $k>0$, the trace
estimates provided by the Proposition \ref{Proposition_trace_est_h} can be
used to control the coefficients in (\ref{krein_coeff_1_1}). In particular,
using (\ref{Trace_est1_1})-(\ref{Trace_est1_3}) with the condition
(\ref{Fermi_golden}) yields%
\begin{gather}
\mathcal{M}^{h}\left(  k,\theta_{1},\theta_{2},\mathcal{V}^{h}\right)
=-\frac{1}{h^{2}}%
\begin{pmatrix}
\alpha\left(  \theta_{2}\right)  &  &  & \\
& \alpha\left(  \theta_{1}\right)  &  & \\
&  & \alpha\left(  -\theta_{2}\right)  & \\
&  &  & \alpha\left(  -\theta_{1}\right)
\end{pmatrix}
\label{M_k_teta_h1}\\
+%
\begin{pmatrix}
\mathcal{O}\left(  \theta_{2}\right)  \mathcal{O}\left(  \frac{1}{h^{4}%
}\right)  & \mathcal{O}\left(  \theta_{2}\right)  \mathcal{O}\left(  \frac
{1}{h^{5}}\right)  & \mathcal{O}\left(  \theta_{2}\right)  \mathcal{O}\left(
\frac{1}{h^{4}}\right)  & \mathcal{O}\left(  \theta_{2}\right)  \mathcal{O}%
\left(  \frac{1}{h^{5}}\right) \\
\mathcal{O}\left(  \theta_{1}\right)  \mathcal{O}\left(  \frac{1}{h^{3}%
}\right)  & \mathcal{O}\left(  \theta_{1}\right)  \mathcal{O}\left(  \frac
{1}{h^{4}}\right)  & \mathcal{O}\left(  \theta_{1}\right)  \mathcal{O}\left(
\frac{1}{h^{3}}\right)  & \mathcal{O}\left(  \theta_{1}\right)  \mathcal{O}%
\left(  \frac{1}{h^{4}}\right) \\
\mathcal{O}\left(  \theta_{2}\right)  \mathcal{O}\left(  \frac{1}{h^{4}%
}\right)  & \mathcal{O}\left(  \theta_{2}\right)  \mathcal{O}\left(  \frac
{1}{h^{5}}\right)  & \mathcal{O}\left(  \theta_{2}\right)  \mathcal{O}\left(
\frac{1}{h^{4}}\right)  & \mathcal{O}\left(  \theta_{2}\right)  \mathcal{O}%
\left(  \frac{1}{h^{5}}\right) \\
\mathcal{O}\left(  \theta_{1}\right)  \mathcal{O}\left(  \frac{1}{h^{3}%
}\right)  & \mathcal{O}\left(  \theta_{1}\right)  \mathcal{O}\left(  \frac
{1}{h^{4}}\right)  & \mathcal{O}\left(  \theta_{1}\right)  \mathcal{O}\left(
\frac{1}{h^{3}}\right)  & \mathcal{O}\left(  \theta_{1}\right)  \mathcal{O}%
\left(  \frac{1}{h^{4}}\right)
\end{pmatrix}
\,.\nonumber
\end{gather}
Since, for $k\rightarrow-k$, $G^{k,h}$ and $H^{k,h}$ change by complex
conjugation (see for instance the representations (\ref{G_k_jost}%
)-(\ref{H_k_jost})), the trace estimates of the Proposition
\ref{Proposition_trace_est_h} apply to $G^{-k,h}$ and $H^{-k,h}$ for $k>0$ and
$k^{2}\in\left[  \Lambda_{1},\Lambda_{2}\right]  $. Thus, the above expansion
extends to the whole set $\left\{  k\in\mathbb{R\,}\left\vert \ k^{2}%
\in\left[  \Lambda_{1},\Lambda_{2}\right]  \right.  \right\}  $. For
$\left\vert \theta_{j=1,2}\right\vert \lesssim h^{N_{0}}$ with $N_{0}\geq4$
and $h\in\left(  0,h_{0}\right]  $, a direct computation gives%
\begin{equation}
\det1_{\left[  \Lambda_{1},\Lambda_{2}\right]  }\left(  k^{2}\right)
h^{2}\mathcal{M}^{h}\left(  k,\theta_{1},\theta_{2},\mathcal{V}^{h}\right)
\gtrsim1\,, \label{det_h_upbound}%
\end{equation}
uniformly w.r.t. $h\in\left(  0,h_{0}\right]  $ for a suitable $h_{0}>0$.
Thus, the matrix $\mathcal{M}^{h}$ is invertible whenever the parameters
$\theta_{j=1,2}$ are small, depending on $h$. In the above conditions, the
inverse matrix writes as%
\begin{align}
&  1_{\left[  \Lambda_{1},\Lambda_{2}\right]  }\left(  k^{2}\right)  \left(
h^{2}\mathcal{M}^{h}\left(  k,\theta_{1},\theta_{2},\mathcal{V}^{h}\right)
\right)  ^{-1}\label{coeff_h_exp}\\
&  =\frac{1}{\det h^{2}1_{\left[  \Lambda_{1},\Lambda_{2}\right]  }\left(
k^{2}\right)  \mathcal{M}^{h}\left(  k,\theta_{1},\theta_{2},\mathcal{V}%
^{h}\right)  }\left[  \det\left(  A_{\theta_{1},\theta_{2}}\right)
\,diag\left\{  \lambda_{i}^{h}\right\}  \,+\mathcal{R}^{h}\right]
\,.\nonumber
\end{align}
The main term in (\ref{coeff_h_exp}), $\det\left(  A_{\theta_{1},\theta_{2}%
}\right)  \,diag\left(  \lambda_{i}\right)  $, is the $\mathbb{C}^{4,4}$
diagonal matrix defined by the coefficients%
\begin{equation}
\left\{  \lambda_{i}\right\}  _{i=1}^{4}=\left\{  \frac{-1}{\alpha\left(
\theta_{2}\right)  }\,,\ \frac{-1}{\alpha\left(  \theta_{1}\right)
}\,,\ \frac{-1}{\alpha\left(  -\theta_{2}\right)  }\,,\ \frac{-1}%
{\alpha\left(  -\theta_{1}\right)  }\right\}  \,,\quad\det\left(
A_{\theta_{1},\theta_{2}}\right)  =%
{\textstyle\prod\nolimits_{n,j=1,2}}
\alpha\left(  (-1)^{n}\theta_{j}\right)  \,,
\end{equation}
which, according to the explicit form of $\alpha\left(  \cdot\right)  $, are
uniformly bounded when $\theta_{j=1,2}$ are close to the origin. The remainder
$\mathcal{R}^{h}$ is a matrix-valued function of the variables $k$,
$\theta_{j=1,2}$ and $h$; under our assumptions, it results: $\mathcal{R}%
_{n,m}^{h}=\mathcal{O}\left(  h\right)  $, $n,m=1,...4$, in the sense of the
metric space%
\begin{equation}
\left\{  k\in\mathbb{R}\,,\ k^{2}\in\left[  \Lambda_{1},\Lambda_{2}\right]
\right\}  \times\left\{  \left(  \theta_{1},\theta_{2}\right)  \in
\mathbb{C}^{2}\,,\ \left\vert \theta_{j=1,2}\right\vert \lesssim h^{N_{0}%
}\right\}  \times\left(  0,h_{0}\right]  \,.
\end{equation}

Let us consider the coefficients at the r.h.s. of the formula
(\ref{gen_eigenfun_Krein_h}). Due to the above remarks we have $\left(
\mathcal{M}^{h}\left(  k,\theta_{1},\theta_{2},\mathcal{V}\right)  \right)
^{-1}B_{\theta_{1},\theta_{2}}=\mathcal{O}\left(  B_{\theta_{1},\theta_{2}%
}\right)  $ in the matrix-norm sense. Moreover, being $\psi_{0,0}^{h}$
$\mathcal{C}^{1}$-continuous in $x$, we can use the estimate (\ref{Trace_est2}%
) to evaluate $\Gamma_{1}\psi_{0,0}^{h}$; the condition (\ref{Fermi_golden})
yields: $\Gamma_{1}\psi_{0,0}^{h}(\cdot,k,\mathcal{V}^{h})=\mathcal{O}\left(
\frac{1}{h^{2}}\right)  $. Taking into account the explicit form of
$B_{\theta_{1},\theta_{2}}$, it follows%
\begin{equation}
\left\{  \sum_{j=1}^{4}\left[  \left(  \mathcal{M}^{h}\left(  k,\theta
_{1},\theta_{2},\mathcal{V}^{h}\right)  \right)  ^{-1}B_{\theta_{1},\theta
_{2}}\right]  _{ij}\left[  \Gamma_{1}\psi_{0,0}^{h}(\cdot,k,\mathcal{V}%
^{h})\right]  _{j}\right\}  _{i=1}^{4}=%
\begin{pmatrix}
\mathcal{O}\left(  \frac{\theta_{2}}{h^{2}}\right)  \,, & \mathcal{O}\left(
\frac{\theta_{1}}{h^{2}}\right)  \,, & \mathcal{O}\left(  \frac{\theta_{2}%
}{h^{2}}\right)  \,, & \mathcal{O}\left(  \frac{\theta_{1}}{h^{2}}\right)
\end{pmatrix}
\,, \label{coeff_h}%
\end{equation}
which leads to (\ref{gen_eigenfun_exp_h}).

Finally, we notice that, since the matrix coefficients in $\mathcal{M}%
^{h}\left(  k,\theta_{1},\theta_{2},\mathcal{V}^{h}\right)  $ and
$B_{\theta_{1},\theta_{2}}$ holomorphic w.r.t. $\left(  \theta_{1},\theta
_{2}\right)  $, the same holds for the coefficients of $\left(  \mathcal{M}%
^{h}\left(  k,\theta_{1},\theta_{2},\mathcal{U}\right)  \right)
^{-1}B_{\theta_{1},\theta_{2}}$. Then, the symbols $\mathcal{O}\left(
\cdot\right)  $ in (\ref{gen_eigenfun_exp_h}), depending from the variables
$\left(  k,\theta_{1},\theta_{2},h\right)  $, denote holomorphic functions of
$\theta_{1}$ and $\theta_{2}$.
\end{proof}

\subsection{\label{Section_simil}Similarity of operators}

We next construct a similarity between $Q_{\theta_{1},\theta_{2}}%
^{h}(\mathcal{V}^{h})$ and $Q_{0,0}^{h}(\mathcal{V}^{h})$ in a suitable
subspace. Let introduce the generalized Fourier transform associated to
$Q_{0,0}^{h}(\mathcal{V})$%
\begin{equation}
\left(  \mathcal{F}_{\mathcal{V}}^{h}\varphi\right)  (k)=\int_{\mathbb{R}%
}\frac{dx}{\left(  2\pi h\right)  ^{1/2}}\,\left(  \psi_{0,0}^{h}%
(x,k,\mathcal{V})\right)  ^{\ast}\varphi(x)\,,\qquad\varphi\in L^{2}%
(\mathbb{R})\,. \label{gen_Fourier_h}%
\end{equation}
For potentials $\mathcal{V}$ defined as in (\ref{V}), $\mathcal{F}%
_{\mathcal{V}}^{h}$ is a bounded operator on $L^{2}(\mathbb{R})$ with a right
inverse coinciding with the adjoint $\left(  \mathcal{F}_{\mathcal{V}}%
^{h}\right)  ^{\ast}$%
\begin{equation}
\left(  \mathcal{F}_{\mathcal{V}^{h}}^{h}\right)  ^{\ast}f(x)=\int\frac
{dk}{\left(  2\pi h\right)  ^{1/2}}\,\psi_{0,0}^{h}(x,k,\mathcal{V}%
^{h})f(k)\,. \label{gen_Fourier_h_inv}%
\end{equation}
In particular, it results: $\mathcal{F}_{\mathcal{V}}^{h}\left(
\mathcal{F}_{\mathcal{V}}^{h}\right)  ^{\ast}=\mathbb{I}$ in $L^{2}%
(\mathbb{R})$, while the product $\left(  \mathcal{F}_{\mathcal{V}}%
^{h}\right)  ^{\ast}\mathcal{F}_{\mathcal{V}}^{h}$ defines the projector on
the absolutely continuous subspace of $Q_{0,0}^{h}(\mathcal{V})$ (cf.
\cite{Yafa}). For $\mathcal{V}^{h}$ satisfying the assumptions of Section
\ref{Section_Resonances}, we introduce the maps $\phi_{\alpha}^{h}$ and
$\psi_{\alpha}^{h}$, which act on $L^{2}\left(  \mathbb{R}\right)  $ as%
\begin{align}
\phi_{\alpha}^{h}(\varphi,f)  &  =\int_{\mathbb{R}}\frac{dk}{\left(  2\pi
h\right)  ^{1/2}}\,\,f(k)\,G^{\left\vert k\right\vert ,h}\left(
x,\alpha,\mathcal{V}^{h}\right)  \left(  \mathcal{F}_{\mathcal{V}^{h}}%
^{h}\varphi\right)  (k)\,,\quad\alpha\in\left\{  a,b\right\}
\,,\label{Phi_alpha_h}\\
& \nonumber\\
\psi_{\alpha}^{h}(\varphi,f)  &  =\int_{\mathbb{R}}\frac{dk}{\left(  2\pi
h\right)  ^{1/2}}\,\,f(k)\,H^{\left\vert k\right\vert ,h}(x,\alpha
,\mathcal{V}^{h})\left(  \mathcal{F}_{\mathcal{V}^{h}}^{h}\varphi\right)
(k)\,,\quad\alpha\in\left\{  a,b\right\}  \,. \label{Psi_alpha_h}%
\end{align}
Here $G^{k,h}$ and $H^{k,h}$ are the limits of the Green's functions on the
branch cut (see the definition in (\ref{G_k_jost})-(\ref{H_k_jost})), while
$f$ is an auxiliary function, possibly depending on $h$ and $\theta_{j=12}$
aside from $k$.

\begin{lemma}
\label{Lemma_Simil_est}Let $h\in\left(  0,h_{0}\right]  $ and $\mathcal{V}%
^{h}=V+W^{h}$ be defined as in (\ref{V_h})-(\ref{V_h1}) with $h_{0}$ and $c$
suitably small. Assume $f\in L_{k}^{\infty}\left(  \mathbb{R}\right)  $
uniformly w.r.t. $h$, $\theta_{j=12}$ and $\left.  \text{supp }f\subset
\Omega_{c}\left(  V\right)  \right.  $, with%
\begin{equation}
\Omega_{c}\left(  V\right)  =\left\{  k\in\mathbb{R}\,\left\vert
\ V-k^{2}>c\right.  \right\}  \,. \label{Omega_V_1}%
\end{equation}
Then it results%
\begin{equation}
\left\Vert \phi_{\alpha}^{h}(\cdot,f)\right\Vert _{\mathcal{L}\left(
L^{2}\left(  \mathbb{R}\right)  \right)  }+\left\Vert \psi_{\alpha}^{h}%
(\cdot,f)\right\Vert _{\mathcal{L}\left(  L^{2}\left(  \mathbb{R}\right)
\right)  }\leq\frac{C_{a,b,c}}{h^{2}}\,, \label{phi_psi_est}%
\end{equation}
where $C_{a,b,c}$ is a positive constant depending on the data.
\end{lemma}

\begin{proof}
Each of the maps $\phi_{\alpha}^{h}\left(  \cdot,f\right)  $ and $\psi
_{\alpha}^{h}\left(  \cdot,f\right)  $, $\alpha=a,b$, can be expressed as a
superpositions of terms having the following form%
\begin{equation}
1_{\left\{  x\geq\alpha\right\}  }\left(  \mathcal{F}_{\mathcal{V}^{h}}%
^{h}\right)  ^{\ast}\left(  f\mu_{1}+\mathcal{P\circ}f\mu_{2}\right)
\mathcal{F}_{\mathcal{V}^{h}}^{h}+1_{\left\{  x<\alpha\right\}  }\left(
\mathcal{F}_{\mathcal{V}^{h}}^{h}\right)  ^{\ast}\left(  f\mu_{3}%
+\mathcal{P\circ}f\mu_{4}\right)  \mathcal{F}_{\mathcal{V}^{h}}^{h}\,,
\label{phi_psi_formula}%
\end{equation}
where $\mu_{i}$, depending on the variables $k$ and $h$, are bounded w.r.t.
$k$ and behave as $\mathcal{O}\left(  1/h^{2}\right)  $. The estimate
(\ref{phi_psi_est}) is a direct consequence of this representation. Next, we
focus on the case $\alpha=b$ and explicitly consider $\phi_{b}^{h}(\cdot,f)$.
As it follows from (\ref{G_k_jost})-(\ref{gen_eigenfun_h_jost}), the functions
$G^{\left\vert k\right\vert ,h}\left(  \cdot,b,\mathcal{V}^{h}\right)  $ and
$H^{\left\vert k\right\vert ,h}\left(  \cdot,b,\mathcal{V}^{h}\right)  $ allow
the representations%
\begin{align}
1_{\left\{  k>0\right\}  }G^{k,h}\left(  \cdot,b,\mathcal{V}^{h}\right)   &
=\left\{
\begin{array}
[c]{lll}%
-\frac{1}{2ikh}\psi_{0,0}^{h}(\cdot,k,\mathcal{V}^{h})\chi_{-}^{h}\left(
b,k,\mathcal{V}^{h}\right)  \,, &  & x\geq b\,,\\
&  & \\
-\frac{1}{2ikh}\psi_{0,0}^{h}(\cdot,-k,\mathcal{V}^{h})\chi_{+}^{h}\left(
b,k,\mathcal{V}^{h}\right)  \,, &  & x<b\,.
\end{array}
\right. \label{G_gen_eigenfun_k1}\\
& \nonumber\\
1_{\left\{  k>0\right\}  }H^{k,h}\left(  \cdot,b,\mathcal{V}^{h}\right)   &
=\left\{
\begin{array}
[c]{lll}%
\frac{1}{2ikh}\psi_{0,0}^{h}(\cdot,k,\mathcal{V}^{h})\partial_{1}\chi_{-}%
^{h}\left(  b,-k,\mathcal{V}^{h}\right)  \,, &  & x\geq b\,,\\
&  & \\
\frac{1}{2ikh}\psi_{0,0}^{h}(\cdot,-k,\mathcal{V}^{h})\partial_{1}\chi_{+}%
^{h}\left(  b,-k,\mathcal{V}^{h}\right)  \,, &  & x<b\,.
\end{array}
\right.  \label{H_gen_eigenfun_k1}%
\end{align}
The condition supp $f\subset\Omega_{c}\left(  V\right)  $ implies $\left\vert
k\right\vert >c^{\frac{1}{2}}>0$; thus, using (\ref{G_gen_eigenfun_k1}) for
$x\geq b$ we get%
\begin{align}
1_{\left\{  x\geq b\right\}  }\phi_{b}^{h}(\varphi,f)  &  =1_{\left\{  x\geq
b\right\}  }\int_{0}^{+\infty}\frac{dk}{\left(  2\pi h\right)  ^{1/2}%
}\,f(k)\,\mathcal{O}\left(  \frac{1}{h}\right)  \psi_{0,0}^{h}(\cdot
,k,\mathcal{V}^{h})\chi_{-}^{h}\left(  b,k,\mathcal{V}^{h}\right)  \left(
\mathcal{F}_{\mathcal{V}^{h}}^{h}\varphi\right)  (k)+\nonumber\\
& \nonumber\\
1_{\left\{  x\geq b\right\}  }  &  \int_{-\infty}^{0}\frac{dk}{\left(  2\pi
h\right)  ^{1/2}}\,f(k)\,\mathcal{O}\left(  \frac{1}{h}\right)  \psi_{0,0}%
^{h}(\cdot,-k,\mathcal{V}^{h})\chi_{-}^{h}\left(  b,-k,\mathcal{V}^{h}\right)
\left(  \mathcal{F}_{\mathcal{V}^{h}}^{h}\varphi\right)  (k)\,,
\end{align}
while, taking into account the result of the Proposition
\ref{Proposition_Jost_h_local}, it follows%
\begin{align}
1_{\left\{  x\geq b\right\}  }\phi_{b}^{h}(\varphi,f)  &  =1_{\left\{  x\geq
b\right\}  }\int_{0}^{+\infty}\frac{dk}{\left(  2\pi h\right)  ^{1/2}}%
\psi_{0,0}^{h}(\cdot,k,\mathcal{V}^{h})\,f(k)\,\mathcal{O}\left(  \frac{1}%
{h}\right)  \left(  \mathcal{F}_{\mathcal{V}^{h}}^{h}\varphi\right)
(k)+\nonumber\\
& \nonumber\\
1_{\left\{  x\geq b\right\}  }  &  \int_{-\infty}^{0}\frac{dk}{\left(  2\pi
h\right)  ^{1/2}}\psi_{0,0}^{h}(\cdot,-k,\mathcal{V}^{h}\,)\,f(k)\,\mathcal{O}%
\left(  \frac{1}{h}\right)  \left(  \mathcal{F}_{\mathcal{V}^{h}}^{h}%
\varphi\right)  (k)\,,
\end{align}
Let $\mathcal{P}$ denotes be the parity operator: $\mathcal{P}u(t)=u(-t)$. The
previous identity rephrases as%
\begin{equation}
1_{\left\{  x\geq b\right\}  }\phi_{b}^{h}(\varphi,f)=1_{\left\{  x\geq
b\right\}  }\left(  \mathcal{F}_{\mathcal{V}^{h}}^{h}\right)  ^{\ast}\,\left(
1_{\left\{  k\geq0\right\}  }\left(  k\right)  \left(  f(k)\,\mathcal{O}%
\left(  \frac{1}{h}\right)  +\mathcal{P\circ}\left(  f(k)\,\mathcal{O}\left(
\frac{1}{h}\right)  \right)  \right)  \mathcal{F}_{\mathcal{V}^{h}}^{h}%
\varphi\right)  \,, \label{simil_est_1}%
\end{equation}
where the symbols $\mathcal{O}\left(  \cdot\right)  $, denoting functions of
the variables $k$ and $h$, are defined in the sense of the metric space
$\Omega_{c}\left(  V\right)  \times\left(  0,h_{0}\right]  $. Using
(\ref{G_gen_eigenfun_k1}) for $x<b$ leads to%
\begin{align}
1_{\left\{  x<b\right\}  }\phi_{b}^{h}(\varphi,f)  &  =1_{\left\{
x<b\right\}  }\int_{0}^{+\infty}\frac{dk}{\left(  2\pi h\right)  ^{1/2}%
}\,f(k)\,\mathcal{O}\left(  \frac{1}{h}\right)  \psi_{0,0}^{h}(\cdot
,-k,\mathcal{V}^{h})\chi_{+}^{h}\left(  b,k,\mathcal{V}^{h}\right)  \left(
\mathcal{F}_{\mathcal{V}^{h}}^{h}\varphi\right)  (k)+\nonumber\\
& \nonumber\\
1_{\left\{  x<b\right\}  }  &  \int_{-\infty}^{0}\frac{dk}{\left(  2\pi
h\right)  ^{1/2}}\,f(k)\,\mathcal{O}\left(  \frac{1}{h}\right)  \psi_{0,0}%
^{h}(\cdot,k,\mathcal{V}^{h})\chi_{+}^{h}\left(  b,-k,\mathcal{V}^{h}\right)
\left(  \mathcal{F}_{\mathcal{V}^{h}}^{h}\varphi\right)  (k)\,,
\end{align}
and, proceeding as before, we get%
\begin{equation}
1_{\left\{  x<b\right\}  }\phi_{b}^{h}(\varphi,f)=1_{\left\{  x<b\right\}
}\left(  \mathcal{F}_{\mathcal{V}^{h}}^{h}\right)  ^{\ast}\,\left(
1_{\left\{  k<0\right\}  }\left(  k\right)  \left(  f(k)\,\mathcal{O}\left(
\frac{1}{h}\right)  +\mathcal{P\circ}\left(  f(k)\,\mathcal{O}\left(  \frac
{1}{h}\right)  \right)  \right)  \mathcal{F}_{\mathcal{V}^{h}}^{h}%
\varphi\right)  \,, \label{simil_est_2}%
\end{equation}
From (\ref{simil_est_1}) and (\ref{simil_est_2}) we get a representation of
the type given in (\ref{phi_psi_formula}); it follows: $\left\Vert \phi
_{b}^{h}(\varphi,f)\right\Vert _{L^{2}\left(  \mathbb{R}\right)  }%
\lesssim1/h\left\Vert \varphi\right\Vert _{L^{2}\left(  \mathbb{R}\right)  }$.
In the case of $\psi_{\alpha}^{h}(\varphi,f)$, the representation
(\ref{H_gen_eigenfun_k1}) allows similar computations leading to: $\left\Vert
\psi_{\alpha}^{h}(\varphi,f)\right\Vert _{L^{2}\left(  \mathbb{R}\right)
}\lesssim\frac{1}{h^{2}}\left\Vert \varphi\right\Vert _{L^{2}\left(
\mathbb{R}\right)  }$, while for $\alpha=a$, a representation of the type
(\ref{phi_psi_formula}) for the maps (\ref{Phi_alpha_h})-(\ref{Psi_alpha_h})
is obtained by a suitably adaptation of the previous arguments.
\end{proof}

We next consider the operator $\mathcal{W}_{\theta_{1},\theta_{2}}^{h}$
defined by the integral kernel%
\begin{equation}
\mathcal{W}_{\theta_{1},\theta_{2}}^{h}(x,y)=\int_{\mathbb{R}}\frac{dk}{2\pi
h}\,1_{\left[  \Lambda_{1},\Lambda_{2}\right]  }\left(  k^{2}\right)
\psi_{\theta_{1},\theta_{2}}^{h}(x,k,\mathcal{V}^{h})\left(  \psi_{0,0}%
^{h}(x,k,\mathcal{V}^{h})\right)  ^{\ast}\,. \label{W_teta_ker_h}%
\end{equation}
Here the energy subset $\left[  \Lambda_{1},\Lambda_{2}\right]  $ defines a
neighbourhood of the shape resonances and fulfills the property
(\ref{Lambda_0_around}). The spectral projector on $\left[  \Lambda
_{1},\Lambda_{2}\right]  $, next denoted as $\mathcal{P}_{\left[  \Lambda
_{1},\Lambda_{2}\right]  }$, is explicitly given by%
\begin{equation}
\mathcal{P}_{\left[  \Lambda_{1},\Lambda_{2}\right]  }\varphi=\int%
_{\mathbb{R}}\frac{dk}{\left(  2\pi h\right)  ^{1/2}}\,1_{\left[  \Lambda
_{1},\Lambda_{2}\right]  }\left(  k^{2}\right)  \psi_{0,0}^{h}(x,k,\mathcal{V}%
^{h})\left(  \mathcal{F}_{\mathcal{V}^{h}}^{h}\varphi\right)  (k)\,.
\label{Projector}%
\end{equation}

\begin{proposition}
\label{Proposition_W_cont}Let $\mathcal{V}^{h}=V+W^{h}$ satisfy the Condition
\ref{condition_1} and $\left\vert \theta_{j=1,2}\right\vert \lesssim h^{N_{0}%
}$, with $N_{0}\geq4$. If the interval $\left[  \Lambda_{1},\Lambda
_{2}\right]  $ verifies the condition (\ref{Lambda_0_around}) and the lower
bound (\ref{Fermi_golden}) holds, then it exists $\eta>0$ such that: $\left\{
\mathcal{W}_{\theta_{1},\theta_{2}}^{h}\,,\ \theta_{j}\in\mathcal{B}%
_{\eta\,h^{N_{0}}}(0)\,,\ j=1,2\right\}  $ form an analytic family of bounded
operators in $L^{2}(\mathbb{R})$ fulfilling the expansion%
\begin{equation}
\mathcal{W}_{\theta_{1},\theta_{2}}^{h}-\mathcal{P}_{\left[  \Lambda
_{1},\Lambda_{2}\right]  }=\mathcal{O}\left(  h^{N_{0}-2}\right)  \,,
\label{W_teta_exp}%
\end{equation}
in the $\mathcal{L}\left(  L^{2}(\mathbb{R})\right)  $ operator norm.
Moreover, $\mathcal{W}_{\theta_{1},\theta_{2}}^{h}$ maps $L^{2}(\mathbb{R})$
into $D\left(  Q_{\theta_{1},\theta_{2}}^{h}(\mathcal{V}^{h})\right)  $ and it
results%
\begin{equation}
Q_{\theta_{1},\theta_{2}}^{h}(\mathcal{V}^{h})\mathcal{W}_{\theta_{1}%
,\theta_{2}}^{h}\varphi=\mathcal{W}_{\theta_{1},\theta_{2}}^{h}Q_{0,0}%
^{h}(\mathcal{V}^{h})\varphi\,,\qquad\varphi\in L^{2}(\mathbb{R})\,.
\label{W_teta_inter}%
\end{equation}

\end{proposition}

\begin{proof}
The proof is similar to the one given in \cite{Man1} in the case $h=1$ and we
next give a sketch of it. Due to our assumptions, the formula
(\ref{gen_eigenfun_exp_h}) applies provided that $\theta_{j=1,2}\in
\mathcal{B}_{\eta\,h^{N_{0}}}(0)$ for a suitably small $\eta>0$. Then, the
action of $\mathcal{W}_{\theta_{1},\theta_{2}}^{h}$ on $\varphi\in
L^{2}(\mathbb{R})$ writes as%
\begin{align}
\mathcal{W}_{\theta_{1},\theta_{2}}^{h}\varphi &  =\int_{\mathbb{R}}\frac
{dk}{\left(  2\pi h\right)  ^{1/2}}\,1_{\left[  \Lambda_{1},\Lambda
_{2}\right]  }\left(  k^{2}\right)  \psi_{0,0}^{h}(x,k,\mathcal{V}^{h})\left(
\mathcal{F}_{\mathcal{V}^{h}}^{h}\varphi\right)  (k)\nonumber\\
& \nonumber\\
&  +\int_{\mathbb{R}}\frac{dk}{\left(  2\pi h\right)  ^{1/2}}\,1_{\left[
\Lambda_{1},\Lambda_{2}\right]  }\left(  k^{2}\right)  \left[  \mathcal{O}%
\left(  \theta_{2}\right)  G^{\left\vert k\right\vert ,h}\left(
\cdot,b,\mathcal{V}^{h}\right)  +\mathcal{O}\left(  \theta_{2}\right)
G^{\left\vert k\right\vert ,h}(\cdot,a,\mathcal{V}^{h})\right]  \left(
\mathcal{F}_{\mathcal{V}^{h}}^{h}\varphi\right)  (k)\nonumber\\
& \nonumber\\
&  +\int_{\mathbb{R}}\frac{dk}{\left(  2\pi h\right)  ^{1/2}}\,1_{\left[
\Lambda_{1},\Lambda_{2}\right]  }\left(  k^{2}\right)  \left[  \mathcal{O}%
\left(  \theta_{1}\right)  H^{\left\vert k\right\vert ,h}(\cdot,b,\mathcal{V}%
^{h})+\mathcal{O}\left(  \theta_{1}\right)  H^{\left\vert k\right\vert
,h}(\cdot,a,\mathcal{V}^{h})\right]  \left(  \mathcal{F}_{\mathcal{V}^{h}}%
^{h}\varphi\right)  (k)\,. \label{W_teta_exp_h}%
\end{align}
where $\mathcal{O}\left(  \theta_{j}\right)  $, $j=1,2$, here denote bounded
functions of the variables $\left(  k,\theta_{1},\theta_{2},h\right)  $,
holomorphic w.r.t. $\left(  \theta_{1},\theta_{2}\right)  $ and infinitesimal
w.r.t. $\theta_{1}$ or $\theta_{2}$. In what follows we use the identity:
$\mathcal{O}\left(  \theta_{j}\right)  =\theta_{j}\mathcal{O}\left(  1\right)
$ (see the definition \ref{Landau_Notation}) and adopt the notation introduced
in (\ref{Phi_alpha_h})-(\ref{Psi_alpha_h}). Then, (\ref{W_teta_exp_h})
rephrases as%
\begin{equation}
\left(  \mathcal{W}_{\theta_{1},\theta_{2}}^{h}-\mathcal{P}_{\left[
\Lambda_{1},\Lambda_{2}\right]  }\right)  \varphi=\sum_{\alpha=a,b}\left[
\theta_{2}\,\phi_{\alpha}^{h}\left(  \varphi,f_{\alpha}\right)  +\theta
_{1}\,\phi_{\alpha}^{h}\left(  \varphi,f_{\alpha}\right)  \right]  \,,
\label{W_teta_exp_h1}%
\end{equation}
where $f_{\alpha}=1_{\left[  \Lambda_{1},\Lambda_{2}\right]  }\left(
k^{2}\right)  \mathcal{O}\left(  1\right)  $ are bounded functions of $\left(
k,\theta_{1},\theta_{2},h\right)  $, holomorphic w.r.t. $\left(  \theta
_{1},\theta_{2}\right)  $, and supported in $\left\{  k\,\left\vert \ k^{2}%
\in\left[  \Lambda_{1},\Lambda_{2}\right]  \right.  \right\}  $. According to
the condition (\ref{Lambda_0_around}), we have: $\left.  \text{supp }%
f_{\alpha}\subset\Omega_{c}\left(  V\right)  \right.  $ for $c$ suitably
small, and the result of the Lemma \ref{Lemma_Simil_est} applies to the r.h.s.
of (\ref{W_teta_exp_h1}). Using the assumption $\left\vert \theta
_{j=1,2}\right\vert \lesssim h^{N_{0}}$, we conclude that%
\begin{equation}
\left\Vert \mathcal{W}_{\theta_{1},\theta_{2}}^{h}-\mathcal{P}_{\left[
\Lambda_{1},\Lambda_{2}\right]  }\right\Vert _{\mathcal{L}\left(
L^{2}(\mathbb{R}),L^{2}(\mathbb{R})\right)  }=\mathcal{O}\left(  \frac
{\theta_{1}}{h^{2}}\right)  +\mathcal{O}\left(  \frac{\theta_{2}}{h^{2}%
}\right)  =\mathcal{O}\left(  h^{N_{0}-2}\right)  \,. \label{W_teta_h_est}%
\end{equation}

The action of $\mathcal{W}_{\theta_{1},\theta_{2}}^{h}$ over $L^{2}\left(
\mathbb{R}\right)  $ is defined using the expansion (\ref{W_teta_exp_h1}). As
it has been noticed, each one of the maps $\phi_{\alpha}^{h}\left(
\cdot,f_{\alpha}\right)  $ and $\phi_{\alpha}^{h}\left(  \cdot,f_{\alpha
}\right)  $, $\alpha=a,b$, expresses as a superposition of the form (cf.
(\ref{phi_psi_formula}))%
\begin{equation}
1_{\left\{  x\geq\alpha\right\}  }\left(  \mathcal{F}_{\mathcal{V}^{h}}%
^{h}\right)  ^{\ast}\left(  f_{\alpha}\mu_{1}+\mathcal{P\circ}f_{\alpha}%
\mu_{2}\right)  \mathcal{F}_{\mathcal{V}^{h}}^{h}+1_{\left\{  x<\alpha
\right\}  }\left(  \mathcal{F}_{\mathcal{V}^{h}}^{h}\right)  ^{\ast}\left(
f_{\alpha}\mu_{3}+\mathcal{P\circ}f_{\alpha}\mu_{4}\right)  \mathcal{F}%
_{\mathcal{V}^{h}}^{h}\,, \label{phi_psi_formula_1}%
\end{equation}
where the functions $\mu_{i}$, depending on the variables $\left(  k,h\right)
$, behave as $\mathcal{O}\left(  1/h^{2}\right)  $. Since $f_{\alpha}$ is
compactly supported w.r.t. $k$ and holomorphic w.r.t. $\left(  \theta
_{1},\theta_{2}\right)  $, (\ref{phi_psi_formula_1}) defines an holomorphic
family of bounded maps of $L^{2}\left(  \mathbb{R}\right)  $ into
$H^{2}\left(  \mathbb{R}\backslash\left\{  a,b\right\}  \right)  $. This still
holds in the case of $\mathcal{W}_{\theta_{1},\theta_{2}}^{h}$, as it follows
by using (\ref{W_teta_exp_h1}) and the injection $\mathcal{P}_{\left[
\Lambda_{1},\Lambda_{2}\right]  }\left(  L^{2}\left(  \mathbb{R}\right)
\right)  \hookrightarrow H^{2}\left(  \mathbb{R}\right)  $. Moreover,
according to the definitions (\ref{gen_eigenfun_eq}) and (\ref{W_teta_ker_h}),
$\mathcal{W}_{\theta_{1},\theta_{2}}^{h}\varphi$ fulfills the interface
conditions (\ref{B_C_1}); then: $\mathcal{W}_{\theta_{1},\theta_{2}}%
^{h}\varphi\in D\left(  Q_{\theta_{1},\theta_{2}}^{h}(\mathcal{V}^{h})\right)
$ for $\varphi\in L^{2}\left(  \mathbb{R}\right)  $.

Finally, consider the relation (\ref{W_teta_inter}). Using the functional
calculus of $Q_{0,0}^{h}(\mathcal{V}^{h})$, we have: $\left(  \mathcal{F}%
_{\mathcal{V}^{h}}^{h}\left(  Q_{0,0}^{h}(\mathcal{V}^{h})\varphi\right)
\right)  (k)=k^{2}\left(  \mathcal{F}_{\mathcal{V}^{h}}^{h}\varphi\right)
(k)$. Due to the definitions (\ref{gen_Fourier_h})-(\ref{W_teta_ker_h}), the
r.h.s. of (\ref{W_teta_inter}) writes as%
\begin{equation}
\mathcal{W}_{\theta_{1},\theta_{2}}^{h}Q_{0,0}^{h}(\mathcal{V}^{h}%
)\varphi=\int_{\mathbb{R}}\frac{dk}{2\pi h}\,1_{\left[  \Lambda_{1}%
,\Lambda_{2}\right]  }\left(  k^{2}\right)  \psi_{\theta_{1},\theta_{2}}%
^{h}(\cdot,k,\mathcal{V}^{h})k^{2}\left(  \mathcal{F}_{\mathcal{V}^{h}}%
^{h}\varphi\right)  (k)\,, \label{W_teta_inter1}%
\end{equation}
and is well defined for $\varphi\in L^{2}\left(  \mathbb{R}\right)  $. The
same holds for the l.h.s. of (\ref{W_teta_inter}), since $\mathcal{W}%
_{\theta_{1},\theta_{2}}^{h}$ maps $L^{2}\left(  \mathbb{R}\right)  $ into
$D\left(  Q_{\theta_{1},\theta_{2}}^{h}(\mathcal{V}^{h})\right)  $ and,
according to the relation: $\left.  \left(  Q_{\theta_{1},\theta_{2}}%
^{h}(\mathcal{V}^{h})-k^{2}\right)  \psi_{\theta_{1},\theta_{2}}^{h}%
(\cdot,k,\mathcal{V}^{h})=0\right.  $, it follows%
\begin{equation}
Q_{\theta_{1},\theta_{2}}^{h}(\mathcal{V}^{h})\mathcal{W}_{\theta_{1}%
,\theta_{2}}^{h}\varphi=\int_{\mathbb{R}}\frac{dk}{2\pi h}\,1_{\left[
\Lambda_{1},\Lambda_{2}\right]  }\left(  k^{2}\right)  \psi_{\theta_{1}%
,\theta_{2}}^{h}(\cdot,k,\mathcal{V}^{h})k^{2}\left(  \mathcal{F}%
_{\mathcal{V}^{h}}^{h}\varphi\right)  (k)\,. \label{W_teta_inter2}%
\end{equation}

\end{proof}

\subsection{The quantum dynamics in the small-$h$ asymptotics}

We next consider the dynamical system generated by $iQ_{\theta_{1},\theta_{2}%
}^{h}(\mathcal{V}^{h})$. Let us denote with $H_{\left[  \Lambda_{1}%
,\Lambda_{2}\right]  }$ the projection space%
\begin{equation}
H_{\left[  \Lambda_{1},\Lambda_{2}\right]  }=\mathcal{P}_{\left[  \Lambda
_{1},\Lambda_{2}\right]  }\left(  L^{2}\left(  \mathbb{R}\right)  \right)  \,,
\label{Projection_space}%
\end{equation}
and with $X$ the subset of $D\left(  Q_{\theta_{1},\theta_{2}}^{h}%
(\mathcal{V}^{h})\right)  $ determined by the image%
\begin{equation}
X=\mathcal{W}_{\theta_{1},\theta_{2}}^{h}\left(  H_{\left[  \Lambda
_{1},\Lambda_{2}\right]  }\right)  \,. \label{X}%
\end{equation}
As it follows from (\ref{W_teta_exp}), the restriction $\left.  \mathcal{W}%
_{\theta_{1},\theta_{2}}^{h}\right\vert _{H_{\left[  \Lambda_{1},\Lambda
_{2}\right]  }}$ is an invertible map from $H_{\left[  \Lambda_{1},\Lambda
_{2}\right]  }$ to $X$ and the inverse map allows the expansion%
\begin{equation}
\left(  \left.  \mathcal{W}_{\theta_{1},\theta_{2}}^{h}\right\vert
_{H_{\left[  \Lambda_{1},\Lambda_{2}\right]  }}\right)  ^{-1}-\mathcal{P}%
_{\left[  \Lambda_{1},\Lambda_{2}\right]  }=\mathcal{O}\left(  h^{N_{0}%
-2}\right)  \,, \label{W_teta_exp_h2}%
\end{equation}
provided that the parameters $\theta_{j}$, $h$ and the potential
$\mathcal{V}^{h}$ fulfill the assumptions of the Proposition
\ref{Proposition_W_cont} with $h_{0}$ small enough. Under this prescription,
we introduce the operator%
\begin{equation}
e^{-itQ_{\theta_{1},\theta_{2}}^{h}(\mathcal{V}^{h})}=\mathcal{W}_{\theta
_{1},\theta_{2}}^{h}e^{-itQ_{0,0}^{h}(\mathcal{V}^{h})}\left(  \left.
\mathcal{W}_{\theta_{1},\theta_{2}}^{h}\right\vert _{H_{\left[  \Lambda
_{1},\Lambda_{2}\right]  }}\right)  ^{-1}\,, \label{propagator_teta}%
\end{equation}
defined on $X$. Since the propagators $e^{-itQ_{0,0}^{h}(\mathcal{V}^{h})}$
form a strongly continuous group of unitary maps of $H_{\left[  \Lambda
_{1},\Lambda_{2}\right]  }$ into itself, and $\mathcal{W}_{\theta_{1}%
,\theta_{2}}^{h}$ is an analytic family w.r.t. $\left(  \theta_{1},\theta
_{2}\right)  $, the modified propagator $e^{-itQ_{\theta_{1},\theta_{2}}%
^{h}(\mathcal{V}^{h})}$ has the same regularity w.r.t. the time and the
parameters $\left(  \theta_{1},\theta_{2}\right)  $, defining a
strongly-continuous flow on $X$. From the identity: $i\partial_{t}%
e^{-itQ_{0,0}^{h}(\mathcal{V}^{h})}\psi=Q_{0,0}^{h}(\mathcal{V}^{h}%
)e^{-itQ_{0,0}^{h}(\mathcal{V}^{h})}\psi$, holding in $L^{2}\left(
\mathbb{R}\right)  $ for any $\psi\in H^{2}\left(  \mathbb{R}\right)  $, it
follows%
\begin{equation}
i\partial_{t}\left(  e^{-itQ_{\theta_{1},\theta_{2}}^{h}(\mathcal{V}^{h}%
)}u\right)  =Q_{\theta_{1},\theta_{2}}^{h}(\mathcal{V}^{h})e^{-itQ_{\theta
_{1},\theta_{2}}^{h}(\mathcal{V}^{h})}u\,,\qquad u\in X\,,
\label{propagator_teta_eq}%
\end{equation}
(recall that $H_{\left[  \Lambda_{1},\Lambda_{2}\right]  }\subset H^{2}\left(
\mathbb{R}\right)  $). Then $e^{-itQ_{\theta_{1},\theta_{2}}^{h}%
(\mathcal{V}^{h})}$ identifies with the quantum dynamical system generated by
$iQ_{\theta_{1},\theta_{2}}^{h}(\mathcal{V}^{h})$. Let $\psi\in H_{\left[
\Lambda_{1},\Lambda_{2}\right]  }$; from the definition (\ref{propagator_teta}%
), follows%
\begin{equation}
e^{-itQ_{\theta_{1},\theta_{2}}^{h}(\mathcal{V}^{h})}\mathcal{W}_{\theta
_{1},\theta_{2}}^{h}\psi=\mathcal{W}_{\theta_{1},\theta_{2}}^{h}%
e^{-itQ_{0,0}^{h}(\mathcal{V}^{h})}\psi\,.
\end{equation}
In the assumptions of the Proposition \ref{Proposition_W_cont}, the expansion
(\ref{W_teta_exp}) holds and using the relation: $\left.  \mathcal{P}_{\left[
\Lambda_{1},\Lambda_{2}\right]  }e^{-itQ_{0,0}^{h}(\mathcal{V}^{h}%
)}\mathcal{P}_{\left[  \Lambda_{1},\Lambda_{2}\right]  }=e^{-itQ_{0,0}%
^{h}(\mathcal{V}^{h})}\mathcal{P}_{\left[  \Lambda_{1},\Lambda_{2}\right]
}\right.  $, we get%
\begin{equation}
\left(  e^{-itQ_{\theta_{1},\theta_{2}}^{h}(\mathcal{V}^{h})}\mathcal{W}%
_{\theta_{1},\theta_{2}}^{h}-e^{-itQ_{0,0}^{h}(\mathcal{V}^{h})}\right)
\mathcal{P}_{\left[  \Lambda_{1},\Lambda_{2}\right]  }=\mathcal{R}^{h}\left(
t,\theta_{1},\theta_{2}\right)  \,, \label{semigroup_exp1}%
\end{equation}
with%
\begin{equation}
\sup_{t\in\mathbb{R\,}}\left\Vert \mathcal{R}^{h}\left(  t,\theta_{1}%
,\theta_{2}\right)  \right\Vert _{\mathcal{L}\left(  L^{2}(\mathbb{R})\right)
}=\mathcal{O}\left(  h^{N_{0}-2}\right)  \,. \label{semigroup_est_1}%
\end{equation}

The modified dynamics can be defined directly on $H_{\left[  \Lambda
_{1},\Lambda_{2}\right]  }$. This point is considered in the next theorem.

\begin{theorem}
\label{Theorem_1}Let $h\in\left(  0,h_{0}\right]  $ and $\left\vert
\theta_{j=1,2}\right\vert \lesssim h^{N_{0}}$, with $N_{0}\geq4$. If the
$\mathcal{V}^{h}$ and $\left[  \Lambda_{1},\Lambda_{2}\right]  $ satisfy the
assumptions of the Proposition \ref{Proposition_W_cont}, then $iQ_{\theta
_{1},\theta_{2}}^{h}(\mathcal{V}^{h})$ generates a strongly continuous group
of bounded operators on $H_{\left[  \Lambda_{1},\Lambda_{2}\right]  }$. For a
fixed $t$, $e^{-itQ_{\theta_{1},\theta_{2}}^{h}(\mathcal{V}^{h})}$ is analytic
w.r.t. $\left(  \theta_{1},\theta_{2}\right)  $ and the expansion%
\begin{equation}
\left(  e^{-itQ_{\theta_{1},\theta_{2}}^{h}(\mathcal{V}^{h})}-e^{-itQ_{0,0}%
^{h}(\mathcal{V}^{h})}\right)  \mathcal{P}_{\left[  \Lambda_{1},\Lambda
_{2}\right]  }=\mathcal{\tilde{R}}^{h}\left(  t,\theta_{1},\theta_{2}\right)
\,, \label{semigroup_exp2}%
\end{equation}
holds with
\begin{equation}
\sup_{t\in\mathbb{R\,}}\left\Vert \mathcal{\tilde{R}}^{h}\left(  t,\theta
_{1},\theta_{2}\right)  \right\Vert _{\mathcal{L}\left(  L^{2}(\mathbb{R}%
)\right)  }=\mathcal{O}\left(  h^{N_{0}-2}\right)  \,. \label{semigroup_est}%
\end{equation}

\end{theorem}

\begin{proof}
Using the expansion (\ref{W_teta_exp_h1}), we get the formal identity%
\begin{equation}
e^{-itQ_{\theta_{1},\theta_{2}}^{h}(\mathcal{V}^{h})}\mathcal{P}_{\left[
\Lambda_{1},\Lambda_{2}\right]  }\varphi=e^{-itQ_{\theta_{1},\theta_{2}}%
^{h}(\mathcal{V}^{h})}\mathcal{W}_{\theta_{1},\theta_{2}}^{h}\varphi
-e^{-itQ_{\theta_{1},\theta_{2}}^{h}(\mathcal{V}^{h})}\sum_{\alpha=a,b}\left[
\theta_{2}\,\phi_{\alpha}^{h}\left(  \varphi,f_{\alpha}\right)  +\theta
_{1}\,\phi_{\alpha}^{h}\left(  \varphi,f_{\alpha}\right)  \right]  \,.
\label{propagator_teta_exp1}%
\end{equation}
where $\phi_{\alpha}^{h}(\cdot,f_{\alpha})$ and $\psi_{\alpha}^{h}%
(\cdot,f_{\alpha})$ are defined by (\ref{Phi_alpha_h})-(\ref{Psi_alpha_h})
with the bounded functions $f_{\alpha}$ depending on $\left(  k,\theta
_{1},\theta_{2},h\right)  $, holomorphic w.r.t. $\left(  \theta_{1},\theta
_{2}\right)  $ and supported in $\left\{  k\,\left\vert \ k^{2}\in\left[
\Lambda_{1},\Lambda_{2}\right]  \right.  \right\}  $. The first term at the
r.h.s. of (\ref{propagator_teta_exp1}) describes the action of
$e^{-itQ_{\theta_{1},\theta_{2}}^{h}(\mathcal{V}^{h})}$ on $\mathcal{W}%
_{\theta_{1},\theta_{2}}^{h}\varphi\in X$. This is properly defined according
to (\ref{propagator_teta}). For the second term, we notice that $Q_{\theta
_{1},\theta_{2}}^{h}(\mathcal{V}^{h})$ is a restriction of the operator
$Q^{h}(\mathcal{V}^{h})$ (see eq. (\ref{Q})), while the functions $G^{k,h}$,
$H^{k,h}$, appearing in the definitions of the maps $\phi_{\alpha}^{h}%
(\cdot,f_{\alpha})$ and $\psi_{\alpha}^{h}(\cdot,f_{\alpha})$, are the limits,
as $z\rightarrow k^{2}\pm i0$, of functions in $\ker\left(  Q^{h}%
(\mathcal{V}^{h})-z\right)  $; then $Q_{\theta_{1},\theta_{2}}^{h}%
(\mathcal{V}^{h})$ formally acts on $G^{k,h}$, $H^{k,h}$ as%
\begin{equation}
Q_{\theta_{1},\theta_{2}}^{h}(\mathcal{V}^{h})\Psi=Q^{h}(\mathcal{V}^{h}%
)\Psi=k^{2}\Psi,\qquad\Psi=G^{k,h},H^{k,h}\,. \label{propagator_teta_exp1_1}%
\end{equation}
Using (\ref{Phi_alpha_h})-(\ref{Psi_alpha_h}) and the identity: $\left.
\left(  \mathcal{F}_{\mathcal{V}^{h}}^{h}\left(  e^{-itQ_{0,0}^{h}%
(\mathcal{V}^{h})}\varphi\right)  \right)  (k)=e^{-ik^{2}t}\left(
\mathcal{F}_{\mathcal{V}^{h}}^{h}\left(  \varphi\right)  \right)  (k)\right.
$, we get%
\begin{equation}
e^{-itQ_{\theta_{1},\theta_{2}}^{h}(\mathcal{V}^{h})}\phi_{\alpha}^{h}%
(\varphi)=\phi_{\alpha}^{h}(e^{-itQ_{0,0}^{h}(\mathcal{V}^{h})}\varphi
,f_{\alpha})\,,\quad e^{-itQ_{\theta_{1},\theta_{2}}^{h}(\mathcal{V}^{h})}%
\psi_{\alpha}^{h}(\varphi)=\psi_{\alpha}^{h}(e^{-itQ_{0,0}^{h}(\mathcal{V}%
^{h})}\varphi,f_{\alpha})\,,\qquad\alpha\in\left\{  a,b\right\}  \,,
\label{Psi_h_evolution}%
\end{equation}
From the Lemma \ref{Lemma_Simil_est}, it results%
\begin{equation}
\left\Vert e^{-itQ_{\theta_{1},\theta_{2}}^{h}(\mathcal{V}^{h})}\phi_{\alpha
}^{h}(\varphi,f_{\alpha})\right\Vert _{L^{2}\left(  \mathbb{R}\right)
}+\left\Vert e^{-itQ_{\theta_{1},\theta_{2}}^{h}(\mathcal{V}^{h})}\psi
_{\alpha}^{h}(\varphi,f_{\alpha})\right\Vert _{L^{2}\left(  \mathbb{R}\right)
}\lesssim\mathcal{O}\left(  \frac{1}{h^{2}}\right)  \left\Vert \varphi
\right\Vert _{L^{2}\left(  \mathbb{R}\right)  }\,.
\label{propagator_teta_est1}%
\end{equation}
Then, $e^{-itQ_{\theta_{1},\theta_{2}}^{h}(\mathcal{V}^{h})}\phi_{\alpha}%
^{h}(\cdot,f_{\alpha})$ and $e^{-itQ_{\theta_{1},\theta_{2}}^{h}%
(\mathcal{V}^{h})}\psi_{\alpha}^{h}(\cdot,f_{\alpha})$ define strongly
continuous operator-valued functions of $t$ having the group property.
Moreover, for any fixed $t$, these are analytic families w.r.t. $\left(
\theta_{1},\theta_{2}\right)  $, since $f_{\alpha}$ are holomorphic in these
variables. It follows that $e^{-itQ_{\theta_{1},\theta_{2}}^{h}(\mathcal{V}%
^{h})}\mathcal{P}_{\left[  \Lambda_{1},\Lambda_{2}\right]  }$ forms a strongly
continuous-in-time group of bounded operators, holomorphic w.r.t. $\left(
\theta_{1},\theta_{2}\right)  $.

Using the expansion (\ref{semigroup_exp1}) and the definition
(\ref{propagator_teta_exp1}), we get%
\[
\left(  e^{-itQ_{\theta_{1},\theta_{2}}^{h}(\mathcal{V}^{h})}-e^{-itQ_{0,0}%
^{h}(\mathcal{V}^{h})}\right)  \mathcal{P}_{\left[  \Lambda_{1},\Lambda
_{2}\right]  }=\mathcal{R}^{h}\left(  t,\theta_{1},\theta_{2}\right)
-e^{-itQ_{\theta_{1},\theta_{2}}^{h}(\mathcal{V}^{h})}\sum_{\alpha=a,b}\left[
\mathcal{O}\left(  \theta_{2}\right)  \phi_{\alpha}^{h}(\varphi)+\mathcal{O}%
\left(  \theta_{1}\right)  \psi_{\alpha}^{h}(\varphi)\right]  \,.
\]
Then, taking into account the estimates (\ref{semigroup_est_1}),
(\ref{propagator_teta_est1}) and the assumption $\left\vert \theta
_{j=1,2}\right\vert \lesssim h^{N_{0}}$, the expansion (\ref{semigroup_exp2}%
)-(\ref{semigroup_est}) follows.
\end{proof}

\subsection{Conclusions and further perspectives}

Double scale systems with quantum wells in a semiclassical island have been
adopted as models for the mathematical description of resonant
heterostructures, like tunnelling diods. In this framework, the non-linear
effects connected to the accumulation of the charges inside the wells are
taken into account by using non-linear Schr\"{o}dinger-Poisson operators (see
e.g. in \cite{BNP1}). When the quantum scale is small compared to the
barrier's length, it is a general belief that the quantum transport in these
systems is driven by a finite number of resonant states related to shape
resonances. A rigorous mathematical approach to this problem has been provided
with in \cite{BNP1}, \cite{BNP2}, where far-from-equilibrium steady states are
considered. In the non-stationary case, the time evolution of a resonance
$z_{res}$ is expected to generate a non-linear perturbation of the (linear)
background potential, whose variations in time are characterized by the time
scale $\varepsilon=1/\operatorname{Im}z_{res}$. If the conditions
(\ref{Fermi_golden}) hold, we get: $\varepsilon\sim e^{-\frac{1}{h}}$ and the
non-linear part of the Hamiltonian become an adiabatic perturbation for small
$h$ (we refer to \cite{PrSj} and \cite{PrSj1}). In this connection, adiabatic
approximations for the evolution of resonant states appear to play a central
role in the study of the non linear transport problem. In \cite{FMN2}, a
version of the adiabatic theorem has been proved for shape resonances in the
regime of quantum wells in a semiclassical island with artificial interface
conditions. The result of the Theorem \ref{Theorem_1} justifies the use of
these operators in the modelling of resonant heterostructures providing in
this way with a useful framework for the application of the method introduced
in \cite{FMN2}.

Our work generalizes an analogous investigation developed in \cite{Man1} in
the case of a purely quantum scaling. Two main restrictions, with respect to
the $h$-independent case, appear in our framework. The first one concerns the
expansion (\ref{semigroup_exp2}) in the Theorem \ref{Theorem_1}, which is
limited to an energy subspace corresponding to a cluster of shape resonances.
The second one is related to the interface parameters $\theta_{j=1,2}$;
namely, in order to prove the existence of the intertwining operators
$\mathcal{W}_{\theta_{1},\theta_{2}}^{h}$, these are assumed to be
polinomially small functions of the quantum scale $h$. It is worthwhile to
notice that none of these conditions constitutes a real obstruction in the
perspective of the applications. Indeed, the relevant initial states,
describing charge careers in models of quantum transport through resonant
heterostructures, have energies close to the resonances; then, they roughly
belong to the spectral subspace $H_{\left[  \Lambda_{1},\Lambda_{2}\right]  }$
introduced in (\ref{Projection_space}) (see e.g. in (\cite{FMN3})). For the
second point, let us recall that the adiabatic theorem obtained in \cite{FMN2}
(see Theorem $7.1$ in \cite{FMN2}) applies to $Q_{\theta,3\theta}^{h}\left(
\mathcal{V}^{h}\right)  $ for a potential $\mathcal{V}^{h}$ having the scaling
introduced in the Section \ref{Section_Resonances} and with: $\theta
=ch^{N_{0}}$, for some $N_{0}\in\mathbb{N}$. Thus, this Theorem and the result
presented in our Theorem \ref{Theorem_1} apply in the same framework.

Let consider a time-dependent potential $\mathcal{V}^{h}\left(  t\right)  $
such that the Condition \ref{condition_1} holds for any fixed $t$. If the
energy interval $\left[  \Lambda_{1},\Lambda_{2}\right]  $ verifies
(\ref{Lambda_0_around}) uniformly w.r.t. $t$, the corresponding non-autonomous
Hamiltonian $Q_{\theta_{1},\theta_{2}}^{h}(\mathcal{V}^{h}\left(  t\right)  )$
generates a time-dependent cluster of shape resonances, exponentially close to
the continuous spectrum as $h\rightarrow0$, with energies embedded in $\left[
\Lambda_{1},\Lambda_{2}\right]  $ \ for any time. In the perspective of
investigating the adiabatic evolution of shape resonances through modified
Schr\"{o}dinger operators, an extension of the previous analysis is needed:
the aim is to obtain an uniform-in-time control for the distance between the
modified dynamics and the unitary one in the limit $h\rightarrow0$, being the
interface parameters $\theta_{j=1,2}$ polynomially small w.r.t. $h$ and the
potential $\mathcal{V}^{h}\left(  t\right)  $ an assigned function of the time.

A classical tool, to define the quantum dynamical system generated by a
non-autonomous Hamiltonian, consists in using a sequence of step functions
defined as products of propagators (cf. \cite{Yosh}) to approximate the
evolution operator. This approach requires the continuity in time and the
\emph{stability} of the generator of the dynamics (aside from the existence of
an \emph{admissible space}, see the related definitions in \cite{Kato1}).
Under these assumptions, the existence of the dynamics and its regularity
w.r.t. to the time, depending on the initial state, have been established
(e.g. in \cite{Kato1}). In general, the stability is easily deduced when the
instantaneous Hamiltonian generates a dynamical system of contractions. In the
other cases (i.e. when the semigroup is only uniformly bounded in time) a
direct check of this property is a rather difficult task. This is actually the
case of the operators $iQ_{\theta_{1},\theta_{2}}^{h}$. A detailed analysis of
this point is needed in order to prove that the expansion
(\ref{semigroup_exp2}) still holds in the non-autonomous case.

\bigskip

\appendix

\section{\label{App_expest}Exponential decay estimates.}

Next we give Agmon-type estimates for the functions $\mathcal{G}^{z,h}$,
$\mathcal{H}^{z,h}$, $j=0,1$ (defined as solutions of (\ref{Green_eq}%
)-(\ref{D_Green_eq})), and the resolvent $\left(  \mathcal{P}_{z}^{h}\left(
V\right)  -z\right)  ^{-1}$ in the 'filled well' case ($W^{h}=0$). Recall the
Agmon identity (see \cite{Agm},\cite{Hel})%
\begin{gather}
\int_{a}^{b}u_{2}^{\ast}e^{2\frac{f}{h}}\left(  -h^{2}\partial_{x}%
^{2}+\mathcal{V}-z\right)  u_{1}\,dx=\int_{a}^{b}hv_{2}^{\prime\ast}%
hv_{1}^{\prime}\,dx+\int_{a}^{b}\left(  \mathcal{V}-z-f^{\prime2}\right)
v_{2}^{\ast}v_{1}\,dx\nonumber\\
\nonumber\\
+\int_{a}^{b}hf^{\prime}\left(  v_{2}^{\ast}v_{1}^{\prime}-v_{2}^{\prime\ast
}v_{1}\right)  \,dx+h^{2}\left(  e^{2\frac{f(a)}{h}}u_{2}^{\ast}u_{1}^{\prime
}(a)-e^{2\frac{f(b)}{h}}u_{2}^{\ast}u_{1}^{\prime}(b)\right)  \,,\qquad
v_{j}=e^{\frac{f}{h}}u_{j}\,, \label{Agmon_Id}%
\end{gather}
holding for $u_{j}\in H^{1}\left(  \left[  a,b\right]  \right)  $, $j=1,2$,
$z\in\mathbb{C}$, $\mathcal{V}\in L^{\infty}\left(  \left(  a,b\right)
\right)  $ and $f\in W^{1,\infty}\left(  \left(  a,b\right)  \right)  $.

\begin{lemma}
\label{Lemma_Agmon}Consider the problem
\begin{equation}
\left\{
\begin{array}
[c]{l}%
\left(  -h^{2}\partial_{x}^{2}+\mathcal{V}-\zeta^{2}\right)  u=0\,,\qquad
\text{in }\left(  a,b\right)  \,,\\
\\
\left[  h\partial_{x}+i\zeta\right]  u(a)=\gamma_{a}\,,\quad\left[
h\partial_{x}-i\zeta\right]  u(b)=\gamma_{b}\,,
\end{array}
\right.  \label{Ag_eq}%
\end{equation}
where: $\mathcal{V}\in L^{\infty}\left(  \left(  a,b\right)  ,\mathbb{R}%
\right)  $, $\gamma_{a},\gamma_{b}\in\mathbb{C}$, and $h>0$ is suitably small.
Assume $\zeta\in\overline{\mathbb{C}^{+}}$ such that: $\mathcal{V}%
-\operatorname{Re}\zeta^{2}>c$ for some $c>0$ and $\varphi\left(
\cdot\right)  =d_{Ag}\left(  \cdot,K,\mathcal{V},\operatorname{Re}\zeta
^{2}\right)  $, being $K$ any compact in $\left[  a,b\right]  $; the solution
of (\ref{Ag_eq}) fulfills the estimate%
\begin{equation}
h^{\frac{1}{2}}\sup_{\left[  a,b\right]  }\left\vert e^{\frac{\varphi}{h}%
}u\right\vert +\left\Vert he^{\frac{\varphi}{h}}u^{\prime}\right\Vert
_{L^{2}\left(  \left[  a,b\right]  \right)  }+\left\Vert e^{\frac{\varphi}{h}%
}u\right\Vert _{L^{2}\left(  \left[  a,b\right]  \right)  }\leq C_{a,b,c}%
\frac{1}{h^{\frac{1}{2}}}\left(  e^{\frac{\varphi(a)}{h}}\left\vert \gamma
_{a}\right\vert +e^{\frac{\varphi(b)}{h}}\left\vert \gamma_{b}\right\vert
\right)  \,, \label{Agmon_est1}%
\end{equation}
with $C_{a,b,c}>0$ possibly depending on the data.
\end{lemma}

\begin{proof}
Let $\varphi_{h}\left(  \cdot\right)  =d_{Ag}\left(  \cdot,K,\mathcal{V}%
-h,\operatorname{Re}\zeta^{2}\right)  $; using the identity (\ref{Agmon_Id})
with $u_{1}=u_{2}=u$ and $f=\varphi_{h}$, we get
\begin{equation}
\left\Vert hv^{\prime}\right\Vert _{L^{2}\left(  \left[  a,b\right]  \right)
}^{2}+\int_{a}^{b}\left(  \mathcal{V}-\zeta^{2}-\varphi_{h}^{\prime2}\right)
\left\vert v\right\vert ^{2}\,dx+2ih\int_{a}^{b}\varphi_{h}^{\prime
}\operatorname{Im}\left(  v^{\ast}v^{\prime}\right)  \,dx+h^{2}\left(
e^{2\frac{\varphi_{h}(a)}{h}}u^{\ast}u^{\prime}(a)-e^{2\frac{\varphi_{h}%
(b)}{h}}u^{\ast}u^{\prime}(b)\right)  =0\,, \label{Agmon_est0}%
\end{equation}
with: $v=e^{\frac{\varphi_{h}}{h}}u$. For $h$ small, the condition:
$\mathcal{V}-\operatorname{Re}\zeta^{2}>c$ entails $\mathcal{V}%
-h-\operatorname{Re}\zeta^{2}>0$ and
\begin{equation}
\varphi_{h}^{\prime}=\left\{
\begin{array}
[c]{lll}%
0\,, &  & \text{in }K\,,\\
&  & \\
\left(  \mathcal{V}-h-\operatorname{Re}\zeta^{2}\right)  ^{\frac{1}{2}}\,, &
& \text{in }\left[  a,b\right]  \backslash K\,.
\end{array}
\right.  \label{phi_h_prime}%
\end{equation}
Then, the real part of (\ref{Agmon_est0}) reads as%
\begin{equation}
\left\Vert hv^{\prime}\right\Vert _{L^{2}\left(  \left[  a,b\right]  \right)
}^{2}+\int_{K}\left(  \mathcal{V}-\operatorname{Re}\zeta^{2}\right)
\left\vert v\right\vert ^{2}\,dx+h\int_{\left[  a,b\right]  \backslash
K}\left\vert v\right\vert ^{2}\,dx+h^{2}\operatorname{Re}\left(
e^{2\frac{\varphi_{h}(a)}{h}}u^{\ast}u^{\prime}(a)-e^{2\frac{\varphi_{h}%
(b)}{h}}u^{\ast}u^{\prime}(b)\right)  =0\,.
\end{equation}
Taking into account the boundary conditions in (\ref{Ag_eq}), our assumptions
($\mathcal{V}-\operatorname{Re}\zeta^{2}>c$ and $\operatorname{Im}\zeta\geq0$)
imply%
\begin{equation}
\left\Vert hv^{\prime}\right\Vert _{L^{2}\left(  \left[  a,b\right]  \right)
}^{2}+h\left\Vert v\right\Vert _{L^{2}\left(  \left[  a,b\right]  \right)
}^{2}\leq h\operatorname{Re}\left(  e^{2\frac{\varphi_{h}(a)}{h}}\left\vert
u^{\ast}(a)\right\vert \left\vert \gamma_{a}\right\vert -e^{2\frac{\varphi
_{h}(b)}{h}}\left\vert u^{\ast}(b)\right\vert \left\vert \gamma_{b}\right\vert
\right)  \,. \label{Agmon_est1_0}%
\end{equation}
Using (\ref{h_GN}), this leads to%
\begin{equation}
\left\Vert hv^{\prime}\right\Vert _{L^{2}\left(  \left[  a,b\right]  \right)
}^{2}+h\left\Vert v\right\Vert _{L^{2}\left(  \left[  a,b\right]  \right)
}^{2}\leq C_{b-a}h^{\frac{1}{2}}\left\Vert v\right\Vert _{H^{1,h}\left(
\left[  a,b\right]  \right)  }\operatorname{Re}\left(  e^{2\frac{\varphi
_{h}(a)}{h}}\left\vert \gamma_{a}\right\vert -e^{2\frac{\varphi_{h}(b)}{h}%
}\left\vert \gamma_{b}\right\vert \right)  \,, \label{Agmon_est1_1}%
\end{equation}
which yields%
\[
\left\Vert v\right\Vert _{H^{1,h}\left(  \left[  a,b\right]  \right)  }%
\leq\frac{C_{a,b,c}}{h^{\frac{1}{2}}}\left(  e^{\frac{\varphi_{h}(a)}{h}%
}\left\vert \gamma_{a}\right\vert +e^{\frac{\varphi_{h}(b)}{h}}\left\vert
\gamma_{b}\right\vert \right)  \,.
\]
Due to our assumptions, it results: $\left\vert \varphi-\varphi_{h}\right\vert
=\mathcal{O}\left(  h\right)  $, uniformly w.r.t. $x\in\left[  a,b\right]  $.
Hence, $e^{\frac{\varphi}{h}}\sim e^{\frac{\varphi_{h}}{h}}$ as $h\rightarrow
0$ and the previous inequality rephrases as%
\begin{equation}
\left\Vert v\right\Vert _{H^{1,h}\left(  \left[  a,b\right]  \right)  }%
\leq\frac{C_{a,b,c}}{h^{\frac{1}{2}}}\left(  e^{\frac{\varphi(a)}{h}%
}\left\vert \gamma_{a}\right\vert +e^{\frac{\varphi(b)}{h}}\left\vert
\gamma_{b}\right\vert \right)  \,. \label{Agmon_est3}%
\end{equation}
Finally, the estimate (\ref{Agmon_est1}) is deduced from (\ref{Agmon_est3}) by
taking into account (\ref{h_GN}) and (\ref{phi_h_prime}).
\end{proof}

When the operator $Q_{0,0}^{h}(V)$ is defined with a potential $V$ bounded
from below and compactly supported in $\left[  a,b\right]  $, the related
Green's functions $\mathcal{G}^{z,h}\left(  \cdot,y,V\right)  $ and
$\mathcal{H}^{z,h}\left(  \cdot,y,V\right)  $ solve equations of the type
(\ref{Ag_eq}). Hence, exponential decay estimates are obtained as a
straightforward consequence of the previous Lemma.

\begin{lemma}
\label{Lemma_H1_est}Let $\mathcal{V}\in L^{\infty}\left(  \mathbb{R}%
,\mathbb{R}\right)  $ be supported on $\left[  a,b\right]  $ and assume
$\zeta\in\overline{\mathbb{C}^{+}}$ such that: $\mathcal{V}-\operatorname{Re}%
\zeta^{2}>c$. For $y\in\left\{  a,b\right\}  $, the functions $\mathcal{G}%
^{\zeta^{2},h}\left(  \cdot,y,\mathcal{V}\right)  $ and $\mathcal{H}%
^{\zeta^{2},h}\left(  \cdot,y,\mathcal{V}\right)  $, defined in
(\ref{Green_eq})-(\ref{D_Green_eq}), allow the estimates%
\[%
\begin{array}
[c]{lll}%
i) &  & h^{\frac{1}{2}}\sup_{\left[  a,b\right]  }\left\vert e^{\frac{\varphi
}{h}}\mathcal{G}^{\zeta^{2},h}\left(  \cdot,y,\mathcal{V}\right)  \right\vert
+\left\Vert he^{\frac{\varphi}{h}}\partial_{1}\mathcal{G}^{\zeta^{2},h}\left(
\cdot,y,\mathcal{V}\right)  \right\Vert _{L^{2}\left(  \left[  a,b\right]
\right)  }+\left\Vert e^{\frac{\varphi}{h}}\mathcal{G}^{\zeta^{2},h}\left(
\cdot,y,\mathcal{V}\right)  \right\Vert _{L^{2}\left(  \left[  a,b\right]
\right)  }\leq\frac{C_{a,b,c}}{h^{\frac{3}{2}}}e^{\frac{\varphi(y)}{h}}\,,\\
&  & \\
ii) &  & h^{\frac{1}{2}}\sup_{\left[  a,b\right]  }\left\vert e^{\frac
{\varphi}{h}}\mathcal{H}^{\zeta^{2},h}\left(  \cdot,y,\mathcal{V}\right)
\right\vert +\left\Vert he^{\frac{\varphi}{h}}\partial_{1}\mathcal{H}%
^{\zeta^{2},h}\left(  \cdot,y,\mathcal{V}\right)  \right\Vert _{L^{2}\left(
\left[  a,b\right]  \right)  }+\left\Vert e^{\frac{\varphi}{h}}\mathcal{H}%
^{\zeta^{2},h}\left(  \cdot,y,\mathcal{V}\right)  \right\Vert _{L^{2}\left(
\left[  a,b\right]  \right)  }\leq\frac{C_{a,b,c}}{h^{\frac{5}{2}}}%
e^{\frac{\varphi(y)}{h}}\left\vert \zeta\right\vert \,,
\end{array}
\]
with: $\varphi\left(  \cdot\right)  =d_{Ag}\left(  \cdot,K,\mathcal{V}%
,\operatorname{Re}\zeta^{2}\right)  $, where $K$ is any compact in $\left[
a,b\right]  $, $C_{a,b,c}>0$ possibly depending on the data, and $h>0$ small.
\end{lemma}

\begin{proof}
We explicitly consider the case $y=a$; the result for $y=b$ follows from a
similar argument. According to (\ref{Green_eq})-(\ref{D_Green_eq}),
$\mathcal{G}^{\zeta^{2},h}\left(  \cdot,a,\mathcal{V}\right)  $ and
$\mathcal{H}^{\zeta^{2},h}\left(  \cdot,a,\mathcal{V}\right)  $ are the
solutions of equations of type-(\ref{Ag_eq}) with: $\gamma_{a}=-\frac{1}{h}$
and $\gamma_{b}=0$ in the case of $\mathcal{G}^{\zeta^{2},h}$, and:
$\gamma_{a}=\frac{i\zeta}{h^{2}}$ and $\gamma_{b}=0$ in the case of
$\mathcal{H}^{\zeta^{2},h}$. From our assumptions, the Lemma \ref{Lemma_Agmon}
applies, leading in this case to the inequalities $(i)$ and $(ii)$.
\end{proof}

Similar bounds, holding in the case of generalized eigenfunctions, have been
given in Proposition 4.4 of \cite{FMN2}. Let us recall here this result.

\begin{lemma}
\label{Lemma_H1_est1}Let $\mathcal{V}\in L^{\infty}\left(  \mathbb{R}%
,\mathbb{R}\right)  $ be supported on $\left[  a,b\right]  $ and assume
$k\in\mathbb{R}$ such that: $\mathcal{V}-\operatorname{Re}k^{2}>c$. For
$y\in\left\{  a,b\right\}  $, the functions $\psi_{0,0}^{h}(\cdot
,k,\mathcal{V})$, defined in (\ref{gen_eigenfun_eq})-(\ref{gen_eigenfun_ext2}%
), allow the estimates%
\[
h^{\frac{1}{2}}\sup_{\left[  a,b\right]  }\left\vert e^{\frac{\varphi}{h}}%
\psi_{0,0}^{h}(\cdot,k,\mathcal{V})\right\vert +\left\Vert he^{\frac{\varphi
}{h}}\partial_{1}\psi_{0,0}^{h}(\cdot,k,\mathcal{V})\right\Vert _{L^{2}\left(
\left[  a,b\right]  \right)  }+\left\Vert e^{\frac{\varphi}{h}}\psi_{0,0}%
^{h}(\cdot,k,\mathcal{V})\right\Vert _{L^{2}\left(  \left[  a,b\right]
\right)  }\leq C_{a,b,c}h^{-\frac{1}{2}}\,,
\]
where $\varphi\left(  \cdot\right)  =d_{Ag}\left(  \cdot,a,\mathcal{V}%
,k^{2}\right)  $ if $k>0$ and $\varphi\left(  \cdot\right)  =d_{Ag}\left(
\cdot,b,\mathcal{V},k^{2}\right)  $ if $k<0$.
\end{lemma}

\section{\label{App_jostest}The Jost's solutions in the $h$-dependent case}

In what follows: $\mathcal{C}_{x}^{n}(U)$ is the set of $\mathcal{C}^{n}%
$-continuous functions w.r.t. $x\in U\subseteq\mathbb{R}$, while
$\mathcal{H}_{z}(D)$ is the set of holomorphic functions w.r.t. $z\in
D\subseteq\mathbb{C}$. Let assume%
\begin{equation}
\mathcal{V}\in L^{\infty}\left(  \mathbb{R},\mathbb{R}\right)  \,,\quad
\text{supp }\mathcal{V}=\left[  a,b\right]  \,,\quad\inf\mathcal{V}>c\,,
\label{V_barrier}%
\end{equation}
for some positive $c$, and $\chi_{\pm}^{h}\left(  \cdot,\zeta,\mathcal{V}%
\right)  $ denote the solutions of the equation%
\begin{equation}
\left(  -h^{2}\partial_{x}^{2}+\mathcal{V}\right)  u=\zeta^{2}u\,,
\label{Jost_eq_h1}%
\end{equation}
fulfilling the conditions%
\begin{equation}
\left.  \chi_{+}^{h}\left(  \cdot,\zeta,\mathcal{V}\right)  \right\vert
_{x>b}=e^{i\frac{\zeta}{h}x}\,,\qquad\left.  \chi_{-}^{h}\left(  \cdot
,\zeta,\mathcal{V}\right)  \right\vert _{x<a}=e^{-i\frac{\zeta}{h}x}\,.
\label{Jost_sol_ext_h}%
\end{equation}
It is well known that these functions, usually referred to as the 'Jost's
solutions' of (\ref{Jost_eq_h1}), are $\mathcal{C}_{x}^{1}\left(
\mathbb{R},\,\mathcal{H}_{\zeta}\left(  \mathbb{C}^{+}\right)  \right)  $ for
any fixed $h>0$ and have continuous extensions w.r.t. $\zeta$ to the real axis
with the only possible exception of $k=0$, while, in the particular case of a
potential barrier, $\zeta\rightarrow\chi_{\pm}^{h}\left(  \cdot,\zeta
,\mathcal{V}\right)  $ extends to the whole real line including the origin
(for generic $L^{1}$-potentials we refer to \cite{Yafa}, while the case of a
barrier is explicitly considered in \cite{Man1} when $h=1$). Our purpose is to
study the small-$h$ beahviour of $\chi_{+}^{h}\left(  \cdot,k,\mathcal{V}%
\right)  $ under suitable restrictions on $k^{2}$ and $\mathcal{V}$. We next
introduce the notation
\begin{equation}
\Omega_{c}\left(  \mathcal{V}\right)  =\left\{  k\in\mathbb{R}\,\left\vert
\ \mathcal{V}-k^{2}>c\right.  \right\}  \,. \label{Omega_V}%
\end{equation}
According to the assumption (\ref{V_barrier}), $\Omega_{c}\left(
\mathcal{V}\right)  $ is a non-empty and bounded subset provided that $c>0$ is
small enough.

\begin{proposition}
\label{Proposition_h_bounds}Let $\mathcal{V}$ be defined by (\ref{V}),
$h\in\left(  0,h_{0}\right]  $ with $h_{0}$ is suitably small and assume
$k\in\Omega_{c}\left(  \mathcal{V}\right)  $ for some $c>0$. The relations%
\begin{equation}
\chi_{\pm}^{h}\left(  \cdot,k,\mathcal{V}\right)  =\mathcal{O}\left(
1\right)  \,,\qquad\partial_{1}\chi_{\pm}^{h}\left(  \cdot,k,\mathcal{V}%
\right)  =\mathcal{O}\left(  \frac{1}{h}\right)  \,, \label{small_h_bounds}%
\end{equation}
hold being the symbols $\mathcal{O}\left(  \cdot\right)  $ referred to the
metric space $\left.  \mathbb{R}\times\Omega_{c}\left(  \mathcal{V}\right)
\times\left(  0,h_{0}\right]  \right.  $.
\end{proposition}

\begin{proof}
To simplify the notations, the explicit dependence of $\chi_{\pm}^{h}$, from
the potential $\mathcal{V}$ is omitted. We start considering the Jost's
solutions $\chi_{\pm}^{h}\left(  \cdot,k\right)  $. Making use of the exterior
conditions (\ref{Jost_sol_ext_h}) and adapting the to this $h$-dependent
setting well known relations in 1D scattering theory (see the relations (1.6)
and (1.8) in the chp. 5 of \cite{Yafa}), we get%
\begin{align}
\left.  \chi_{+}^{h}\left(  \cdot,k\right)  \right\vert _{x<a}  &  =\frac
{h}{2ik}\left(  \left(  w_{0}^{h}(k)\right)  ^{\ast}e^{-i\frac{k}{h}x}%
-w^{h}(k)e^{i\frac{k}{h}x}\right)  \,,\label{Jost_sol_ext_h_plus}\\
& \nonumber\\
\left.  \chi_{-}^{h}\left(  \cdot,k\right)  \right\vert _{x>b}  &  =\frac
{h}{2ik}\left(  w_{0}^{h}(k)e^{i\frac{k}{h}x}-w^{h}(k)e^{-i\frac{k}{h}%
x}\right)  \,, \label{Jost_sol_ext_h_min}%
\end{align}
where $w^{h}$ and $w_{0}^{h}$ respectively denote the Wronskians associated to
the couples $\left\{  \chi_{+}^{h}\left(  \cdot,k\right)  ,\chi_{-}^{h}\left(
\cdot,k\right)  \right\}  $ and $\left\{  \chi_{+}^{h}\left(  \cdot,-k\right)
,\chi_{-}^{h}\left(  \cdot,k\right)  \right\}  $; these are defined according
to%
\begin{equation}
w\left(  f,g\right)  =fg^{\prime}-f^{\prime}g \label{wronskian}%
\end{equation}
As a consequence of (\ref{Jost_sol_ext_h}), (\ref{Jost_sol_ext_h_plus}) and
(\ref{Jost_sol_ext_h_min}), the functions $\chi_{\pm}^{h}\left(
\cdot,k\right)  $ are solutions of the problem (\ref{Ag_eq}) where $\zeta
=k\in\mathbb{R}$, while: $\gamma_{a}=-hw^{h}(k)e^{i\frac{k}{h}a}$ and
$\gamma_{b}=0$ in the case of $\chi_{+}^{h}\left(  \cdot,k\right)  $, or
$\gamma_{a}=0$ and $\gamma_{b}=hw^{h}(k)e^{-i\frac{k}{h}b}$ in the case of
$\chi_{-}^{h}\left(  \cdot,k\right)  $. Proceeding as in the proof of the
Lemma \ref{Lemma_Agmon} in the absence of the exponential weight (which
corresponds to take $\varphi_{h}=0$ and $v=\chi_{\pm}^{h}$ in
(\ref{Agmon_est1_0})), we obtain%
\begin{equation}
\left\Vert \chi_{\pm}^{h}\left(  \cdot,k\right)  \right\Vert _{H^{1,h}%
(a,b)}^{2}\lesssim h^{2}\left\vert w^{h}(k)\right\vert \,.
\label{Jost_sol_inside_h}%
\end{equation}
According to the definition of $w^{h}(k)$, the equivalent representations
hold
\begin{align}
w^{h}(k)  &  =e^{-i\frac{k}{h}a}\left(  -i\frac{k}{h}\chi_{+}^{h}\left(
a,k\right)  -\partial_{1}\chi_{+}^{h}\left(  a,k\right)  \right)
\,,\label{w_h_1}\\
& \nonumber\\
w^{h}(k)  &  =e^{i\frac{k}{h}b}\left(  \partial_{1}\chi_{-}^{h}\left(
b,k\right)  -i\frac{k}{h}\chi_{-}^{h}\left(  b,k\right)  \right)  \,.
\label{w_h_2}%
\end{align}
We use the $h$-dependent norms introduced in (\ref{h_norm}); from the
relations (\ref{w_h_1})-(\ref{w_h_2}) and the inequality (\ref{Agmon_est1_0}),
it follows%
\begin{equation}
\left\vert w^{h}(k)\right\vert \leq\frac{\left\vert k\right\vert }%
{h}\left\vert \chi_{+}^{h}\left(  a,k\right)  \right\vert +\left\vert
\partial_{1}\chi_{+}^{h}\left(  a,k\right)  \right\vert \lesssim
\frac{\left\vert k\right\vert }{h^{3/2}}\left\Vert \chi_{+}^{h}\left(
\cdot,k\right)  \right\Vert _{H^{1,h}(a,b)}+\frac{1}{h^{1/2}}\left\Vert
\partial_{1}\chi_{+}^{h}\left(  \cdot,k\right)  \right\Vert _{H^{1,h}(a,b)}\,,
\label{w_h_ineq}%
\end{equation}
and%
\begin{equation}
\left\vert w^{h}(k)\right\vert \leq\frac{\left\vert k\right\vert }%
{h}\left\vert \chi_{-}^{h}\left(  b,k\right)  \right\vert +\left\vert
\partial_{1}\chi_{-}^{h}\left(  b,k\right)  \right\vert \lesssim
\frac{\left\vert k\right\vert }{h^{3/2}}\left\Vert \chi_{-}^{h}\left(
\cdot,k\right)  \right\Vert _{H^{1,h}(a,b)}+\frac{1}{h^{1/2}}\left\Vert
\partial_{1}\chi_{-}^{h}\left(  \cdot,k\right)  \right\Vert _{H^{1,h}(a,b)}\,.
\end{equation}
Exploiting the equivalence of $\left\Vert u\right\Vert _{H^{1,h}(a,b)}$ with
$\left.  \left\Vert hu^{\prime}\right\Vert _{L^{2}(a,b)}+\left\Vert
u\right\Vert _{L^{2}(a,b)}\right.  $, and using the identity:\newline$\left.
h\partial_{1}^{2}\chi_{\pm}^{h}=\frac{1}{h}\left(  \mathcal{V}-k^{2}\right)
\chi_{\pm}^{h}\right.  $, we get
\[
\left\Vert \partial_{1}\chi_{\pm}^{h}\left(  \cdot,k\right)  \right\Vert
_{H^{1,h}(a,b)}\lesssim\left\Vert h\partial_{1}^{2}\chi_{\pm}^{h}\left(
\cdot,k\right)  \right\Vert _{L^{2}(a,b)}+\left\Vert \partial_{1}\chi_{\pm
}^{h}\left(  \cdot,k\right)  \right\Vert _{L^{2}(a,b)}\lesssim\frac{1}%
{h}\left\Vert \chi_{\pm}^{h}\left(  \cdot,k\right)  \right\Vert _{L^{2}%
(a,b)}+\left\Vert \partial_{1}\chi_{\pm}^{h}\left(  \cdot,k\right)
\right\Vert _{L^{2}(a,b)}\,,
\]
which yields%
\begin{equation}
\left\Vert \partial_{1}\chi_{\pm}^{h}\left(  \cdot,k\right)  \right\Vert
_{H^{1,h}(a,b)}\lesssim1/h\left\Vert \chi_{\pm}^{h}\left(  \cdot,k\right)
\right\Vert _{H^{1,h}(a,b)}\,. \label{Jost_sol_inside_h_est0}%
\end{equation}
Replacing this inequality at the r.h.s. of (\ref{w_h_ineq}) leads to%
\begin{equation}
\left\vert w^{h}(k)\right\vert \lesssim\frac{1}{h^{3/2}}\left(  1+\left\vert
k\right\vert \right)  \left\Vert \chi_{+}^{h}\left(  \cdot,k\right)
\right\Vert _{H^{1,h}(a,b)}\,, \label{w_h_ineq1}%
\end{equation}
while, using a similar\ estimate of $\left\vert w^{h}(k)\right\vert $ can be
given in terms of the norm $\left\Vert \chi_{-}^{h}\left(  \cdot,k\right)
\right\Vert _{H^{1,h}(a,b)}$. Then, using (\ref{Jost_sol_inside_h}) we get%
\begin{equation}
\left\Vert \chi_{\pm}^{h}\left(  \cdot,k\right)  \right\Vert _{H^{1,h}%
(a,b)}\lesssim h^{1/2}\left(  1+\left\vert k\right\vert \right)  \,,
\label{Jost_sol_inside_h_est1}%
\end{equation}
which, due to the inequality (\ref{h_GN}), entails: $\left.  1_{\left[
a,b\right]  }\left\vert \chi_{\pm}^{h}\left(  \cdot,k\right)  \right\vert
\lesssim\left(  1+\left\vert k\right\vert \right)  \right.  $, while, from
(\ref{Jost_sol_inside_h_est0}), follows%
\begin{equation}
1_{\left[  a,b\right]  }\left\vert \partial_{1}\chi_{\pm}^{h}\left(
\cdot,k\right)  \right\vert \lesssim h^{-1}\left(  1+\left\vert k\right\vert
\right)  \,.
\end{equation}
Since $\Omega_{c}\left(  \mathcal{V}\right)  $ is a bounded set, it results%
\begin{equation}
1_{\left[  a,b\right]  }\chi_{\pm}^{h}\left(  \cdot,k\right)  =\mathcal{O}%
\left(  1\right)  \,,\qquad1_{\left[  a,b\right]  }\left\vert \partial_{1}%
\chi_{\pm}^{h}\left(  \cdot,k\right)  \right\vert =\mathcal{O}\left(  \frac
{1}{h}\right)  \,. \label{Jost_sol_inside_h_est}%
\end{equation}
In the exterior domain, the solutions are described by the relations
(\ref{Jost_sol_ext_h}) and (\ref{Jost_sol_ext_h_plus}%
)-(\ref{Jost_sol_ext_h_min}). Then, to extend the above inequalities to the
whole real axis, estimates for $\chi_{+}^{h}$ and $\chi_{-}^{h}$ in $\left(
-\infty,a\right)  $ and $\left(  b,+\infty\right)  $ respectively are needed.
According to (\ref{w_h_ineq1}) and (\ref{Jost_sol_inside_h_est1}), it follows
that: $w^{h}(k)=\mathcal{O}\left(  \frac{1}{h}\right)  $ uniformly w.r.t.
$k\in\Omega_{c}\left(  \mathcal{V}\right)  $. Moreover, the representation%
\begin{equation}
w_{0}^{h}(k)=\chi_{+}^{h}\left(  a,-k\right)  \partial_{1}\chi_{-}^{h}\left(
a,k\right)  -\partial_{1}\chi_{+}^{h}\left(  a,-k\right)  \chi_{-}^{h}\left(
a,k\right)  \,,
\end{equation}
and the previous estimates yield: $w_{0}^{h}(k)=\mathcal{O}\left(  \frac{1}%
{h}\right)  $. As recalled above, for any $x\in\mathbb{R}$, the maps:
$k\rightarrow\chi_{\pm}^{h}\left(  x,k\right)  $, $w^{h}(k)$ and $w_{0}%
^{h}(k)$ are continuous including $k=0$. This implies that $w^{h}$ and
$w_{0}^{h}$ behaves in $k=0$ as%
\begin{equation}
w_{0}^{h}(k)e^{i\frac{k}{h}x}=a_{0}(h)+\mathcal{O}\left(  \frac{k}{h}\right)
\,,\qquad w^{h}(k)e^{-i\frac{k}{h}x}=a_{0}(h)+\mathcal{O}\left(  \frac{k}%
{h}\right)  \,,
\end{equation}
where $a_{0}(h)=\mathcal{O}\left(  \frac{1}{h}\right)  \in\mathbb{R}$.
Therefore, the r.h.s. of (\ref{Jost_sol_ext_h_plus}) and
(\ref{Jost_sol_ext_h_min}) result uniformly bounded as $h>0$, $k\in\Omega
_{c}\left(  \mathcal{V}\right)  $ and $x\in\mathbb{R}$, while their
derivatives w.r.t. $x$ behaves as $\mathcal{O}\left(  \frac{k}{h}\right)  $.
\end{proof}

The Green's functions $\mathcal{G}^{\zeta^{2},h}\left(  \cdot,y,\mathcal{V}%
\right)  $ and $\mathcal{H}^{\zeta^{2},h}\left(  \cdot,y,\mathcal{V}\right)
$, introduced as solutions of the equations (\ref{Green_eq})-(\ref{D_Green_eq}%
), are related to the Jost's solutions of (\ref{Jost_eq_h1}) according to
\begin{align}
\mathcal{G}^{\zeta^{2},h}\left(  \cdot,y,\mathcal{V}\right)   &  =\frac
{1}{h^{2}w\left(  \chi_{+}^{h}\left(  \cdot,\zeta,V\right)  ,\chi_{-}%
^{h}\left(  \cdot,\zeta,V\right)  \right)  }\mathcal{\,}\left\{
\begin{array}
[c]{c}%
\chi_{+}^{h}\left(  \cdot,\zeta,\mathcal{V}\right)  \chi_{-}^{h}\left(
y,\zeta,\mathcal{V}\right)  \,,\qquad x\geq y\,,\\
\\
\chi_{-}^{h}\left(  \cdot,\zeta,\mathcal{V}\right)  \chi_{+}^{h}\left(
y,\zeta,\mathcal{V}\right)  \,,\qquad x<y\,,
\end{array}
\right. \label{G_z_h}\\
& \nonumber\\
\mathcal{H}^{\zeta^{2},h}\left(  \cdot,y,\mathcal{V}\right)   &  =\frac
{1}{h^{2}w\left(  \chi_{+}^{h}\left(  \cdot,\zeta,V\right)  ,\chi_{-}%
^{h}\left(  \cdot,\zeta,V\right)  \right)  }\mathcal{\,}\left\{
\begin{array}
[c]{c}%
\chi_{+}^{h}\left(  \cdot,\zeta,\mathcal{V}\right)  \,\partial_{1}\chi_{-}%
^{h}\left(  y,\zeta,\mathcal{V}\right)  \,,\qquad x\geq y\,,\\
\\
\chi_{-}^{h}\left(  \cdot,\zeta,\mathcal{V}\right)  \,\partial_{1}\chi_{+}%
^{h}\left(  y,\zeta,\mathcal{V}\right)  \,,\qquad x<y\,.
\end{array}
\right.  \label{H_z_h}%
\end{align}
where $\zeta\in\mathbb{C}^{+}$. Adopting the Wronskian's notation
$w^{h}(k,\mathcal{V})$ introduced in the proof of the Proposition
\ref{Proposition_h_bounds} and the definition given in (\ref{GH_k}), the
limits as $\zeta\rightarrow k$ represent as%
\begin{align}
G^{k,h}\left(  \cdot,y,\mathcal{V}\right)   &  =\frac{1}{h^{2}w^{h}%
(k,\mathcal{V})}\mathcal{\,}\left\{
\begin{array}
[c]{c}%
\chi_{+}^{h}\left(  \cdot,k,\mathcal{V}\right)  \chi_{-}^{h}\left(
y,k,\mathcal{V}\right)  \,,\ x\geq y\,,\\
\\
\chi_{-}^{h}\left(  \cdot,k,\mathcal{V}\right)  \chi_{+}^{h}\left(
y,k,\mathcal{V}\right)  \,,\ x<y\,,
\end{array}
\right. \label{G_k_jost}\\
& \nonumber\\
H^{k,h}\left(  \cdot,y,\mathcal{V}\right)   &  =\frac{-1}{h^{2}w^{h}%
(k,\mathcal{V})}\mathcal{\,}\left\{
\begin{array}
[c]{c}%
\chi_{+}^{h}\left(  \cdot,k,\mathcal{V}\right)  \,\partial_{1}\chi_{-}%
^{h}\left(  y,k,\mathcal{V}\right)  \,,\ x\geq y\,,\\
\\
\chi_{-}^{h}\left(  \cdot,k,\mathcal{V}\right)  \,\partial_{1}\chi_{+}%
^{h}\left(  y,k,\mathcal{V}\right)  \,,\ x<y\,,
\end{array}
\right.  \label{H_k_jost}%
\end{align}
while the corresponding formula for the generalized eigenfunctions is%
\begin{equation}
\psi_{0,0}^{h}(x,k,\mathcal{V})=\left\{
\begin{array}
[c]{lll}%
-\frac{2ik}{hw^{h}(k,\mathcal{V})}\chi_{+}^{h}(x,k,\mathcal{V})\,, &  &
\text{for }k\geq0\,,\\
&  & \\
\frac{2ik}{hw^{h}(-k,\mathcal{V})}\chi_{-}^{h}(x,-k,\mathcal{V})\,, &  &
\text{for }k<0\,.
\end{array}
\right.  \label{gen_eigenfun_h_jost}%
\end{equation}

Next we focus on the case where the potential, depending on $h$, is defined
according to the conditions (\ref{V_h})-(\ref{V_h1}): namely it is assumed
that $\mathcal{V}^{h}=V+W^{h}$ is formed by the superposition of a potential
barrier, $V$, plus a selfadjoint term, $W^{h}$, supported inside $\left(
a,b\right)  $ with support of size $h$. The functions $\chi_{\pm}^{h}\left(
x,k,\mathcal{V}^{h}\right)  $ are defined by the equations (\ref{Jost_eq_h1})
and the exterior conditions (\ref{Jost_sol_ext_h}). In particular,
(\ref{Jost_eq_h1}) rephrases in terms of the integral equation%
\begin{equation}
\chi_{\pm}^{h}\left(  x,k,\mathcal{V}^{h}\right)  =\chi_{\pm}^{h}\left(
x,k,V\right)  -\frac{1}{h^{2}}%
{\displaystyle\int_{x}^{x_{\pm}}}
\mathcal{K}^{h}\left(  t,x,k,V\right)  \,W^{h}(t)\chi_{\pm}^{h}\left(
t,k,\mathcal{V}^{h}\right)  \,dt\,, \label{Jost_eq_h_local}%
\end{equation}
where the boundary points are fixed by: $x_{+}=b$ and $x_{-}=a$, while
$x<x_{+}$ or $x>x_{-}$ depending on the case which is being considered. The
kernel $\mathcal{K}^{h}$, depending on the 'unperturbed' Jost's functions
(i.e. those defined by $W^{h}=0$), is defined according to%
\begin{equation}
\mathcal{K}^{h}\left(  t,x,k,V\right)  =\frac{\chi_{+}^{h}\left(
t,k,V\right)  \chi_{-}^{h}\left(  x,k,V\right)  -\chi_{-}^{h}\left(
t,k,V\right)  \chi_{+}^{h}\left(  x,k,V\right)  }{w^{h}(k,V)}\,.
\label{Kernel_h}%
\end{equation}
Exploiting the result of the Proposition \ref{Proposition_h_bounds}, we obtain
the following characterization.

\begin{lemma}
Let $\mathcal{V}$, $h,c$ and $k$ be defined according to the assumptions of
the Proposition \ref{Proposition_h_bounds}. Then the relations%
\begin{equation}
\mathcal{K}^{h}\left(  t,x,k,\mathcal{V}\right)  =\mathcal{O}\left(  h\right)
\,,\qquad\partial_{x}\mathcal{K}^{h}\left(  t,x,k,\mathcal{V}\right)
=\mathcal{O}\left(  1\right)  \,\,, \label{Kernel_h_est}%
\end{equation}
hold being the symbols $\mathcal{O}\left(  \cdot\right)  $ referred to the
metric space $\left.  \mathbb{R}^{2}\times\Omega_{c}\left(  \mathcal{V}%
\right)  \times\left(  0,h_{0}\right]  \right.  $.
\end{lemma}

\begin{proof}
[Sketch of the proof]Rephrasing the relation (1.9) in chp. 5 of \cite{Yafa} in
the $h$-dependent case, we get%
\begin{equation}
\left\vert w^{h}(k,\mathcal{V})\right\vert ^{2}=\frac{k^{2}}{h^{2}}+\left\vert
w_{0}^{h}(k,\mathcal{V})\right\vert ^{2}\,,
\end{equation}
(see the Wronskian's notation $w^{h}$, $w_{0}^{h}$ introduced in the
Proposition \ref{Proposition_h_bounds}). This entails: $\left\vert
w^{h}(k,\mathcal{V})\right\vert ^{-1}\leq\frac{h}{\left\vert k\right\vert }$
and the terms $\frac{\pm2ik}{hw^{h}(\mp k,\mathcal{V})}$ are bounded uniformly
w.r.t. $h>0$ and $k\in\mathbb{R}$, while $\frac{1}{h^{2}w^{h}(k,\mathcal{V}%
)}=\mathcal{O}\left(  \frac{1}{hk}\right)  $. Thus, according to the
representations (\ref{G_k_jost})-(\ref{gen_eigenfun_h_jost}), the relations%
\begin{equation}%
\begin{array}
[c]{ccccc}%
G^{k,h}\left(  \cdot,y,\mathcal{V}\right)  =\mathcal{O}\left(  \frac{1}%
{h}\right)  \,, &  & H^{k,h}\left(  \cdot,y,\mathcal{V}\right)  =\mathcal{O}%
\left(  \frac{1}{h^{2}}\right)  \,, &  & \partial_{1}^{j}\psi_{-}^{h}%
(\cdot,k,\mathcal{V})=\mathcal{O}\left(  \frac{1}{h^{j}}\right)  \,,\ j=0,1\,,
\end{array}
\label{Green_h_bounds}%
\end{equation}
are straightforward consequences of the characterization of $\chi_{\pm}^{h}$
obtained above. These allow to describe the beahviour of $\mathcal{K}^{h}$ and
$\partial_{x}\mathcal{K}^{h}$ as $h\rightarrow0$. Let us explicitly consider
the case of $\mathcal{K}^{h}\left(  t,x,k,\mathcal{V}\right)  $ when $t\geq
x$. Using the representations (\ref{G_k_jost})-(\ref{gen_eigenfun_h_jost}) and
the relations: $\left(  \chi_{\pm}^{h}\left(  \cdot,k,\mathcal{V}\right)
\right)  ^{\ast}=\chi_{\pm}^{h}\left(  \cdot,-k,\mathcal{V}\right)  $, the
identity (\ref{Kernel_h}) can be rephrased in terms of $G^{k,h}$ and $\psi
_{-}^{h}$ as
\begin{equation}
1_{\left\{  t\geq x\right\}  }\left(  t\right)  \mathcal{K}^{h}\left(
t,x,k,\mathcal{V}\right)  =\left\{
\begin{array}
[c]{lll}%
h^{2}G^{k,h}\left(  t,x,\mathcal{V}\right)  +\frac{h}{2ik}\psi_{-}^{h}\left(
x,k,\mathcal{V}\right)  \chi_{-}^{h}\left(  t,k,\mathcal{V}\right)  \,, &  &
\text{for }k\geq0\,,\\
&  & \\
h^{2}G^{k,h}\left(  t,x,\mathcal{V}\right)  +\frac{h}{2ik}\left(  \psi_{-}%
^{h}\left(  t,k,\mathcal{V}\right)  \right)  ^{\ast}\chi_{+}^{h}\left(
x,k,\mathcal{V}\right)  \,, &  & \text{for }k<0\,.
\end{array}
\right.  \label{Kernel_h_1_1}%
\end{equation}
Then, (\ref{Kernel_h_est}) follows from (\ref{Green_h_bounds}) and the result
of the Proposition \ref{Proposition_h_bounds}. The other cases can be analyzed
following the same line.
\end{proof}

This framework allows to discuss the properties of $\chi_{\pm}^{h}\left(
x,k,\mathcal{V}^{h}\right)  $ by making use of the results obtained in the
case when $W^{h}=0$.

\begin{proposition}
\label{Proposition_Jost_h_local}Let $\mathcal{V}^{h}=V+W^{h}$ be defined
according to the conditions (\ref{V_h})-(\ref{V_h1}) being $h\in\left(
0,h_{0}\right]  $ with $h_{0}$ suitably small and assume $k\in\Omega
_{c}\left(  V\right)  $ for some $c>0$. The relations%
\begin{equation}
\chi_{\pm}^{h}\left(  \cdot,k,\mathcal{V}^{h}\right)  =\mathcal{O}\left(
1\right)  \,,\qquad\partial_{1}\chi_{\pm}^{h}\left(  \cdot,k,\mathcal{V}%
^{h}\right)  =\mathcal{O}\left(  \frac{1}{h}\right)  \,, \label{Jost_bound_h}%
\end{equation}
hold with $\mathcal{O(\cdot)}$ referred to the metric space $\mathbb{R}%
\times\Omega_{c}\left(  V\right)  \times\left(  0,h_{0}\right]  $.
\end{proposition}

\begin{proof}
We explicitly consider the case $\chi_{+}^{h}\left(  \cdot,k,\mathcal{V}%
^{h}\right)  $. The proof, in the case of the function $\chi_{-}^{h}%
$\thinspace, is obtained following the same line. For $x<b$, $\chi_{+}%
^{h}\left(  \cdot,k,\mathcal{V}^{h}\right)  $ solves the equation
(\ref{Jost_eq_h_local}) with $x_{+}=b$. We look for the solution solution of
this problem in the form of a Picard's series: $\chi_{+}^{h}\left(
\cdot,k,\mathcal{V}^{h}\right)  =\sum_{n=0}^{+\infty}\chi_{+,n}^{h}\left(
\cdot,k,\mathcal{V}^{h}\right)  $%
\begin{equation}
\chi_{+,0}^{h}\left(  \cdot,k,\mathcal{V}^{h}\right)  =\chi_{+}^{h}\left(
\cdot,k,V\right)  \,,\quad\chi_{+,n}^{h}\left(  x,k,\mathcal{V}^{h}\right)
=-\frac{1}{h^{2}}%
{\displaystyle\int\limits_{x}^{b}}
\mathcal{K}^{h}\left(  t,x,k,V\right)  \,W^{h}(t)\chi_{+,n-1}^{h}\left(
t,k,\mathcal{V}^{h}\right)  \,dt\,\,. \label{Picard_h}%
\end{equation}
Due to our assumptions and to results of the Proposition
\ref{Proposition_h_bounds}, the first term of this expansion is continuous in
$x$ and $k$, and bounded according to the relations (\ref{small_h_bounds}); in
what follows we set%
\begin{equation}
M\left(  c,V\right)  =\sup_{x\in\mathbb{R}\,,\ k\in\Omega_{c}\left(  V\right)
\,,\ h\in\left(  0,h_{0}\right]  }\left\vert \chi_{+,0}^{h}\left(
\cdot,k,\mathcal{V}^{h}\right)  \right\vert \,. \label{M_def}%
\end{equation}
The second contribution is given by%
\begin{equation}
\chi_{+,1}^{h}\left(  x,k,\mathcal{V}^{h}\right)  =-\frac{1}{h^{2}}%
{\displaystyle\int\limits_{x}^{b}}
\mathcal{K}^{h}\left(  t,x,k,V\right)  \,W^{h}(t)\chi_{+,0}^{h}\left(
t,k,\mathcal{V}^{h}\right)  \,dt\,\,.
\end{equation}
As it follows from the regularity of the kernel $\mathcal{K}^{h}$ and
$\chi_{+,0}^{h}$, this is a continuous function w.r.t. $x$ and $k$, while from
the relation (\ref{Kernel_h_est}), it results%
\begin{equation}
\left\vert \chi_{+,1}^{h}\left(  x,k,\mathcal{V}^{h}\right)  \right\vert \leq
M\left(  c,V\right)  F\left(  x,c,\mathcal{V}^{h}\right)  \,,
\label{induction_0}%
\end{equation}
where%
\begin{equation}
F\left(  x,h,c,\mathcal{V}^{h}\right)  =\frac{C\left(  c,V\right)  }{h}%
{\displaystyle\int\limits_{x}^{b}}
\left\vert W^{h}(t)\right\vert \,dt\,,\qquad C\left(  c,V\right)  =\frac{1}%
{h}\sup_{\substack{t,x\in\mathbb{R}\\k\in\Omega_{c}\left(  V\right)
}}\left\vert \mathcal{K}^{h}\left(  t,x,k,V\right)  \right\vert \,.
\label{F_def}%
\end{equation}
Next, assume $\chi_{+,n-1}^{h}\left(  x,k,\mathcal{V}^{h}\right)  $ to be
continuous w.r.t. $x$ and $k$ fulfilling the inequality%
\begin{equation}
\left\vert \chi_{+,n-1}^{h}\left(  x,k,\mathcal{V}^{h}\right)  \right\vert
\leq M\left(  c,V\right)  \frac{F^{n-1}\left(  x,h,c,\mathcal{V}^{h}\right)
}{\left(  n-1\right)  !}\,, \label{induction_n}%
\end{equation}
and consider the term $\chi_{+,n}^{h}\left(  x,k,\mathcal{V}^{h}\right)  $;
due to the equation (\ref{Picard_h}), this still continuous w.r.t. $x$ and $k$
allowing the estimate%
\begin{align*}
\left\vert \chi_{+,n}^{h}\left(  x,k,\mathcal{V}^{h}\right)  \right\vert  &
\leq\frac{1}{h^{2}}%
{\displaystyle\int\limits_{x}^{b}}
\left\vert \mathcal{K}^{h}\left(  t,x,k,V\right)  \,W^{h}(t)\chi_{+,n-1}%
^{h}\left(  t,k,\mathcal{V}^{h}\right)  \right\vert \,dt\,\\
& \\
&  \leq\frac{M\left(  c,V\right)  C\left(  c,V\right)  }{h}%
{\displaystyle\int\limits_{x}^{b}}
\left\vert W^{h}(t)\right\vert \frac{F^{n-1}\left(  t,h,c,\mathcal{V}%
^{h}\right)  }{\left(  n-1\right)  !}\,dt\,=-\frac{M\left(  c,V\right)  }{n!}%
{\displaystyle\int\limits_{x}^{b}}
\partial_{t}F^{n}\left(  t,h,c,\mathcal{V}^{h}\right)  \,dt\\
& \\
&  =\frac{M\left(  c,V\right)  }{n!}F^{n}\left(  x,h,c,\mathcal{V}^{h}\right)
\end{align*}
Then, an induction argument shows that the Picard's series uniformly converges
to $\chi_{+}^{h}$ and%
\begin{equation}
\sup_{x<b\,,\ k\in\Omega_{c}\left(  V\right)  }\left\vert \chi_{+}^{h}\left(
x,k,\mathcal{V}^{h}\right)  \right\vert \leq M\left(  c,\mathcal{V}\right)
e^{\left\Vert F\left(  \cdot,h,c,\mathcal{V}^{h}\right)  \right\Vert
_{L^{\infty}(a,b)}}\,. \label{induction_est}%
\end{equation}
Since%
\begin{equation}
\left\Vert F\left(  \cdot,h,c,\mathcal{V}^{h}\right)  \right\Vert _{L^{\infty
}(a,b)}\lesssim\frac{1}{h}\left\Vert W^{h}\right\Vert _{L^{1}(a,b)}%
\lesssim1\,, \label{F_est}%
\end{equation}
we obtain%
\begin{equation}
\sup_{x<b\,,\ k\in\Omega_{c}\left(  V\right)  }\left\vert \chi_{+}^{h}\left(
x,k,\mathcal{V}^{h}\right)  \right\vert \lesssim1\,\Rightarrow1_{\left\{
x<b\right\}  }(x)1_{\Omega_{c}\left(  V\right)  }\left(  k\right)  \chi
_{+}^{h}\left(  x,k,\mathcal{V}^{h}\right)  =\mathcal{O}\left(  1\right)  \,.
\label{Jost_est_1_h}%
\end{equation}
The first relation in (\ref{Jost_bound_h}) is a consequence of
(\ref{Jost_est_1_h}) and of the exterior condition $\left.  \chi_{+}%
^{h}\left(  \cdot,k,\mathcal{V}^{h}\right)  \right\vert _{x>b}=e^{i\frac{k}%
{h}x}$.

Next consider $\partial_{1}\chi_{+}^{h}$; for $x\geq b$ this is explicitly
defined by $i\frac{k}{h}e^{i\frac{k}{h}x}$, while, for $x<b$, it fulfills the
equation%
\begin{equation}
\partial_{x}\chi_{+}^{h}\left(  x,k,\mathcal{V}^{h}\right)  =\partial_{x}%
\chi_{+}^{h}\left(  x,k,V\right)  -\frac{1}{h^{2}}%
{\displaystyle\int\limits_{x}^{b}}
\partial_{x}\mathcal{K}^{h}\left(  t,x,k,V\right)  W^{h}(t)\chi_{+}^{h}\left(
t,k,\mathcal{V}^{h}\right)  \,dt\,. \label{Jost_eq_h_local1}%
\end{equation}
Being $\chi_{+}^{h}$ and $\partial_{x}\mathcal{K}^{h}$ uniformly bounded for
$k\in\Omega_{c}\left(  V\right)  $, the integral part at the r.h.s. of
(\ref{Jost_eq_h_local1}) results dominated $\mathcal{O}\left(  1/h\right)  $
and, according to the result of Proposition \ref{Proposition_h_bounds}, the
same holds for the function $\partial_{x}\chi_{+}^{h}\left(  x,k,V\right)  $.
It follows%
\begin{equation}
1_{\left\{  x<b\right\}  }(x)1_{\Omega_{c}\left(  V\right)  }\left(  k\right)
\partial_{x}\chi_{+}^{h}\left(  x,k,\mathcal{V}^{h}\right)  =\mathcal{O}%
\left(  \frac{1}{h}\right)  \,.
\end{equation}
This relations and the explicit form $\partial_{x}\chi_{+}^{h}=i\frac{k}%
{h}e^{i\frac{k}{h}x}$, holding for $x>b$, yields the second relation in
(\ref{Jost_bound_h}).
\end{proof}


\bigskip

\begin{thebibliography}{99}                                                                                               %


\bibitem {Agm}S. Agmon. \emph{Lectures on exponential decay of solutions of
second-order elliptic equations: bounds on eigenfunctions of }$\emph{N}%
$\emph{-body Schr\"{o}dinger operators.} Volume \textbf{29} of
\emph{Mathematical Notes}, Princeton University Press, Princeton, NJ, 1982.

\bibitem {AgCo}J. Aguilar, J.M. Combes. A class of analytic perturbations for
one-body Schr\"{o}dinger Hamiltonians. \emph{Comm. Math. Phys.}, \textbf{22},
269--279, 1971.

\bibitem {BaCo}E. Balslev, J.M. Combes. Spectral properties of many-body
Schr\"{o}dinger operators with dilatation-analytic interactions. \emph{Comm.
Math. Phys.}, \textbf{22}, 280--294, 1971.

\bibitem {BeMaNe}J. Behrndt, M.M. Malamud, H. Neidhardt. Scattering matrices
and Weyl functions. \emph{Proc. Lond. Math. Soc.} \textbf{97, }no.3, 568--598, 2008.

\bibitem {BNP1}V. Bonnaillie-No\"{e}l, F. Nier, Y. Patel. Far from equlibrium
steady states of 1D-Schr\"{o}dinger-Poisson systems with quantum wells I.
\emph{Ann. I.H.P. An. Non Lin\'{e}aire,} \textbf{25}, 937-968, 2008.

\bibitem {BNP2}V. Bonnaillie-No\"{e}l, F. Nier, Y. Patel. Far from equlibrium
steady states of 1D-Schr\"{o}dinger-Poisson systems with quantum wells II.
J\emph{. Math. Soc. of Japan.}, \textbf{61}, 65-106, 2009.

\bibitem {BNP3}V. Bonnaillie-No\"{e}l, F. Nier, Y. Patel. Computing the steady
states for an asymptotic model of quantum transport in resonant
heterostructures. \emph{Journal of Computational Physics}, \textbf{219}(2),
644-670, 2006.

\bibitem {FMN2}A. Faraj, A. Mantile, F. Nier. Adiabatic evolution of 1D shape
resonances: an artificial interface conditions approach. \emph{M3AS},
\textbf{21} no. 3, 541-618, 2011.

\bibitem {FMN3}A. Faraj, A. Mantile, F. Nier. An explicit model for the
adiabatic evolution of quantum observables driven by 1D shape resonances.
\emph{J. Phys. A: Math. Theor.} \textbf{43,} 2010.

\bibitem {Hel}B. Helffer. \emph{Semiclassical analysis for the Schr\"{o}dinger
operator and applications}, volume 1336 of \emph{Lecture Notes in Mathematics.
}Springer-Verlag, Berlin, 1988.

\bibitem {HeSj1}B. Helffer, J. Sj\"{o}strand. \emph{R\'{e}sonances en limite
semi-classique. M\'{e}m. Soc. Mat. France (N.S.), }number 24-25, 1986.

\bibitem {JoPrSj}G. Jona-Lasionio, C. Presilla, J. Sj\"{o}strand. On the
Schr\"{o}dinger equation with concentrated non linearities. \emph{Ann.
Physics, }\textbf{240} no.1, 1-21,1995.

\bibitem {Kato}T. Kato. Wave Operators and Similarity for Some Non-Selfadjoint
Operators. \emph{Math. Annalen, }162, 258-279, 1966.

\bibitem {Kato1}T. Kato. Linear evolution equations of "hyperbolic" type.
\emph{J. Fac. Sci. Univ. Tokyo Sect. I}, \textbf{17,} 241-258, 1970.

\bibitem {Man1}A. Mantile. Wave operators, similarity and dynamics for a class
of Schr\"{o}dinger operators with generic non-mixed interface conditions in
1D. \emph{To appear on JMP 2013.}

\bibitem {Nenciu}G. Nenciu, Linear adiabatic theory. Exponential estimates.
\emph{Comm. Math.Phys.,}\textbf{152}(3); 479-496, 1993.

\bibitem {PrSj}C. Presilla, J. Sj\"{o}strand. Transport properties in resonant
tunneling heterostructures. \emph{J. Math. Phys., }\textbf{37}(10), 4816-4844, 1996.

\bibitem {PrSj1}C. Presilla, J. Sj\"{o}strand. Nonlinear resonant tunneling in
systems coupled to quantum reservoirs. \emph{Phys. Rev. B: Condensed matter},
\textbf{55} no15, 9310-9313, 1997.

\bibitem {Ryz1}V.A. Ryzhov, Functional model of a class of non-selfadjoint
extensions of symmetric operators. \emph{Oper. Theory Adv. Appl.}, vol. 174,
117-158, 2007.

\bibitem {Wolf}F. Wolf, On the essential spectrum of partial differential
boundary problems. \emph{Comm. Pure Appl. Math.} \textbf{12} , 211--228, 1959.

\bibitem {Yafa}D.R. Yafaev. \emph{Mathematical Scattering Theory: Analytic
theory. }Mathematical Surveys and Monographs vol. 158, American Mathematical
Society, Providence, 2010.

\bibitem {Yosh}K. Yoshida. \emph{Functional Analysis - Second Edition}.
Springer - Verlag, Berlin, 1968.
\end{thebibliography}
\end{document}